\documentclass[a4paper,UKenglish,cleveref, autoref, thm-restate]{lipics-v2021}
\nolinenumbers

\usepackage{subcaption} 
\usepackage{xspace}
\usepackage{amsmath}
\usepackage{xfrac}
\usepackage{todonotes}
\usepackage{mathtools}
\usepackage{booktabs}
\usepackage{mdframed}
\usepackage{niceAlgorithm}
\usepackage{graphicx}
\usepackage{longtable}
\newtheorem{problem_description}{Problem}
\usepackage{caption}

\newcommand{\LL}{\ell} 
\newcommand{\A}{A} 
\newcommand{\PP}{P} 
\newcommand{\conv}{\text{conv}}
\newcommand{\bigO}[1]{\ensuremath{\mathcal{O}}{\left(#1\right)}\xspace}
\newcommand{\bigtO}[1]{\ensuremath{\mathcal{\tilde{O}}}{\left(#1\right)}\xspace}

\newcommand{\biProblem}{\ensuremath{P_F}\xspace}

\newcommand{\freeProblem}{\ensuremath{P^\alpha_\text{F}}\xspace}

\newcommand{\problem}[1]{\ensuremath{P^\alpha_{#1}}\xspace}
\newcommand{\solution}{\ensuremath{\mathcal{S}}\xspace}
\newcommand{\polygons}{\ensuremath{\mathcal{B}}\xspace}
\newcommand{\subdivision}{\ensuremath{\mathcal{D}}\xspace}
\newcommand{\region}{\ensuremath{G}\xspace}
\newcommand{\enclosed}[2]{\ensuremath{\region(#1,#2)}\xspace}
\newcommand{\sarc}[2]{C_{#1}^{#2}}
\newcommand{\allowedArcs}{\ensuremath{F_\alpha(\polygons)}\xspace}
\newcommand{\Per}{\operatorname{Per}}

\newcommand{\combinatorialsol}{\ensuremath{\mathcal{K}}\xspace}
\newcommand{\combinatorialsolevaluated}[1]{\ensuremath{\mathcal{K}[{#1}]}\xspace}

\newcommand{\validSet}[1]{\ensuremath{\mathcal{I}({#1})}\xspace}
\newcommand{\possibleArcs}{\ensuremath{\hat{F}_\alpha(\polygons)}\xspace}
\newcommand{\wind}{\ensuremath{\omega}\xspace}
\newcommand{\tangentangle}[2]{\ensuremath{\sphericalangle(#1,#2)}\xspace}

\newcommand{\ignoreset}[1]{\ensuremath{\Phi({#1})}\xspace}

\newcommand{\forwardRay}[2]{\ensuremath{\vec{r}_{#1}(#2)}\xspace}

\newcommand{\cev}[1]{\reflectbox{\ensuremath{\vec{\reflectbox{\ensuremath{#1}}}}}}
\newcommand{\backwardRay}[2]{\ensuremath{\cev{r}_{#1}(#2)}\xspace}
  
\newcommand{\vecTwo}[2]{\ensuremath{\begin{pmatrix}#1\\#2\end{pmatrix}}}

\SetKwFunction{Recurse}{Recurse}
\SetKwFunction{makeSet}{makeSet}
\SetKwFunction{unionFn}{union}
\SetKwFunction{find}{find}
\SetKwFunction{GetPoly}{GetPoly}
\SetKwFunction{ComputeUsefulArcs}{CompUArcs}
\SetKwFunction{optimizer}{UPA-Opt}
\SetKw{KwAnd}{and}
\SetKw{KwOr}{or}
\SetKwFunction{union}{union}

\def\CC{{C\nolinebreak[4]\hspace{-.05em}\raisebox{.4ex}{\tiny\textbf{++}}}}

\newcommand{\amin}{\ensuremath{\alpha_\text{min}}\xspace}
\newcommand{\amax}{\ensuremath{\alpha_\text{max}}\xspace}
\newcommand{\anew}{\ensuremath{\alpha_N}\xspace}
\newcommand{\alowsup}{\ensuremath{\alpha^\text{sup}_L}\xspace}
\newcommand{\alowinf}{\ensuremath{\alpha^\text{inf}_L}\xspace}
\newcommand{\aupsup}{\ensuremath{\alpha^\text{sup}_U}\xspace}
\newcommand{\aupinf}{\ensuremath{\alpha^\text{inf}_U}\xspace}
\newcommand{\anewsup}{\ensuremath{\alpha^\text{sup}_N}\xspace}
\newcommand{\anewinf}{\ensuremath{\alpha^\text{inf}_N}\xspace}

\newcommand{\lmin}{\ensuremath{\lambda_\text{min}}\xspace}
\newcommand{\lmax}{\ensuremath{\lambda_\text{max}}\xspace}
\newcommand{\lnew}{\ensuremath{\lambda_N}\xspace}

\newcommand{\comblow}{\ensuremath{\combinatorialsol_L}\xspace}
\newcommand{\combnew}{\ensuremath{\combinatorialsol_N}\xspace}
\newcommand{\combup}{\ensuremath{\combinatorialsol_U}\xspace}
\newcommand{\sollow}{\ensuremath{\solution_L}\xspace}
\newcommand{\solnew}{\ensuremath{\solution_N}\xspace}
\newcommand{\solup}{\ensuremath{\solution_U}\xspace}

\newcommand{\done}[1]{#1}

\title{Bicriteria Polygon Aggregation with Arbitrary Shapes} 

\author{Lotte Blank}{University of Bonn, Germany}{lblank@uni-bonn.de}{https://orcid.org/0000-0002-6410-8323}{}

\author{David Eppstein}{University of California, Irvine \and\url{https://www.ics.uci.edu/~eppstein/}}{eppstein@uci.edu}{}{Research supported in part by NSF grant CCF-2212129.}

\author{Jan-Henrik Haunert}{University of Bonn, Germany}{haunert@igg.uni-bonn.de}{https://orcid.org/0000-0001-8005-943X}{}

\author{Herman Haverkort}{University of Bonn, Germany}{haverkort@uni-bonn.de}{}{}

\author{Benedikt Kolbe}{University of Bonn, Hausdorff Center for Mathematics, Lamarr Institute for Machine Learning and Artificial Intelligence, Germany}{bkolbe@uni-bonn.de}{https://orcid.org/0009-0005-0440-4912}{This research has partly been funded by the Federal Ministry of Education and Research of Germany and the state of North-Rhine Westphalia as part of the Lamarr-Institute for Machine Learning and Artificial Intelligence.} 

\author{Philip Mayer\footnote{corresponding author}}{University of Bonn, Germany}{pmayer@uni-bonn.de}{https://orcid.org/0009-0007-4800-7753}{}

\author{Petra Mutzel}{University of Bonn \and Lamarr Institute, Germany}{pmutzel@uni-bonn.de}{https://orcid.org/0000-0001-7621-971X}{}

\author{Alexander Naumann}{University of Bonn, Germany}{naumann@igg.uni-bonn.de}{https://orcid.org/0009-0009-5442-3336}{}

\author{Jonas Sauer}{Karlsruhe Institute of Technology, Germany}{jonas.sauer@kit.edu}{https://orcid.org/0000-0002-7196-7468}{}

\authorrunning{Blank, Eppstein, Haunert, Haverkort, Kolbe, Mayer, Mutzel, Naumann, Sauer} 

\Copyright{Lotte Blank and David Eppstein and Jan-Henrik Haunert and Herman Haverkort and Benedikt Kolbe and Philip Mayer and Petra Mutzel and Alexander Naumann and Jonas Sauer} 



\category{} 

\relatedversion{} 

\supplement{Our source code and datasets are available in a public repository: \url{https://github.com/GeometryCodes/ArbitraryBicriteriaShapes}}

\funding{This research was partially funded by the Deutsche Forschungsgemeinschaft (DFG, German Research Foundation) under grant FOR-5361 -- 459420781.}

\acknowledgements{We thank Anne Driemel and Frederik Brüning for fruitful initial discussions.}

\begin{document}

\maketitle

\begin{abstract}
We study the problem of aggregating a set of polygons by covering them with disjoint representative regions, thereby inducing a clustering of the polygons. Equivalently, this can be seen as a fence enclosure problem, where the goal is to surround the polygons with a set of closed curves.
Our objective is to minimize a weighted sum of the total area and the total perimeter of the regions, which naturally extends other fencing problems and has applications in geographical information systems.
Previously, this objective was only studied in a restricted variant, in which the boundary curves of the regions must be selected from a fixed subdivision of the plane. 
It is natural to ask whether the problem is still tractable if this restriction is removed, allowing output regions to be bounded by arbitrary curves. We provide a positive answer in the form of an algorithm with runtime $\mathcal{\tilde{O}}(n^4)$, where $n$ is the number of input vertices. To achieve this, we fully characterize the optimal solutions by showing that their boundaries are composed of input edges and circular arcs of constant radius.
Additionally, we consider the parametric problem, where for every weighting factor we seek a solution that is optimal for that trade-off of area and perimeter. We show that $\bigO{n^2}$ combinatorial solutions suffice to describe all optimal solutions across all weighting factors, and provide both an exact algorithm and an approximation scheme. 
To make the algorithms scalable in practice, we develop engineering techniques that exploit structural properties of the solutions.
Our experimental evaluation on real-world data shows linear runtime in practice, even for the parametric variant.

\end{abstract}

\section{Introduction}\label{sec:intro}
\begin{figure}[!bp]
    \centering
    \includegraphics[width=0.95\textwidth]{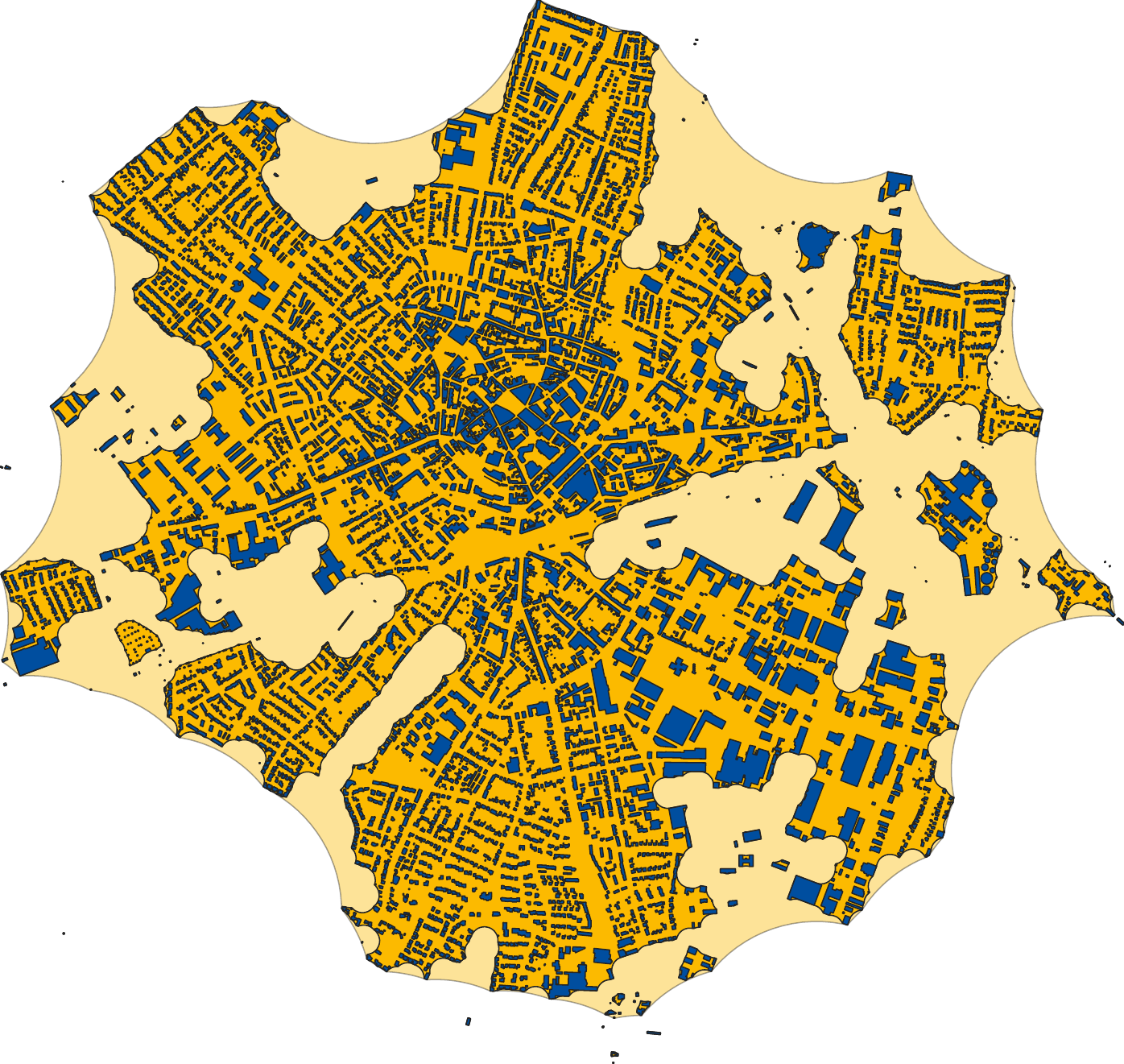}
    \caption{All figures show input polygons in blue and solutions in orange. Shown here is the unrestricted polygon aggregation problem with objective function~$g_\alpha$ for the town of Euskirchen. Two optimal solutions are superimposed in increasingly lighter shades of orange, with~$\alpha \in \{500,3000\}$.}
    \label{fig:initial_example}
\end{figure} 
    We study the task of aggregating a set of polygons~$\polygons$ in the plane by covering them with a set of disjoint representative regions~$\solution$, i.e., for every polygon~$B \in \polygons$ there must be a region~$S \in\solution$ with~$B \subseteq S$; see \Cref{fig:initial_example}. This problem formulation appears frequently in cartographic applications~\cite{rottmann2024bicritshapes, Peng_BuildingGeneralization_2017,Funke2024Smoot-70095,FME_area_amalgamator}.
    In the context of computational geometry, it can be classified as a type of \emph{fence enclosure}, or \emph{fencing}, problem~\cite{AbrahamsenABCMR18,DBLP:conf/wads/ArkinKM91,DBLP:journals/jal/CapoyleasRW91}. Here, given a set of points or polygons, we want to enclose them by a set of closed curves (fences) such that a certain objective is achieved.
    For example, Arkin et al.~\cite{DBLP:conf/wads/ArkinKM91} consider the problem of surrounding resource-rich plots of land by one or multiple fences while minimizing construction costs.

    In our work, we consider a bi-criteria objective function, which was first proposed by Rottmann et al.~\cite{rottmann2024bicritshapes} and is motivated by geographical information systems~(GIS): a solution~$\solution$ is optimal for some parameter~$\alpha \in [0,\infty]$ if it minimizes the linear combination  $g_\alpha(\solution)=\A(\solution)+\alpha\PP(\solution)$ of the overall area~$\A(\solution)$ and the overall perimeter~$\PP(\solution)$.   
    In the fencing analogy, the perimeter term corresponds to the cost of the fences and the area term to the cost of acquiring additional land. Both terms have been considered as separate objective functions for the fencing problem~\cite{DBLP:conf/wads/ArkinKM91,AbrahamsenABCMR18}, but not jointly.
    From a theoretical perspective, their combination into a bi-criteria cost function provides a natural generalization.

    The work by Rottmann et al.~\cite{rottmann2024bicritshapes}  considers the objective function~$g_\alpha$ in very restricted settings, leading to a constrained problem.
    Here, a subdivision~$\subdivision$ of the plane (e.g., a constrained Delaunay triangulation) is fixed and the boundary curves of the representative regions must be edges of~$\subdivision$. Rottmann et~al.\ give a transformation of this problem to a graph-cut problem, allowing it to be solved in~$\mathcal O(n^2 / \log n)$ time. 
    In addition to the problem where $\alpha$ is fixed, they also consider the parametric problem of finding a set~$\mathfrak{S}$ that contains an optimal solution for every value of~$\alpha$.
    {A powerful property of the objective function~$g_\alpha$ is that the optimal solutions are \emph{nested} with respect to~$\alpha$, i.e., an optimal solution for value $\alpha$ is contained in all optimal solutions for values $\alpha'>\alpha$~\cite{rottmann2024bicritshapes}.}
    {This implies that the optimal solutions for different values of~$\alpha$ induce a hierarchical clustering of the input polygons.}
     Using this, they show that~$\mathfrak{S}$ has size~$\bigO{n}$ and can be computed using a dichotomic scheme.

    Fixing the subdivision~$\subdivision$ makes both problems much easier to solve, but this comes at the cost of constraining the solution space. This raises a natural question: is the problem still tractable if the restrictions imposed by the subdivision are lifted?
     We answer this question by investigating a more general problem setting: any set of regions with differentiable boundary curves that covers~$\polygons$ is allowed, and we aim to find one that minimizes $g_\alpha$.
     
    \subparagraph*{GIS Applications.} Rottmann et al.\ introduced the bicriteria polygon aggregation problem in a GIS context, where polygons represent map features such as buildings or bodies of water. One potential application is map generalization, where a given map must be transformed into a less detailed map to be displayed at a smaller scale. This requires aggregation and omission of details. Especially at the transition from large-scale to small-scale maps, buildings must be grouped into areas representing entire settlements~\cite{Funke2024Smoot-70095,rottmann2024bicritshapes,Peng_BuildingGeneralization_2017}.
    The delineation of settlement areas is also an important task in urban analytics, e.g., to study urban growth and morphology, where the aggregations provide evidence of urban land use~\cite{Arribas2021,DeBellefon2021,Harig2021}.
    
    In such applications, the objective function~$g_\alpha$ models a trade-off between two natural goals: The output regions should be compact and have low shape complexity, which is captured by minimizing the perimeter~\cite{MacEachren_MapComplexity_1982,Harrie_MapReadability_2015}. On the other hand, they should remain as faithful as possible to the input polygons, which is captured by minimizing the area. The trade-off parameter~$\alpha$ corresponds to the scale at which the map is viewed: larger scales (zoomed in) prioritize faithfulness to the original input, whereas smaller scales (zoomed out) prioritize simplicity.
    {The nestedness property with respect to the parameter $\alpha$ guaranteed by our problem formulation provides visual stability across scale changes~\cite{Funke2024Smoot-70095,rottmann2024bicritshapes,Peng_BuildingGeneralization_2017}, which is a desirable property in map generalization.}
    Furthermore, the spatial patterns encoded in the induced (hierarchical) clustering are useful for data analysis tasks in urban morphology~\cite{AndersSester2000}, such as the classification and comparison of different settlement types~\cite{Jochem2021,Taubenbock2020}.

\subparagraph*{Our Contributions.}
\begin{itemize}
    \item We characterize solutions for the unrestricted polygon aggregation problem (denoted $\freeProblem$) and show that their boundary curves consist of parts of the input boundaries and circular arcs.
    Unlike previous approaches, the shapes of the boundary curves are not specified in advance.
    Rather, they emerge naturally as a result of minimizing the objective.
    \item Using this characterization, we construct a subdivision~$\subdivision_C$ of size~$\bigO{n^4}$ such that an optimal solution for~$\freeProblem$ can be obtained as an optimal solution for the subdivision-restricted problem~$\problem{\subdivision_C}$ in $\bigtO{n^4}$ time. 
    \item We show that the optimal solutions for~$\freeProblem$ are nested with respect to~$\alpha$. In addition, we show that the problem also exhibits a \emph{subset-nestedness} property: if new polygons are added to the input, the optimal solution will only grow.
    \item We propose a preprocessing algorithm that leverages subset-nestedness by iteratively merging sets of nearby polygons that belong to the same output region. This drastically improves the runtime, especially for large~$\alpha$ values. 
    \item We prove that a solution set~$\mathfrak{S}$ for the parametric unrestricted  aggregation problem contains~$\bigO{n^2}$ (combinatorial) solutions.
    \item  We give a fully polynomial-time approximation scheme (FPTAS) for the parametric problem. We also give an exact polynomial-time algorithm under the assumption that we have access to an oracle that returns an intersection point of two functions.
    \item We present engineering ideas, based on $\alpha$-nestedness, to improve the runtime in practice.
    
    \item We evaluate our algorithms in the map generalization task where real-world building footprints from towns and cities are aggregated. Here, we show that, despite the $\bigtO{n^4}$ worst-case runtime, our approach scales linearly with the instance size and~$\alpha$ in practice.
\end{itemize}

\section{Related Work}

\subparagraph*{Fencing.}
Several fencing problems for point sets in the plane were first introduced by Arkin, Khuller, and Mitchell~\cite{DBLP:conf/wads/ArkinKM91}, as well as Capoyleas, Rote, and Woeginger~\cite{DBLP:journals/jal/CapoyleasRW91}. Capoyleas et al.\ mainly study the $k$-clustering variant, where the goal is to enclose the input points with~$k$ fences, whereas Arkin et al.\ allow either exactly one or an arbitrary number of fences and consider several possible cost measures, including the fence length and the enclosed area.

Abrahamsen et al.~\cite{AbrahamsenABCMR18} give polynomial-time algorithms for two fencing variants: (1) Find a set of enclosing curves of minimal length where each curve additionally incurs a fixed opening cost. (2) Find a set of~$k$ enclosing curves of minimal length. 
Another fencing variant is \emph{geometric multicut}~\cite{DBLP:journals/dcg/AbrahamsenGLR20}, where input polygons are divided into~$k$ different color classes and the objective is to find a minimum-perimeter fencing that separates polygons with different colors from each other. For~$k=2$, this can be reduced to a min-cut problem and solved in~$\bigtO{n^4}$ time. For~$k \geq 3$, it is NP-hard, but an approximation algorithm is given. 
 
\subparagraph*{Polygon Aggregation.} Polygon aggregation can be viewed as a combination of two subproblems: partitioning the polygons into clusters and computing a representative shape for each cluster.
Several popular approaches specify the type of representative shape in advance, which reduces it to a clustering problem.
For example, minimum-perimeter fencing~\cite{AbrahamsenABCMR18} corresponds to choosing convex hulls as the representative shapes.
A drawback of this approach is that it tends to include large empty areas in the solution.
The same problem occurs when using the $\chi$-hull, a generalization of the convex hull for aggregation purposes~\cite{DUCKHAM20083224}.
Other generalizations of the convex hull, such as the concave hull introduced by Moreira and Santo~\cite{MoreiraS07}, or $\alpha$-shapes~\cite{DBLP:journals/tog/EdelsbrunnerM94}, mitigate this problem.
However, both of these approaches can lead to narrow bridges in the output (see Appendix~\ref{sec:alpha_shapes}). 
Recently, Funke and Storandt~\cite{Funke2024Smoot-70095} proposed an approach that first clusters the polygons based on a distance threshold~$\beta$ and then computes an~$\alpha$-shape for each cluster. Clusters with fewer points than some threshold value~$\mu$ are discarded altogether. 
Their approach runs in~$\bigO{n\log n}$ time. Although this allows for the aggregation of country-sized instances within minutes, a major drawback is that the three parameters ($\alpha, \beta$ and~$\mu$) must be specified by the user, and the narrow bridge problem is only alleviated if they are chosen appropriately for the given instance.

Two approaches deviate from the paradigm of choosing the type of representative shape in advance:
the method of Damen, van~Kreveld and Spaan~\cite{damen2008high}, which applies morphological closure and opening operators to the input polygons, and the previously mentioned approach by Rottmann et~al.~\cite{rottmann2024bicritshapes}. The parametric problem proposed by Rottmann et~al.\ has recently been reduced to the source-sink-monotone parametric min-cut problem~\cite{beines2024}, {which enables the use of faster algorithms in practice.}
Using this reduction, instances representing large cities can be solved optimally within seconds.


\subparagraph*{Map Generalization.} Besides polygon aggregation, map generalization involves a range of well-studied operators, including elimination, displacement, exaggeration, detail elimination, squaring, and typification (see \cite{MCMasterGeneraliazation1992,damen2008high} for further reading). 
Some approaches focus primarily on the simplification (and local aggregation) of polygons in large-scale maps, where individual buildings are still distinguishable~\cite{HaunertWolff2010Simplification,BuchinMeulemansSpeckmann2011,BuchinMWRS2016}.
However, maps of scales smaller than $1:50\,000$ usually do not show individual buildings but rather representatives of entire settlements~\cite{Touya17_UrbanAreaVisualization}. 

\section{Problem Definition}\label{sec:problem_def}
We introduce some basic notation. A curve is a continuous map $f\colon[0,1]\to \mathbb{R}^2$. We will also identify a curve with its image. The length of a piecewise differentiable curve $f$ is denoted by $\LL(f)$. The area of a (Lebesgue-measurable) closed region~$S$ is denoted by $\A(S)$.
The perimeter~$P(S)$ is the length of its boundary~$\partial S$.
If $S$ is a polygon, its vertices are denoted by $V(S)$, its edges by $E(S)$, and the convex hull of $S$ by $\conv(S)$.
Additionally, we assume that outer boundary curves of regions are oriented counter-clockwise and inner boundary curves (which enclose holes) are oriented clockwise. Thus, along each boundary curve, the interior of $S$ lies locally to the left. 
We extend the above definitions to sets $\solution = \{S_1, \ldots, S_k\}$ of (disjoint) closed regions or polygons. Most applications focus on polygonal inputs, but during the preprocessing algorithm discussed at the end of Section~\ref{sec:unrestricted:paper}, we have to consider intermediate sub-instances that are bounded by circular arcs in addition to straight lines.
\begin{definition}[Circular polygon]\label{def:circular-polygon:paper}
In a circular polygon, every edge, i.e., a boundary segment connecting two vertices, is a straight line or a circular arc that bends inwards and has central angle not larger than $\pi$. In an $\alpha$-circular polygon, all circular arcs have radius $\alpha$. 
\end{definition}
\begin{problem_description}[Unrestricted polygon aggregation problem $\freeProblem$]\label{prob:generalproblem:paper}
Let $\polygons$ be a set of interior-disjoint (circular) polygons and~$\alpha \geq 0$. A (feasible) solution to $\freeProblem$ is a set $\solution$ of disjoint closed regions with piecewise differentiable boundary curves that cover all input polygons, i.e., for every $B\in\polygons$ there exists an  $S\in\solution$ with $B\subseteq S$. A~solution is optimal if it minimizes the objective
\begin{align*}
    g_\alpha(\solution)=\A(\solution)+\alpha\PP(\solution)
\end{align*}
among all feasible solutions.
For~$\alpha = \infty$, the objective function becomes~$g_\infty(\solution) = \PP(\solution)$.
\end{problem_description}

This objective function is equivalent to the function $f_\lambda(\solution)=\lambda A(\solution)+(1-\lambda)P(\solution)$ from~\cite{rottmann2024bicritshapes} with the reparameterization~$\alpha=\frac{1-\lambda}{\lambda}$.
Additionally, most of the following results extend to inputs with arbitrary piecewise differentiable boundaries (cf. Appendix~\ref{sec:unrestricted:appendix}). 

\begin{definition}
A curve~$f$ from~$u$ to~$v$ is a \textbf{boundary piece} of a solution~$\solution$ if~$f \subseteq \partial \solution$ \done{and~$\solution$ includes} the area to the left of~$f$.
\done{A boundary piece}~$f$ \done{is} \textbf{constrained} if~$f \subseteq \partial \polygons$ and \textbf{free} if~$f \cap \partial \polygons \subseteq \{u,v\}$.
The set of (inclusion-)maximal free pieces in~$\solution$ is denoted by~$F_\polygons(\solution)$.
\end{definition}
Note that every boundary curve can be described as an alternating sequence of constrained and maximal free boundary pieces. The constrained pieces may consist of a single point; see \Cref{fig:free_piece_description:paper} for an example.
\begin{figure}[tbh]
\centering
\includegraphics[scale=0.9]{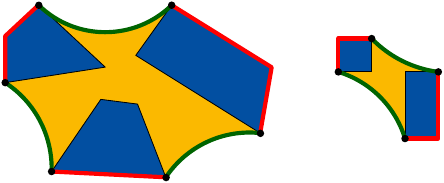}
\caption{An example solution in which every solution region is a circular polygon. Free pieces are shown in dark green, and constrained pieces are shown in red. The vertices are the endpoints of the inclusion-maximal free pieces.}
\label{fig:free_piece_description:paper}
\end{figure}

\section{Solving the Unrestricted Polygon Aggregation Problem}\label{sec:unrestricted:paper}
In this section, we investigate the unrestricted problem~$\freeProblem$. We begin by characterizing the shape of the free boundary pieces and show that only a polynomial number of curves need to be inspected to find an optimal solution. We then use these curves to build a subdivision $\subdivision_C$ of $\conv(\polygons)\setminus\polygons$ that has polynomial size, thereby reducing $\freeProblem$ to the subdivision-restricted problem~$\problem{\subdivision_C}$.
This allows us to apply the minimum cut-based algorithm proposed in~\cite{rottmann2024bicritshapes} to solve the unrestricted problem in polynomial time.
Finally, we note that the optimal solutions to~$\freeProblem$ are nested with respect to~$\alpha$ and with respect to introducing new input polygons.
We sketch some proofs and omit others; a full version can be found in~\Cref{{sec:unrestricted:appendix}}.

{We start by characterizing the optimal solutions of the corner cases~$\alpha=0$ and $\alpha=\infty$.
In the former, the input polygons~$\polygons$ form an optimal solution because they minimize the area.
In the latter, the perimeter must be minimized, so every region~$S$ in an optimal solution must be the convex hull~$\conv(\polygons(S))$ of its contained (circular) polygons~$\polygons(S)$.
It is easy to see that the general case~$\alpha \in (0,\infty)$ can be reduced to the case~$\alpha = 1$ by scaling the input.
\begin{restatable}{observation}{observationScale}
\label{obs:scale:paper}
For a set~$\mathcal{R}$ of regions and~$c > 0$, let~$\sigma_c(\mathcal{R})$ denote the same set of regions scaled uniformly by the factor~$c$.
A solution~$\solution$ for~$\freeProblem$ with input polygons~$\polygons$ and~$\alpha \in (0,\infty)$ is optimal iff~$\sigma_{1/\alpha}(\solution)$ is optimal for~$P^1_F$ with input polygons $\sigma_{1/\alpha}(\polygons)$.
\end{restatable}
Thus, we only consider the cost function $g_1$ for the proofs in this section.
\Cref{obs:scale:paper} describes a crucial property of the problem formulation: the parameter~$\alpha$ corresponds directly to the scale at which the input polygons are viewed.
It follows that the set of solutions that are optimal for at least one value of~$\alpha$ does not change when scaling the input.

We continue by investigating which curves locally optimize the area-perimeter trade-off.
\begin{restatable}{definition}{efficiency}\label{def:efficiency:paper}
Let~$f$ be a curve from $u$ to $v$ that does not intersect the line through~$u$ and~$v$ properly.
Let~$R^{{uv}}_f$ denote the region enclosed by $\overline{uv}$ and $f$.
The \textbf{efficiency} of~$f$ is defined as
\[
e_{uv}(f)=\begin{cases}
-\LL (f)+\LL (\overline{uv})+\A(R^{uv}_f) & \text{if $f$ is to the left of~$\overrightarrow{uv}$,}\\
-\LL (f)+\LL (\overline{uv})-\A(R^{uv}_f) & \text{if $f$ is to the right of~$\overrightarrow{uv}$.}
\end{cases}
\]
\end{restatable}
If we consider a solution that uses~$\overrightarrow{uv}$ as a free boundary piece, then~$e_{uv}(f)$ is the improvement in the objective value if we replace~$\overrightarrow{uv}$ with~$f$.
We now characterize the curves that have maximum efficiency and show that any solution that uses a curve of non-maximum efficiency can be improved locally.
\begin{figure}
    \centering

    \begin{minipage}[t]{0.6\textwidth}
        \centering
        \includegraphics[width=0.45\textwidth]{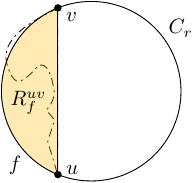}
        \caption{For fixed $\ell(f)$, the enclosed area and in turn the efficiency is maximized if $f$ is circular.}\label{fig:ArcIsBest:paper}
    \end{minipage}
    \hfill
    \begin{minipage}[t]{0.38\textwidth}
        \centering
        \includegraphics{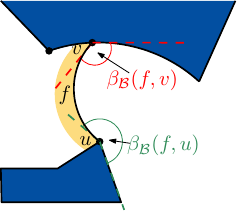}
        \captionsetup{textformat=simple}\caption{Two examples for the attachment angle $\beta_\polygons(f,x)$ .}\label{fig:outer_arc_angle:paper}
    \end{minipage}
\end{figure}
\begin{restatable}{lemma}{efficiencyCircular}
    \label{lemma:arcs_enclose_much:paper}
    Let $u,v$ be two points. For fixed length $l$, the (unique) most efficient curve~$f$ from $u$ to $v$ with $\LL(f)=l$ that does not intersect the line through $u$ and $v$ properly is a circular arc that lies to the left of~$\overrightarrow{uv}$.
\end{restatable}
\begin{proof}
    Let $C_r$ be the circular arc with radius $r$ from $u$ to $v$ in the half-plane to the left of~$\overrightarrow{uv}$ 
    such that $\LL(C_r)+\LL(f)=2\pi r$ (see~\cref{fig:ArcIsBest:paper}). 
    The curve~$f$ lies to the left of~$\overrightarrow{uv}$ since the efficiency is non-positive otherwise. Hence, $C_r$ together with $f$ is a closed curve with perimeter $P=2\pi r$. The isoperimetric inequality~\cite{Steiner1838, hurwitzisoperimetric} states that for every closed curve~$C$ with enclosed region~$R$, it holds that $4\pi \A(R)\leq \LL(C)^2$. Furthermore, equality holds if and only if $C$ is a circle.
    It follows that $C_r$ together with $f$ has to be a circle to maximize the enclosed area and in turn maximize the efficiency. Consequently, $f$ is a circular arc.
\end{proof}
In the remainder of this paper, we denote by $\sarc{r}{uv}$ the circular arc of radius $r$ and central angle $\theta \leq \pi$ that connects the points $u$ and $v$ and lies to the left of~$\overrightarrow{uv}$.
The next lemma follows by combining \Cref{lemma:arcs_enclose_much:paper} with an analytical argument and a local improvement step. 


\begin{restatable}{lemma}{localImprovements}
\label{lemma:local_structure:paper}
    Let $\solution$ be a solution for \done{$P_F^1$} and $u,v$ two points on the boundary of~$\solution$.
    If~$u$ and~$v$ are connected by a free piece other than~\done{$C^{uv}_{1}$},
    then $\solution$ is not optimal.
\end{restatable}

We briefly describe how free boundary pieces attach to the input boundaries. 
At any point where a curve meets an input boundary, we compare their directions via their tangents; at corner points, we use the one-sided tangents of the adjacent boundary segments. 
This allows us to measure the counterclockwise angle between the free piece and the input boundary at the point of attachment. We denote this angle by $\beta_\polygons(f,x)$ and interpret it as the exterior angle at $x$ of a solution region that uses $f$; see Figure~\ref{fig:outer_arc_angle:paper}. 
A formal definition is given in Appendix~\ref{sec:unrestricted:appendix}.
Using this, we characterize the free pieces that appear in an optimal solution.

\begin{proposition}\label{prop:arc_properties:paper}
    Let $\polygons$ be an instance of~$\freeProblem$ \done{for~$\alpha\in(0,\infty)$}. There exists a solution $\solution$ that minimizes $g_\alpha(\solution)$, such that the boundary $\partial \solution$ is a collection of piecewise differentiable curves and every free boundary piece $f$ with endpoints $u$ and $v$ has the following properties:
    \begin{description}
    \item[P1:] The distance $d=\LL(\overline{uv})$ of the endpoints is at most $2\alpha$.  
    \done{\item[P2:] The piece $f$ is circular with radius $\alpha$.} 
    \item[P3:] The central angle $\theta$ of $f$ is at most $\pi$.
    \item[P4:] If \(x \in \{u,v\}\) lies on the input boundary~$\partial \polygons$, then \(\beta_\polygons(f,x) \ge \pi\). 
    \done{\item[P5:] The piece $f$ bends to the left (i.e., it is to the left of $\overrightarrow{uv}$).}
    \end{description}
\end{proposition}

\begin{proof}[Proof sketch]
    The existence of an optimal solution follows from results in geometric measure theory (cf.~\Cref{prop:there_exists_a_solution:full} in~\Cref{{sec:unrestricted:appendix}}).
    Properties P1, P2, and P5 follow directly from~\Cref{lemma:local_structure:paper}.
    If~$f$ has a central angle~$\theta > \pi$, we can replace it by the corresponding arc~$f'$ with central angle~$2\pi - \theta$ that bends in the same direction. A comparison shows that the area gain is always outweighed by the perimeter loss, which implies P3.
    For P4, one can show (cf.\ Appendix~\ref{sec:unrestricted:appendix}) that if \(\beta_\polygons(f,x) < \pi\), then a generalized version of~\Cref{lemma:local_structure:paper} applies to the partially constrained piece formed by $f$ concatenated with the next segment on~\(\partial B\).
\end{proof}
With this we can define a set~$\allowedArcs$ of free pieces such that there exists an optimal solution that uses only maximal free pieces from~$\allowedArcs$. 
\begin{definition}\label{def:arcs:paper}
    Let~$\polygons$ be an instance of~$\freeProblem$. For~$\alpha\in(0,\infty)$, let~$\allowedArcs$ denote the set of circular arcs~$f=\sarc{\alpha}{uv}$ such that 
    \begin{enumerate}
        \item the endpoints~$u$ and~$v$ lie on the input boundary~$\partial \polygons$, 
        \item {$u$ and $v$ are not in the interior of two parallel straight-line edges,
        \item if $u \in V(\polygons)$ and $v$ lies in the interior of a circular boundary arc $h$ (or vice versa), then $u$(or $v$) is not the center point of $h$,}
        \item the arc~$f$ fulfills properties P1--P5 of~\Cref{prop:arc_properties:paper}, and
        \item the arc~$f$ does not intersect~$\polygons$ properly.
    \end{enumerate}
    We define~$F_0(\polygons)=\emptyset$ and~$F_\infty(\polygons)$ as the set of straight-line segments~$\overline{uv}$ with~$u,v \in V(\polygons)$ that do not intersect~$\polygons$ properly.
    Let~$\subdivision_C$ denote the subdivision of~$\conv(\polygons)\setminus\polygons$ induced by~$\allowedArcs$ and~$E(\polygons)$, and let~$|\subdivision_C|$ denote the total number of vertices, edges, and cells in~$\subdivision_C$.
\end{definition}
Properties (2) and (3) handle ambiguous cases in which an infinite number of arcs connecting two boundary objects with the same objective value may exist (cf.\ Lemma~\ref{lemma:distance2}).
We show that~$\subdivision_C$ has polynomial size and is sufficient to construct an optimal solution.
\begin{restatable}{lemma}{freeSubdivision}
For an instance $\polygons$ of problem $\freeProblem$ with~$|V(\polygons)|=n$, the subdivision~$\subdivision_C$ has size~\done{$|\subdivision_C|\in\bigO{n^4}$} and every optimal solution for~$\problem{\subdivision_C}$ is an optimal solution for~$\freeProblem$.
\label{lemma:freeSubdivision:paper}
\end{restatable}

\begin{proof}[Proof sketch]

Given a pair $b_1,b_2$ of boundary objects, we consider the possible center locations for an arc 
$f = \sarc{\alpha}{uv} \in {F_\alpha}(\polygons)$ with $u\in b_1$ and $v\in b_2$. Simple geometric arguments show that these center locations are intersections of line segments, circles, and circular arcs. 
After resolving ambiguous center locations (such as coinciding line segments and circles), each pair of boundary objects induces at most two arcs in ${F_\alpha}(\polygons)$. 
Because the number of boundary objects is in~$\bigO{n}$, it follows that $\lvert {F_\alpha}(\polygons) \rvert \in \bigO{n^2}$.
Each pair of arcs in~$\allowedArcs \cup E(\polygons)$ intersects properly at most a constant number of times, so the number of vertices, \done{edges} and cells in~$\subdivision_C$ is in $\bigO{n^4}$.
By construction, every boundary piece in~$\solution$ is a path in~$\subdivision_C$.
Hence, $\solution$ is a solution for~$\problem{\subdivision_C}$.
Every optimal solution for~$\problem{\subdivision_C}$ is also a solution for~$\freeProblem$.
To be optimal for~$\problem{\subdivision_C}$, it must have the same objective value as~$\solution$, so it is also optimal for~$\freeProblem$.
\end{proof}
Note that although there may be two arcs connecting the same pair~$b_1,b_2$ of boundary objects, they are uniquely defined by their direction because one of the arcs starts at $b_1$ and the other one at $b_2$. 
Thus, for input edges \( e, h \in E(\polygons) \), we can write \( \sarc{r}{eh} \) to denote the circular arc~$\sarc{r}{uv}$ that satisfies P4 in \cref{prop:arc_properties:paper}, starts at~$u$ on~$e$ and ends at~$v$ on~$h$.

By leveraging planar multi-source-multi-sink min-cut algorithms \cite{DBLP:journals/siamcomp/BorradaileKMNW17,DBLP:conf/icalp/GawrychowskiK18} and a graph transformation adapted from~\cite{DBLP:journals/dcg/AbrahamsenGLR20}, we show in \cref{sec:Rottmann_transformation} that for any subdivision $\mathcal{D}$, the problem $\problem{\subdivision}$ where we only consider free pieces that are paths in $\mathcal{D}$ can be solved in quasilinear-time. Applying these results to $\problem{\subdivision_C}$ yields a polynomial-time algorithm for~$\freeProblem$.
\begin{restatable}{theorem}{algorithm}
\done{\label{theorem:algorithm:paper}
  For an instance $\polygons$ of problem $\freeProblem$ with~$|V(\polygons)|=n$, an optimal solution can be computed in~$  \tilde{\mathcal{O}}\left(|\subdivision_C|\right)=\mathcal{O}\left(|\subdivision_C| \frac{\log^3|\subdivision_C|}{\log^2 \log |\subdivision_C|}\right)\subseteq \mathcal{O}\left(n^4 \frac{\log^3n}{\log^2 \log n}\right)$ time.}
\end{restatable}

Finally, we observe that optimal solutions are nested with respect to changing~$\alpha$ as well as adding polygons, properties we will crucially exploit to obtain a practical algorithm.
\begin{restatable}[$\alpha$-nestedness]{proposition}{alphaNestedness}\label{prop:alpha-nested:paper}
    For an instance~$\polygons$ of~$\freeProblem$ and two values~$0 \leq \alpha < \alpha'$, let~$\solution$ and~$\solution'$ be the respective optimal solutions.
    Then~$\solution\subseteq \solution'$.
\end{restatable}
\begin{restatable}[subset-nestedness]{proposition}{subsetNestedness}
\label{lemma:subset_nestedness:paper}
    Let $\polygons$ and $\polygons'$ be two sets of input polygons, such that $\polygons' \subseteq \polygons$. For a given parameter $\alpha$, let $\solution_{\polygons'}$ be a solution for $\polygons'$ minimizing $g_\alpha$. Then there exists a solution $\solution_{\polygons}$ minimizing $g_\alpha$ for $\polygons$, such that $\solution_{\polygons'} \subseteq \solution_{\polygons}$. 
\end{restatable}

\subparagraph{A Local Preprocessing Algorithm.}\label{sec:preprocessing:paper}
The main contributor to the runtime of the algorithm is the size of the subdivision $\subdivision_C$, which is quadratic in the size of the arc set~$F_\alpha(\polygons)$. While $|F_\alpha(\polygons)|$ is manageable for small $\alpha$, it grows significantly with increasing $\alpha$, leading to subdivisions that become intractable for large $\alpha$. This is because more pairs of potential endpoints satisfy the distance bound of~$2\alpha$ imposed by property~P1.
This effect is partially offset by $\alpha$-nestedness: as $\alpha$ increases, regions in an optimal solution grow, rendering more arcs irrelevant since they intersect an optimal region. Although the exact optimal regions are unknown a priori, subset-nestedness allows us to approximate them from below.

Our preprocessing follows a simple idea: iteratively merge nearby polygons that belong to the same region in an optimal solution. Let $\optimizer_\alpha(\polygons)$ denote the solution of our algorithm applied to $\polygons$. First,  we replace each polygon $B\in \polygons$ by $\optimizer_\alpha(\{B\})$. Then we compute a Delaunay triangulation $\text{DT}$ on the centroids of $\polygons$ and process edges $e \in \text{DT}$ in increasing order of length. For each edge connecting $B_i$ and $B_j$, we compute $\optimizer_\alpha(\{B_i,B_j\})$. If this yields a single region, we replace $B_i$ and $B_j$ by the merged region $\optimizer_\alpha(\{B_i,B_j\})$. Due to subset-nestedness (\Cref{lemma:subset_nestedness:paper}), this merged representative is a subset of the globally optimal solution. Note that, during the preprocessing it may happen that some merged representatives of different polygon clusters overlap properly. We show in \Cref{sec:unrestricted:appendix}  (cf. \Cref{theorem:circ_polygons}) that our algorithm can also handle this case without increasing the runtime.

In theory, the preprocessing may increase the runtime by a factor of $\bigO{n}$, but in practice the subproblems are solved quickly. Further details are given in \Cref{sec:preprocessing:appendix}.

\section{Solving the Parametric Problem}\label{sec:all_solutions:paper}
We turn to the parametric variant of the problem~$\freeProblem$, in which the objective is to find an optimal solution for every~$\alpha\in[0,\infty]$.
For the subdivision-restricted variant, $\alpha$-nestedness implies that there are~$\mathcal O(n)$ optimal solutions~\cite{rottmann2024bicritshapes}, but this implication does not hold for~$\freeProblem$.
Because our definition of solutions is geometric, even an infinitesimal change in $\alpha$ alters the radii of the circular arcs and, in turn, changes the optimal solution.
Hence, the number of optimal geometric solutions is infinite.
To remedy this, we derive a solution description that is geometry-independent and is strictly given combinatorially with respect to the polygon edges and vertices. We show that an optimal parametric solution contains~$\mathcal{O}(n^2)$ combinatorial solutions, and we present a bisection scheme to compute such a solution. We only give intuitive definitions and omit technical details and proofs; these are given in Appendix~\ref{sec:parametric:appendix}.
\subparagraph*{The Parametric Problem.}
\begin{figure}
\centering
\includegraphics{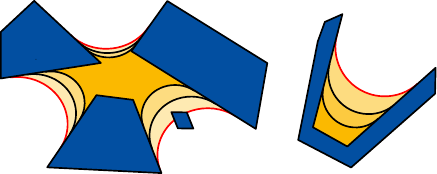}
\caption{Three realizations of the same combinatorial solution for different~$\alpha$. The arcs move with~$\alpha$ but connect the same edge pairs. The red realization is invalid because it intersects a polygon.}
\label{fig:combinatorial_paper:paper}
\end{figure}
We already established that for fixed $\alpha$, every free arc $f\in\allowedArcs$ is defined by the ordered pair $x,y$ of boundary objects that are connected by it. Hence,
for two objects (vertices or edges)~$x,y$ and two values~$\alpha,\alpha'$, the arcs~$\sarc{\alpha}{xy}$ and~$\sarc{\alpha'}{xy}$ are different geometric realizations of the same combinatorial arc~$\sarc{}{xy}$. We define a \emph{combinatorial boundary} as a cyclic sequence of combinatorial arcs and (combinatorial) pieces of the input boundary. Its realization for~$\alpha$ is the closed piecewise curve formed by the respective realizations of the combinatorial arcs and the pieces of the input boundary connecting them.
A \emph{combinatorial region}~$R$ is described by a set of combinatorial boundaries, one of which is designated as the outer boundary and the rest as inner boundaries.
We call~$R$ \emph{valid} for~$\alpha$ if its realization~$R[\alpha]$ forms a region.
A \emph{combinatorial solution}~$\combinatorialsol$ is a set of combinatorial regions; see~\Cref{fig:combinatorial_paper:paper}.
It is \emph{valid} for~$\alpha$ if its realization~$\combinatorialsol[\alpha]$ is a valid geometric solution for~$\freeProblem$.
Notably, $\alpha$-nestedness also holds for realizations of the same combinatorial region.%
\begin{corollary}\label{cor:combinatorial-nested_paper:paper}
    For a combinatorial region~$R$ and two values~$\alpha_1<\alpha_2$ such that~$R$ is valid, we have~$R[\alpha_1] \subseteq R[\alpha_2]$.
\end{corollary}
With this, we can show the following, which is crucial for solving the parametric problem.

\begin{restatable}{proposition}{solutionValidityRange}\label{prop:solution-validity-range_paper:paper}
    For a combinatorial solution~$\combinatorialsol$, the set~$\mathcal{I}(\combinatorialsol)\subseteq(0,\infty)$ of values for which~$\combinatorialsol$ is valid is a single open interval.
\end{restatable}

Using combinatorial solutions, we define the parametric problem.%
\begin{problem_description}[The Parametric Problem]\label{problem:problem_2:paper}
    Let~$\polygons$ be a set of interior-disjoint polygons in the plane. A solution to the parametric unrestricted polygon aggregation problem $\biProblem$ is a set $\mathfrak{K}$ of combinatorial solutions such that for every $\alpha\in[0,\infty]$, there exists a combinatorial solution $\combinatorialsol\in\mathfrak{K}$ such that $\combinatorialsolevaluated{\alpha}$ is an optimal solution for $\freeProblem$. 
\end{problem_description}
By exploiting  $\alpha$-nestedness (see~\cref{prop:alpha-nested:paper}), we show that the set $\mathfrak{K}$ has polynomial size.%
\begin{restatable}{proposition}{parametricSolutions}\label{prop:parametric-solutions:paper}
    Given a set $\polygons$ of interior-disjoint polygons in the plane with $V(\polygons)=n$, every solution~$\mathfrak{K}$ for $\biProblem$ has size $\bigO{n^2}$.
\end{restatable}
\begin{proof} [Proof sketch]
    Iterate over the solutions in $\mathfrak{K}$ in increasing order of the~$\alpha$ values for which they are optimal.
    We can show that each new solution either introduces a free arc for the first time or merges two regions. Taken together, both events happen at most $\bigO{n^2}$ times.
\end{proof}

To solve the parametric problem, we consider the objective values of geometric and combinatorial solutions as functions parameterized in~$\alpha$.
For a geometric solution~$\solution$, the parametric objective function given by~$g(\solution)(\alpha):=g_\alpha(\solution)$ is linear and its slope is the perimeter~$\PP(\solution)$.
For a combinatorial solution~$\combinatorialsol$, we define
\[
g(\combinatorialsol)(\alpha):=\begin{cases}
g_\alpha(\combinatorialsol[\alpha]) & \text{if } \alpha\in\mathcal{I}(\combinatorialsol),\\
\infty & \text{otherwise.}
\end{cases}
\]
It is easy to see that for all $\alpha \in \mathcal{I}(\combinatorialsol)$, the slope of $\combinatorialsol$ is monotonically decreasing. In particular, $g(\combinatorialsol[\alpha])$ is the tangent of $g(\combinatorialsol)$ at $\alpha$.
The objective of~$\biProblem$ can be stated as finding the lower envelope of the functions~$g(\combinatorialsol)$ for all possible combinatorial solutions~$\combinatorialsol$.

For the subdivision-restricted problem variant, the number of geometric solutions is finite and it can be solved using a \emph{dichotomic scheme}~\cite{rottmann2024bicritshapes}.
Given an interval~$[\amin,\amax]$ and optimal solutions~$\sollow$ at~$\amin$ and~$\solup$ at~$\amax$, it explores the interval by recursive bisection.
At the crossing point~$\anew$ of~$g(\sollow)$ and~$g(\solup)$, it computes an optimal solution~$\solnew$.
If~$\solnew$ is better than~$\sollow$ and~$\solup$ for~$\anew$, the algorithm adds~$\solnew$ to the solution set and recurses on the intervals~$[\amin,\anew]$ and~$[\anew,\amax]$.
Otherwise, it terminates with the solution set~$\{ \sollow, \solup \}$.

The correctness of the termination condition requires that a function cannot appear on the lower envelope more than once.
In our scenario, this is not immediately clear because the objective functions of the combinatorial solutions are not linear and they jump to~$\infty$ outside of the validity interval.
However, we can show that this is indeed the case by using another nesting property (cf.~\Cref{lemma:valid_solution_containment}) that relates realizations of combinatorial solutions.
\begin{restatable}{proposition}{lowerEnvelope}\label{lem:lower-envelope:paper}
    For every combinatorial solution~$\combinatorialsol$, the set of values~$\alpha\in[0,\infty]$ such that~$\combinatorialsol[\alpha]$ is optimal is a single interval.
\end{restatable}
{
Under the assumption that  we have an oracle which, given solutions $\combinatorialsol_1$,~$\combinatorialsol_2$, returns the exact intersection point of $g(\combinatorialsol_1)$ and~$g(\combinatorialsol_2)$, we present a polynomial-time algorithm for~$\biProblem$ in Appendix~\ref{sec:dichotomic:exact}.
In practice, we are satisfied with approximating the intersection point, as numerical precision is limited anyway.}

\subparagraph{An Approximation Scheme.}
\begin{algorithm2e}
    \caption{Approximative dichotomic scheme for~$\biProblem$.}\label{alg:chord:heuristic:new:paper}
    $\comblow \gets$ Combinatorial solution optimal for~$\lambda = 0$\;
    Report~$\comblow$\;
    \BlankLine
    $\combup \gets$ Combinatorial solution optimal for~$\lambda = 1$\;
    $\Recurse(\comblow, 0, \combup, 1)$\;
    Report~$\combup$\;
    \BlankLine
    \myproc{$\Recurse(\comblow, \lmin, \combup, \lmax)$}{
        \lIf{$\comblow =\combup$ \KwOr $f(\comblow[\lmin]) = f(\combup[\lmax])$\label{alg:chord:heuristic:new:identical:paper}}{\Return}
    $\lnew \gets (\lmax - \lmin)/2$\;
    $\combnew \gets$ Combinatorial solution optimal for $\lnew$\label{alg:chord:heuristic:new:optimize:paper}\;
    \tcp{Note $f(\combinatorialsol)(\lambda)=\infty$ if not valid}
    $\text{dom}_L\gets f(\combnew)(\lnew) < \min(f(\comblow)(\lnew), f(\comblow[\lmin])(\lnew)/(1+\varepsilon))$\;
    $\text{dom}_U\gets f(\combnew)(\lnew) < \min(f(\combup)(\lnew), f(\combup[\lmax])(\lnew)/(1+\varepsilon))$\;
    \lIf{$\text{dom}_L$}{
        $\Recurse(\comblow, \lmin, \combnew, \lnew)$
    }
    \lIf{$\text{dom}_L$ $\KwAnd$ $\text{dom}_U$}{Report~$\combnew$}
     \lIf{$\text{dom}_U$}{
        $\Recurse(\combnew, \lnew, \combup, \lmax)$
    }
    }
\end{algorithm2e}
\subparagraph*{Approximation Scheme.}
In this section, we define the notion of an approximate solution to the parametric problem, and give an approximation scheme for computing such a solution.%
\begin{definition}
Let $\varepsilon>0$. A $(1+\varepsilon)$-approximate solution is a set $\mathfrak{K}$ of combinatorial solutions such that for every $\alpha\in[0,\infty]$, there exists $\combinatorialsol\in\mathfrak{K}$ and a value $\alpha'$ such that $g_\alpha(\combinatorialsol[\alpha']) \leq (1+\varepsilon) \cdot \text{OPT}(\alpha)$, where~$\text{OPT}(\alpha)$ is the value of an optimal solution to~$\freeProblem$.
\end{definition}
To avoid dealing with the unbounded $\alpha$-domain, we switch to the equivalent objective function $f_\lambda(\solution) = (1-\lambda)\cdot \A(\solution) + \lambda \cdot \PP(\solution)$ with $\lambda \in [0,1]$. In contrast to the standard dichotomic scheme, our approach (given in~\cref{alg:chord:heuristic:new:paper}) does not compute an intersection point but instead performs a binary bisection with $\lnew = (\lmax + \lmin)/2$. Consequently, $\comblow$ and~$\combup$ generally do not have the same objective value at~$\lnew$, since it is not the crossing point. Therefore, the algorithm performs two independent comparisons and may recurse on only one side. Since binary bisection does not guarantee reaching the exact crossing point, we introduce an approximation test: After computing the optimal solution $\combnew$ at~$\lnew$, we compare it to the evaluated geometric solutions~$\comblow[\lmin]$ and~$\combup[\lmax]$ at~$\lnew$. If one of them is within a factor of $(1+\varepsilon)$ of~$\combnew$, we do not recurse on that side.

{We can show that the proposed algorithm is a fully polynomial-time approximation scheme~(FPTAS).
In addition to~$n$ and~$1/\varepsilon$, the runtime also depends on the smoothness of the function~$f(\combinatorialsol)$, which in turn depends on geometric properties of the input.
To bound this, let~$\delta_\text{seg}(\polygons)$ be the minimum distance between any vertex~$v$ and any polygon edge not incident to~$v$.
Intuitively, this is a lower bound on the thickness of each polygon and the minimum distance between any pair of polygons.
The \emph{segment spread} $\Phi := \frac{\text{diam}(V(\polygons))}{\delta_\text{seg}(\polygons)}$ represents the ratio between the maximum and minimum distance/thickness in the input.}

\begin{theorem}
\Cref{alg:chord:heuristic:new:paper} computes a $(1+\varepsilon)$-approximate solution to~$\biProblem$ in time
\begin{align*}
\mathcal O\!\left(n^6 \frac{\log^3 n}{\log^2 \log n} (\log n + \log \Phi + \log \tfrac{1}{\varepsilon})\right).
\end{align*}
\end{theorem}

\subparagraph*{Exploiting $\alpha$-Nestedness.}
Let~$\sollow:=\comblow[\amin]$, $\solnew:=\combnew[\anew]$ and~$\solup:=\combup[\amax]$.
Due to the~$\alpha$-nestedness of optimal solutions, we have
$
\sollow
\subseteq 
\solnew
\subseteq 
\solup.
$
For the subdivision-restricted case, this can be exploited by contracting the problem instance that is solved for~$\anew$~\cite{beines2024}.
This optimization does not carry over to~$\freeProblem$, but we can still exploit~$\alpha$-nestedness:
\begin{enumerate}
    \item Since~$\sollow \subseteq \solnew$ and $\optimizer$ can handle circular polygons as inputs, we can run the optimization at~$\anew$ with~$\sollow$ as the input instance instead of $\polygons$.
    Particularly for larger~$\alpha$ values, this can significantly reduce the complexity of the input.
    \item  Let $\polygons(S)$ denote the polygons contained in a region $S$.
    It follows from~$\solnew \subseteq \solup$ that for every region~$S \in \solnew$, there exists a region $S' \in \solup$ with $\polygons(S) \subseteq \polygons(S')$.
    Therefore, $\solnew$ can be computed by solving the subproblem with input~$\polygons(S')$ for each region~$S' \in \solup$ independently. This allows us to discard arcs between different regions of~$\solup$. 
    \item Now, consider an arbitrary but fixed region $S_L \in \sollow$, and let $S_U \in \solup$ be the region such that $S_L \subseteq S_U$. If $S_L$ and $S_U$ contain the same polygons, then we can find the corresponding region~$S_N \in \solnew$ by solving a simplified sub-instance for every free arc~$f=\sarc{}{xy}$ in~$S_U$ (see~\Cref{fig:meshes:paper}).
    Let~$s=\langle x=v_0,\dots,v_k=y\rangle$ be the sequence of combinatorial vertices between~$u$ and~$v$ on~$\partial S_L$.
    The boundary segment between~$x$ and~$y$ in~$S_N$ must lie fully within the region enclosed by~$f$ and~$s$, and the only vertices within this region are those in~$s$. 
    Thus, it suffices to consider arcs between vertices in~$s$ for this sub-instance.
\begin{figure}
    \centering 
    \includegraphics[width = 0.45\textwidth]{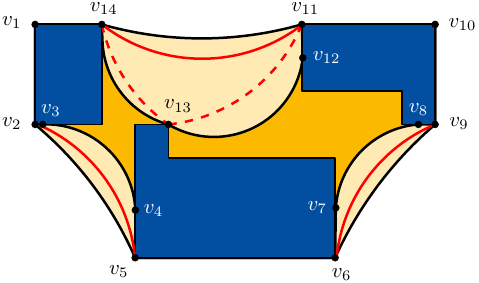}
    \caption{Input polygons are shown in blue, $\sollow$ in dark yellow and $\solup$ in light yellow. By exploiting $\alpha$-nestedness, we can isolate the vertex sequences $\langle v_2,v_3,v_4,v_5 \rangle$, $\langle v_6,v_7,v_8,v_9 \rangle$, and $\langle v_{11},v_{12},\dots,v_{14} \rangle$ of~$\sollow$ as sub-instances, where arcs in $\solnew$ may only connect vertices from the respective sequence. }
    \label{fig:meshes:paper}
\end{figure}
\end{enumerate}

\section{Experiments}\label{sec:experiments}
We implemented our algorithms in \CC~and compiled the code using GCC 13.3.0 with the \texttt{-O3} optimization flag. For geometric computations, we used the CGAL library~\cite{CGAL}. Minimum cuts were computed with the Boykov-Kolmogorov max-flow algorithm~\cite{BK04}, using the implementation from the Boost library~\cite{boost}. All experiments were run sequentially on a single core of an AMD Threadripper 3970X machine with 128\,GB of DDR4-RAM.
As input instances, we used building footprints of German towns and cities extracted from OpenStreetMap~\cite{OSM} (cf.~\Cref{tab:single_parameter_runs:paper}) with input sizes ranging from $1\,800$ to $400\, 000$ input vertices.
All distances are measured in decimeters ($10$ centimeters), which means that the arcs in an optimal solution for~$\freeProblem$ have a radius of~$\alpha$ decimeters.

\subsection{Fixed Parameter Problem}
\begin{figure}[bt]
    \centering
    \begin{minipage}[t]{0.48\textwidth}
        \centering
        \includegraphics[width=0.9\textwidth]{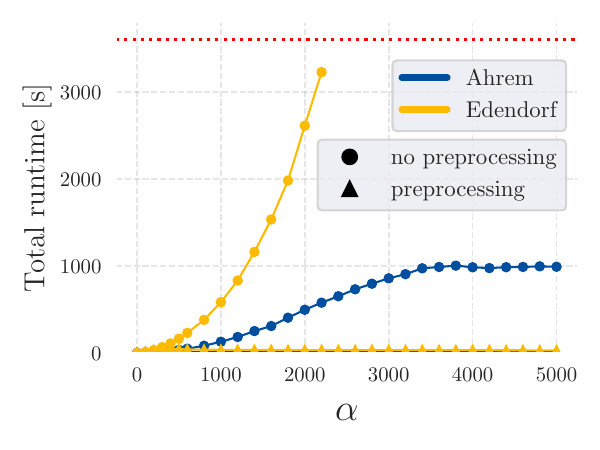}
        \caption{Runtime comparison with and without preprocessing. Runs that exceeded the time limit of~1 hour are omitted.}
        \label{fig:runtimes_speedups:paper}
    \end{minipage}\hfill
    \begin{minipage}[t]{0.48\textwidth}
        \centering
        \includegraphics[width=0.9\textwidth]{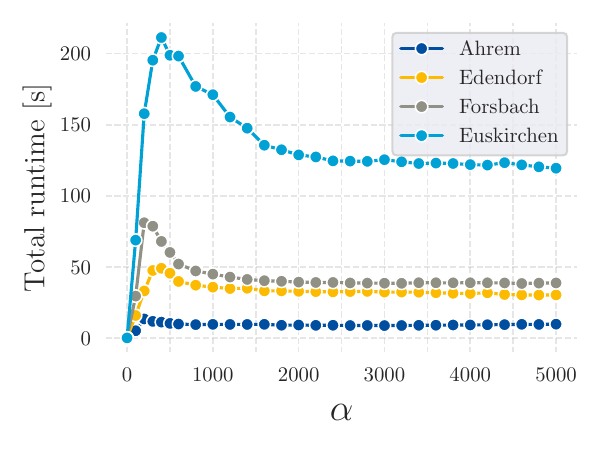}
        \captionsetup{textformat=simple} \caption{Runtime by~$\alpha$ value, with preprocessing enabled.}
        \label{fig:runtimes_all:paper}
    \end{minipage}
\end{figure}%
\begin{figure}[th]
    \centering
    \includegraphics[width=0.95\textwidth]{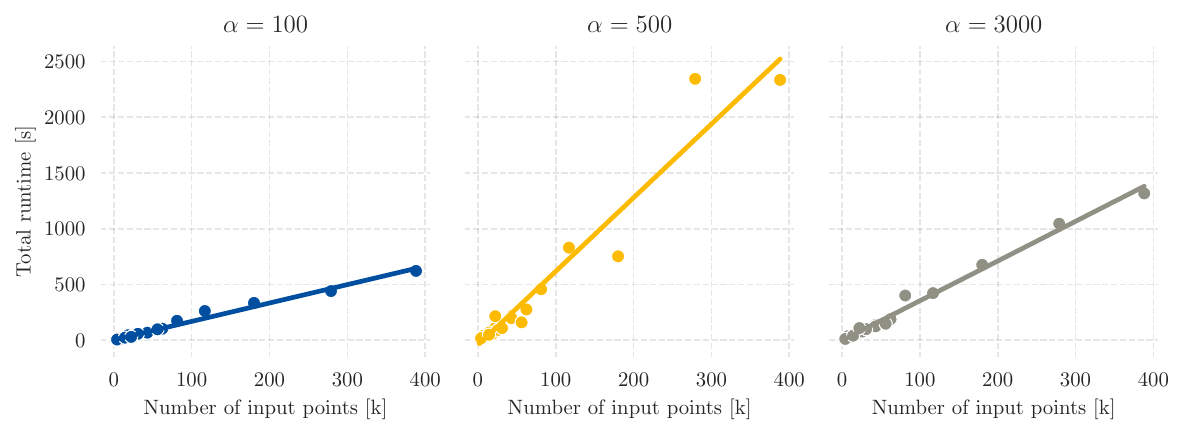}
    \caption{Runtime with preprocessing on all instances for selected $\alpha$-values.}
    \label{fig:total_time_all_alpha:paper}
\end{figure}%
\Cref{fig:runtimes_speedups:paper} shows the impact of the local preprocessing on the instances Ahrem ($|V(\polygons)|=3\,683$) and Edendorf ($|V(\polygons)|=10\,536$). Without preprocessing, the runtime grows quickly with~$\alpha$. For Ahrem, it plateaus around $\alpha=4000$. This is because the set $\allowedArcs$ of free arcs stops growing, since larger endpoint distances make circular arcs likely to intersect other input polygons. However, already for Edendorf, the one-hour time limit is reached before this plateau. With preprocessing, the runtime remains in the range of seconds for all~$\alpha$ values.%

\Cref{fig:runtimes_all:paper} shows only the runtime with preprocessing enabled, now including two additional medium-sized instances.
Furthermore, \Cref{fig:total_time_all_alpha:paper} plots the runtime on all instances for selected~$\alpha$ values. The runtime scales linearly with the instance size and, for a fixed instance, varies by less than a factor of~$4$ depending on $\alpha$. For all instances, the runtime peaks at some $\alpha\in[200,500]$ and then declines, reaching a plateau around $\alpha=3000$; see \cref{fig:edendorf_example:paper} for example aggregations.
\begin{figure}[bt]
    \centering
    \begin{minipage}{0.32\textwidth}
        \centering
        \includegraphics[width=\linewidth]{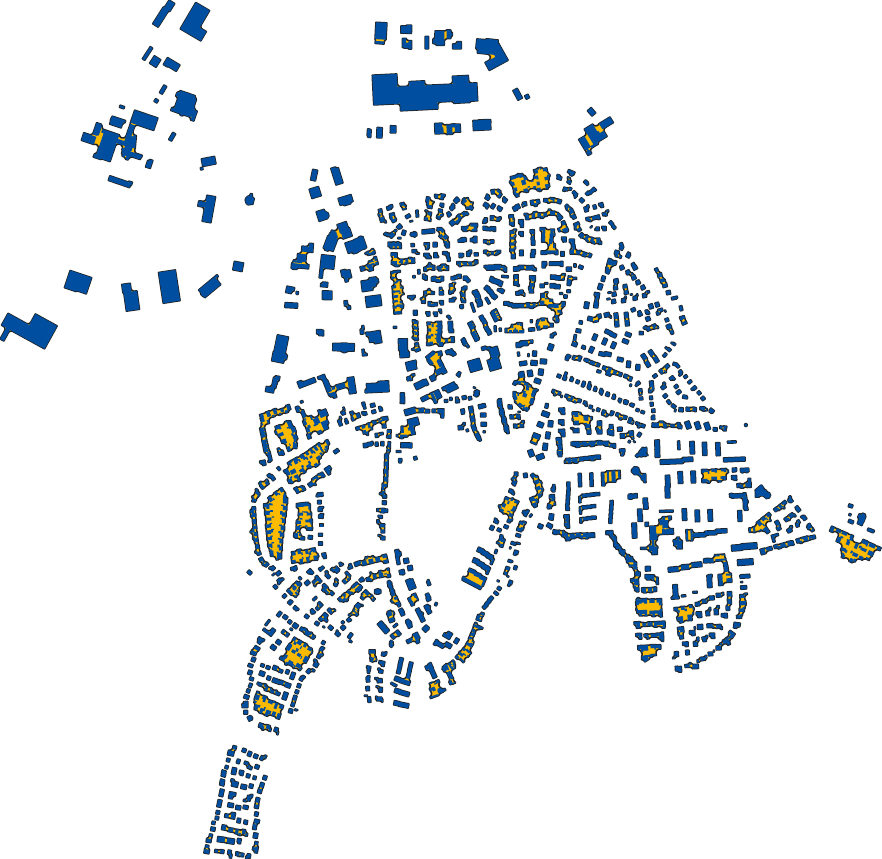}
    \end{minipage}
    \hfill
    \begin{minipage}{0.32\textwidth}
        \centering
        \includegraphics[width=\linewidth]{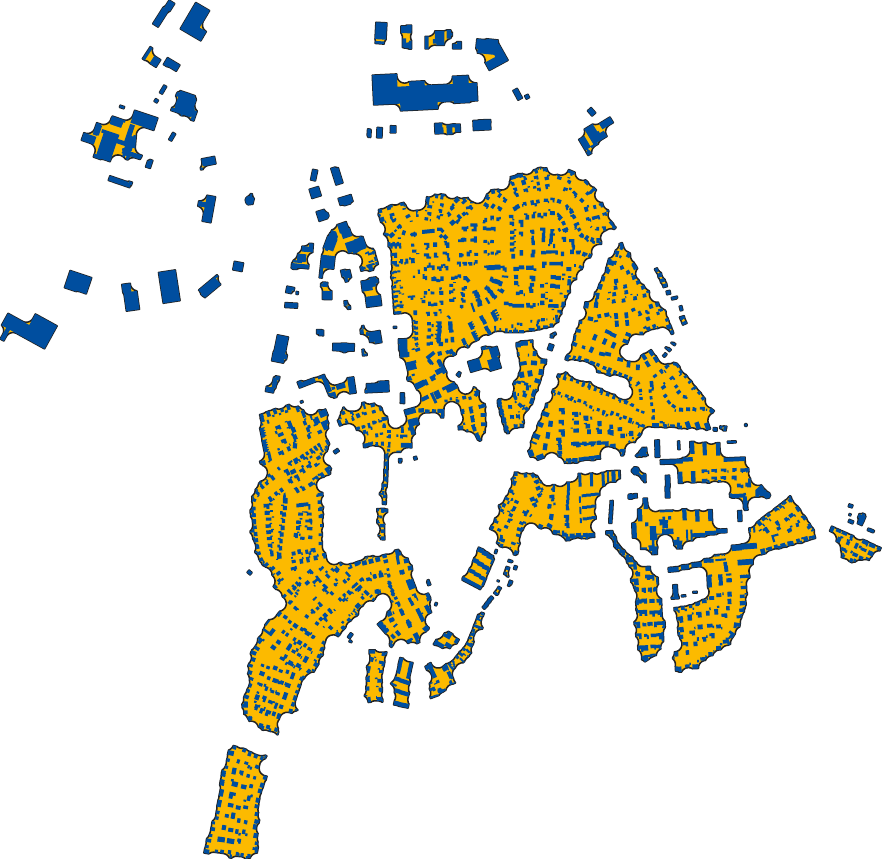}
    \end{minipage}
    \hfill
    \begin{minipage}{0.32\textwidth}
        \centering
        \includegraphics[width=\linewidth]{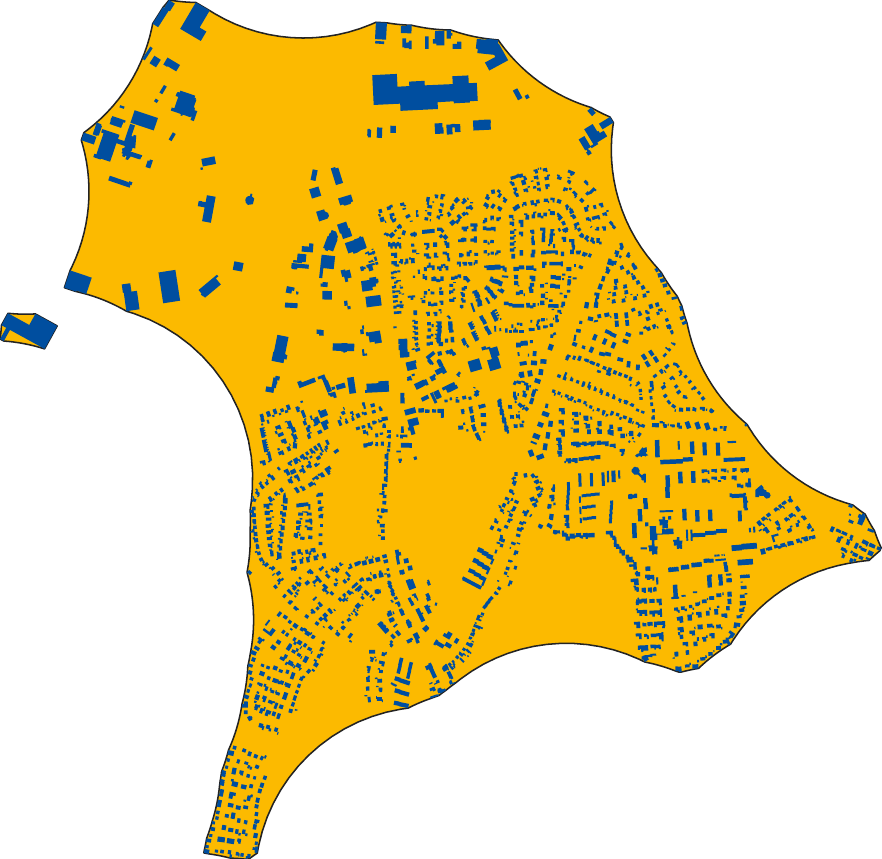}
    \end{minipage}
    \caption{Instance Edendorf. Left: For $\alpha=100$, few polygons are aggregated. Middle: For $\alpha=260$, there are multiple clusters. Right: For $\alpha=3000$, we have one large cluster and one outlier.}
    \label{fig:edendorf_example:paper}
\end{figure}
Next, we highlight some noteworthy observations; for detailed measurements and an in-depth discussion, see \cref{sec:appendix:Fixed_parameter_experiments}.
Overall, the preprocessing drastically reduces the number of generated arcs, especially for large $\alpha$, where the reduction reaches orders of magnitude. On all instances, the size of the subdivision for an $\optimizer$ call during preprocessing lies between $4$ and $7$, so most calls are nearly trivial to solve.
Finally, we investigate why the runtime peaks for medium~$\alpha$ values. For $\alpha \ge 3000$, most merges succeed, and the preprocessing yields a solution that is close to or identical to the final optimal solution. For $\alpha \leq 200$, only few arcs are generated due to the distance threshold; however, the effectiveness of the preprocessing also decreases, as only few merges succeed.
In the range $\alpha\in[200,500]$, solutions are the most volatile: some polygons already merge into large representatives, but still many merges can fail, resulting in larger subdivisions during the preprocessing, and in particular for the final $\optimizer$ call after the preprocessing.

\subsection{Parametric Problem}
\begin{figure}[tbp]
    \centering
    \begin{minipage}{0.48\textwidth}
        \includegraphics[width=\linewidth]{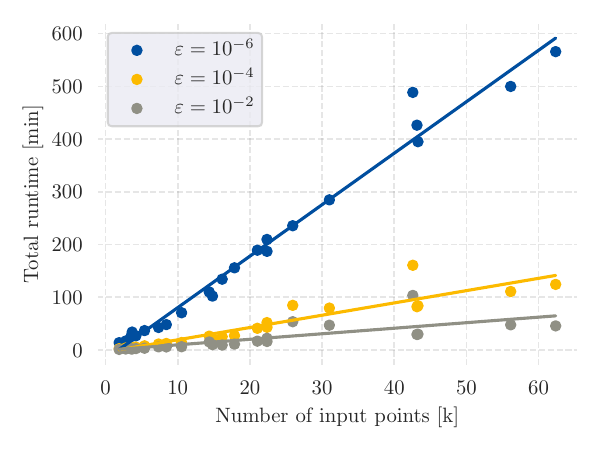}
    \end{minipage}
    \hfill
    \begin{minipage}{0.48\textwidth}
        \includegraphics[width=\linewidth]{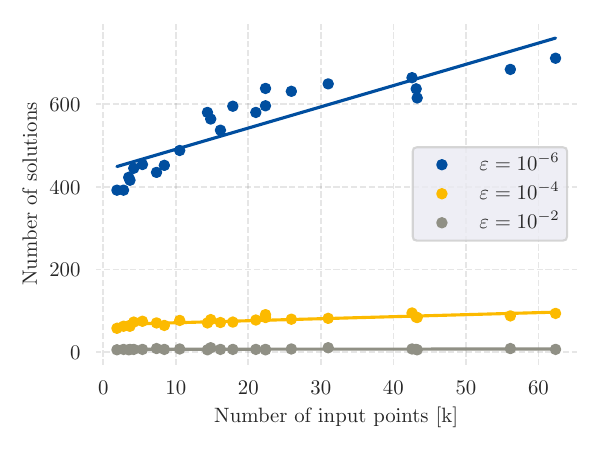}
    \end{minipage}    
    \caption{ Runtime and $\#\,$solutions of the approximation scheme on medium-sized instances.}
    \label{fig:solutions_and_runtime:paper}
\end{figure}%

\Cref{fig:solutions_and_runtime:paper} shows the runtime and number of solutions of the engineered approximation scheme on all instances with fewer than $100\,000$ vertices for $\varepsilon\in\{10^{-2},10^{-4},10^{-6}\}$. Despite the theoretical $\bigO{n^2}$ upper bound on the number of solutions, the algorithm produces significantly fewer than $n$ solutions, and the solution size scales linearly for fixed $\varepsilon$. Independent of $\varepsilon$, the runtime also scales linearly with input size, albeit with large constants.
\Cref{fig:recursive_calls:paper} show that the number of~\Recurse~calls per reported solution is at most~$4$ on all instances, indicating that the approximation of the crossing point converges quickly in practice and that the runtime is essentially proportional to the number of solutions.

\Cref{tab:parametric_results:paper} reports the performance of the approximation scheme with partially disabled $\alpha$-nestedness exploits on the smaller Edendorf and Andernach datasets with $\varepsilon = 10^{-6}$.  Exploiting $\alpha$-nestedness from above and below individually reduces the runtime by a factor of at most~$2$, while combining both and adding the decomposition-based optimization yields speedups of $5.9$ for Edendorf and $7.1$ on Andernach.
\begin{table*}[bt]
\centering
\begin{minipage}{0.4\textwidth}

\includegraphics[width=\textwidth]{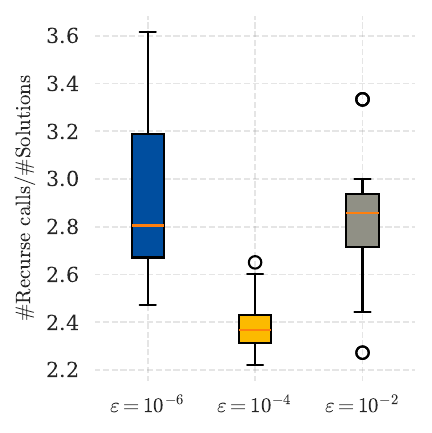}
 \captionsetup{textformat=simple}\captionof{figure}{
Number of recursive calls per reported solution across all instances, grouped by $\varepsilon$ values.
}
\label{fig:recursive_calls:paper}

\end{minipage}
\hfill
\begin{minipage}{0.55\textwidth}
\centering

\caption{
Performance of the approximation scheme for~$\varepsilon = 10^{-6}$.
Subset and Superset indicate whether we exploit $\alpha$-nestedness from below and above, respectively.
}
\label{tab:parametric_results:paper}
\begin{tabular}{lccr}
\toprule
Dataset & Subset & Superset & {Runtime $[\text{min}]$} \\
\midrule

\multirow{4}{*}{Edendorf} 
 & $\circ$ & $\circ$ & 398.48 \\
 & $\circ$ & $\bullet$ & 323.02 \\
 & $\bullet$ & $\circ$ & 217.12 \\
 & $\bullet$ & $\bullet$ & 67.14 \\ 
 [5pt]

\multirow{4}{*}{Andernach} 
 & $\circ$ & $\circ$ & 1246.47 \\
 & $\circ$ & $\bullet$ & 1005.51 \\
 & $\bullet$ & $\circ$ & 890.64 \\
 & $\bullet$ & $\bullet$ & 176.49 \\ 

\bottomrule
\end{tabular}

\end{minipage}

\end{table*}
\section{Conclusion}
We showed that the parametric polygon aggregation problem can still be solved in polynomial time if the subdivision-based restriction~\cite{rottmann2024bicritshapes} is lifted.
With careful engineering that exploits structural properties, our approach also achieves linear runtime on real-world data. 
While the constant factors are higher, it can still process city-sized instances (see Figure~\ref{fig:bonn_aggregated}) within a timeframe of a few hours.
Investing the extra computation time can be beneficial for at least two reasons.
Firstly, situations arise in practice where the subdivision-restricted approach fails to find a sufficiently representative solution because the subdivision is ``in the way'', especially at small scales (cf.~\Cref{appendix:subdivisioncomparison}).
More importantly, removing the restriction guarantees subset-nestedness, which is a structural advantage that ensures consistency across solutions, not only across parameter values but also under the addition or filtering of input polygons. 
This makes the approach particularly well-suited for applications in map generalization and urban analytics where consistency is crucial, for example, when filtering by temporal attributes (e.g., construction dates) or by building types. Here, it ensures that resulting clusterings remain consistent with the entire input, in the sense that more specific subsets refine, rather than contradict, the overall clustering. Beyond this, our preprocessing scheme may be of independent interest for other fencing problems that exhibit subset-nestedness.

\begin{figure}[!ht]
    \centering 
    \includegraphics[width=0.8\textwidth]{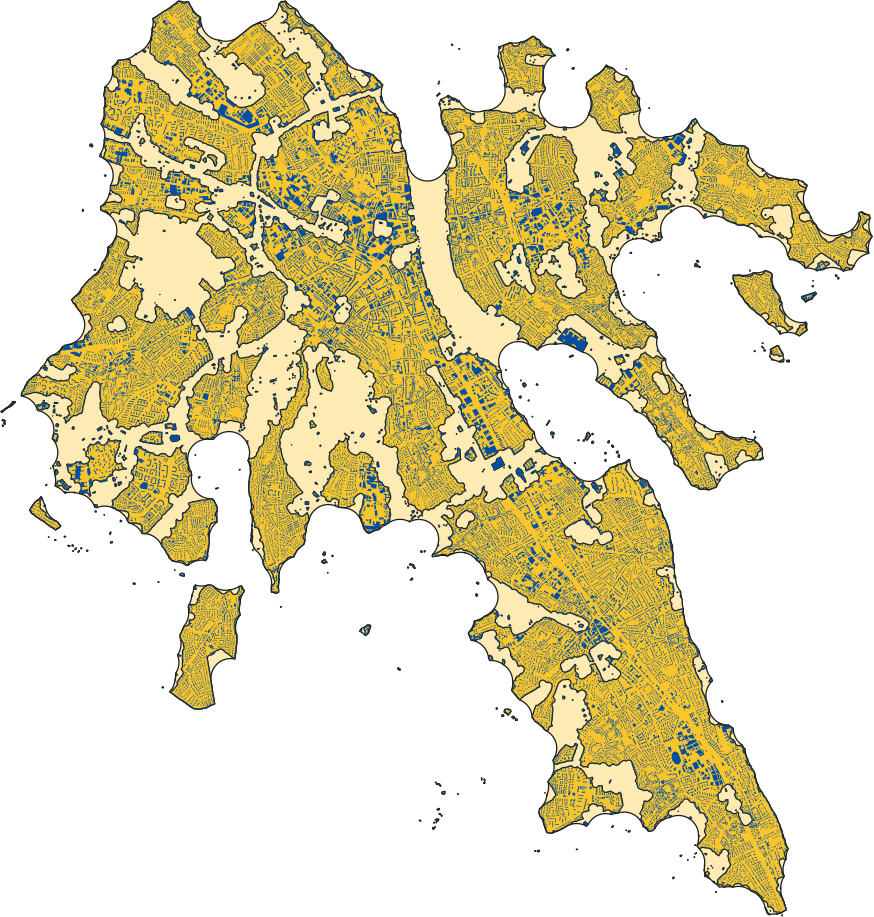}
    \caption{Aggregations of the city of Bonn, Germany, for $\alpha = 500$ and $\alpha = 3000$.}\label{fig:bonn_aggregated}
\end{figure}%

\bibliography{literature}
\appendix
\newpage
\section{A Remark on $\alpha$-Shapes}
Unlike our approach, $\alpha$-shapes tend to connect nearby point sets with long, narrow bridges, which are undesirable in the context of map generalization~\cite{bone2019,rottmann2024bicritshapes}.
Consider the example in~\Cref{fig:alphashapes}.
Here, a narrow bridge appears as~$\alpha$ increases, causing both the area and the perimeter of the solution to increase.
Thus, in terms of our objective function $g_\alpha$, the solution becomes worse in both relevant criteria.

\label{sec:alpha_shapes}

\begin{figure}[htbp]
  \centering
  \begin{subfigure}[t]{0.48\textwidth}
    \centering
    \fbox{\includegraphics[width=0.95\textwidth]{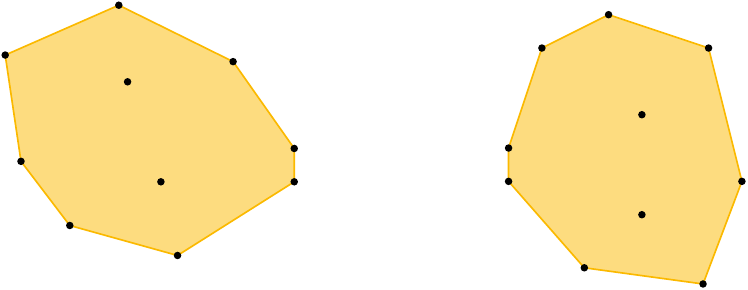}}
    \label{fig:alpha1}
  \end{subfigure}
  \hfill
  \begin{subfigure}[t]{0.48\textwidth}
    \centering
    \fbox{\includegraphics[width=0.95\textwidth]{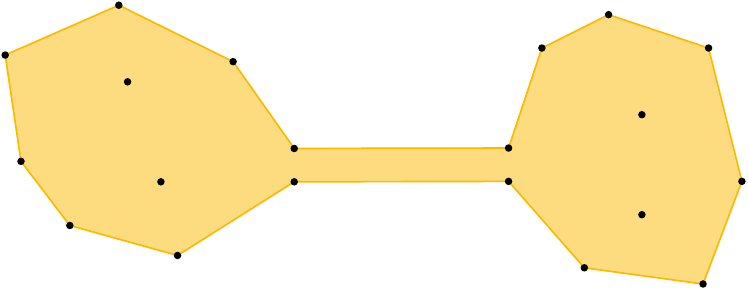}}
    \label{fig:alpha2}
  \end{subfigure}
\caption{Given the depicted point set, the $\alpha$-shape for some value of $\alpha$ is shown on the left. As $\alpha$ increases, a bridge forms (right), increasing the perimeter.}
  \label{fig:alphashapes}
\end{figure}

\section{The Minimum-Cut Approach for Polygon Aggregation}
\label{sec:Rottmann_transformation}
\begin{figure}[htbp]
    \centering
    \begin{minipage}{0.48\textwidth}
        \centering
        \includegraphics[width=\linewidth,page=1]{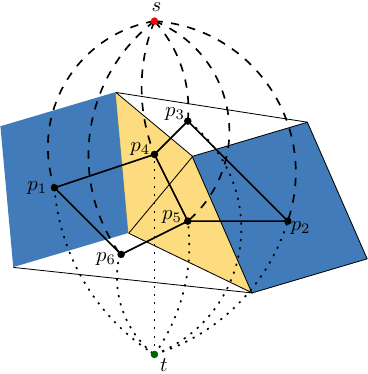}
        \subcaption{Single-source-single-sink graph~$G_{st}$}
        \label{fig:rottman_a}
    \end{minipage}
    \hfill
    \begin{minipage}{0.48\textwidth}
        \centering
        \includegraphics[width=\linewidth,page=2]{images/Rottmann.pdf}
        \subcaption{Planar multi-sink-multi-source graph~$G_\text{M}$}
        \label{fig:rottman_b}
    \end{minipage}
    \caption{Illustrations of the two transformations of~$\problem{\subdivision}$ into minimum cut problems on graphs given in this section.}
    \label{fig:rottman}
\end{figure}
Rottmann et al.~\cite{rottmann2024bicritshapes} give a transformation of the subdivision-restricted polygon aggregation problem~$\problem{\subdivision}$ into a minimum~$(s,t)$-cut problem on an almost-planar graph.
In this section, we give the details of this transformation.
Additionally, we give an alternative transformation, analogous to~\cite{DBLP:journals/dcg/AbrahamsenGLR20}, that yields a planar graph with multiple source and sink vertices. On this graph, faster algorithms for planar multi-source-multi-sink minimum cut can be used.

\subparagraph*{Original Transformation.} Let $\subdivision$ be a planar subdivision of $\conv(\polygons)\setminus\polygons$. 
Then $\mathcal{A}=\subdivision \cup \polygons$ is a planar subdivision of $\conv(\polygons)$.
Let~$G$ denote the geometric dual graph of~$\mathcal{A}$. Here, two vertices $p_i,p_j\in V(G)$ are connected by an edge if the corresponding faces $a_i,a_j\in\mathcal{A}$ share at least one boundary edge.
In the augmented graph~$G_{st}$, we introduce an additional source vertex $s$ and sink vertex~$t$, as well as edges $\{s, p_i\}$ and $\{p_i, t\}$ for every vertex $p_i\in V(G)$. The graph construction is depicted in \Cref{fig:rottman_a}. Let $b(a_i)$ denote the length of the boundary shared by a face $a_i$ and the outer face and let $b(a_i,a_j)$ denote the length of the boundary shared by $a_i$ and $a_j$. We define edge weights $w\colon E(G_{st})\rightarrow\mathbb{R}_{\geq 0}$ as follows:

\begin{itemize}
    \item For $e=\{p_i,p_j\}$ with~$p_i,p_j \in V(G)$, we set 
    \(\begin{aligned}
        w(e) &= \alpha \cdot b(a_i,a_j).
    \end{aligned}\)
 \item For $e=\{p_i,t\}$ with~$p_i \in V(G)$, we set $\;\;w(e) = \A(p_i) + \alpha \cdot b(a_i)$
    \item For $e=\{s,p_i\}$ with~$p_i \in V(G)$, we set 
    \(\begin{aligned}
        \;\;w(e) &= 
        \begin{cases}
            \infty & \text{if } a_i \in \polygons, \\[6pt]
             0 & \text{otherwise}.
        \end{cases}
    \end{aligned}\)    
\end{itemize}
Note that Rottmann et al.~\cite{rottmann2024bicritshapes} use the objective function~$f_\lambda(\solution)=\lambda \A(\solution) + (1-\lambda)\PP(\solution)$, whereas we use the equivalent function~$g_\alpha(\solution) = \A(\solution) + \alpha\PP(\solution)$.
Accordingly, the edge weights in our construction differ slightly.
A solution to the subdivision-restricted polygon aggregation problem~$\problem{\subdivision}$ is a selection of faces~$\solution \subseteq \mathcal A$ with~$\polygons \subseteq \solution$.
Let~$V(\solution)$ denote the corresponding set of vertices in~$G$.
Rottmann et al.\ show that~$g_\alpha(\solution)$ is exactly the value of the~$(s,t)$-cut with source component~$(\{ s \} \cup V(\solution)$ and sink component~$V(G) \setminus V(\solution) \cup \{ t\}$.
Using the fastest known algorithm for minimum $(s,t)$-cut in general graphs, they obtain the following result.
\begin{theorem}[Rottmann et al.~\cite{rottmann2024bicritshapes}]
For any minimum $(s,t)$-cut in $G_{st}$, the set of cells in~$\mathcal{A}$ corresponding to the source component is an optimal solution to $P^\alpha_\mathcal{D}$. 
Hence, $P^\alpha_\mathcal{D}$ can be solved in $O\!\left(\tfrac{|\mathcal{A}|^2}{\log |\mathcal{A}|}\right)$ time.
\end{theorem}

\subparagraph{Planar Transformation.}
The addition of~$s$ and~$t$ does not preserve the planarity of~$G$, preventing us from using faster min-cut algorithms for planar graphs.
Therefore, we construct an alternate augmented graph~$G_\text{M}$ in which~$s$ and~$t$ are replaced with sets~$S = \{ s_i \mid p_i \in V(G) \}$ and~$T = \{ t_i \mid p_i \in V(G) \}$, i.e., every vertex~$p_i \in V(G)$ is associated with a separate source~$s_i$ and sink~$t_i$.
The corresponding edges~$\{s_i,p_i\}$ and~$\{p_i,t_i\}$ are given the same weights as~$\{s,p_i\}$ and~$\{p_i,t\}$ in~$G_{st}$, respectively.
We call a cut in~$G_\text{M}$ an~$(S,T)$-cut if all vertices in~$S$ lie in the source component and all vertices in~$T$ in the sink component.
By construction, every~$(s,t)$-cut in~$G_{st}$ has a corresponding~$(S,T)$-cut in~$G_\text{M}$ and vice versa.
Furthermore, a planar embedding of~$G$ can be extended to a planar embedding of~$G_\text{M}$: for each vertex~$p_i$, we place the vertices~$s_i$ and~$t_i$ inside the face~$a_i$ (see \cref{fig:rottman_b}).
Thus, we can apply algorithms for planar multi-source-multi-sink minimum cut.
The fastest known algorithm for this problem is due to Borradaile et al.~\cite{DBLP:journals/siamcomp/BorradaileKMNW17} and runs in time $ \mathcal{O}\left(|\mathcal{A}| \frac{\log^3|\mathcal{A}|}{\log^2 \log |\mathcal{A}|}\right)$ when combined with the data structure by Gawrychowski and Karczmarz~\cite{DBLP:conf/icalp/GawrychowskiK18} for shortest paths in dense distance graphs.
\begin{theorem}
\label{theorem:multisorucesink}
For any minimum $(S,T)$-cut in $G_\text{M}$, the set of cells in~$\mathcal{A}$ corresponding to the source component is an optimal solution to $P^\alpha_\mathcal{D}$. 
Hence, $P^\alpha_\mathcal{D}$ can be solved in  $ \mathcal{O}\left(|\mathcal{A}| \frac{\log^3|\mathcal{A}|}{\log^2 \log |\mathcal{A}|}\right)$ time.
\end{theorem}

\section{Solving the Unrestricted Aggregation Problem (Including Proofs)}
\label{sec:unrestricted:appendix}
In this section, we investigate the unrestricted aggregation problem~$\freeProblem$. 

\begin{remark}
    Most results for (circular) polygons extend to more general input regions with piecewise differentiable boundary curves. Therefore, in this extended section, we present the characterization of free boundary pieces for this broader class of regions.
\end{remark} 

After giving the characterization of the free boundary pieces, we show that for the special case where the input regions are polygons or circular polygons formed by straight lines and circular arcs, only a polynomial number of curves need to be inspected to find an optimal solution. We then use these curves to build a subdivision $\subdivision_C$ of $\conv(\polygons)\setminus\polygons$ that has polynomial size, thereby reducing $\freeProblem$ to the subdivision-restricted problem~$\problem{\subdivision_C}$ already discussed in~\cite{rottmann2024bicritshapes}.
This allows us to apply the minimum cut-based algorithm proposed in~\cite{rottmann2024bicritshapes} to solve the unrestricted problem in polynomial time.
Finally, we note that the optimal solutions to~$\freeProblem$ are nested with respect to~$\alpha$ and with respect to introducing new input polygons.

{We start by characterizing the optimal solutions for the (simpler) corner cases~$\alpha=0$ and $\alpha=\infty$. For~$\alpha=0$, the solution that consists only of the input regions is optimal since it minimizes the area. For~$\alpha=\infty$, where the objective is to minimize the perimeter, we show that all solution regions are convex hulls.}
\begin{lemma}\label{lem:infinity}
    {For an instance~$\polygons$ of~$P_F^\infty$, every  region $S$ of an optimal solution~$\solution$ is given by the convex hull of the input regions $\polygons(S)$ contained in $S$.}
\end{lemma}
\begin{proof}
Assume for the sake of contradiction that~$S \neq \conv(\polygons(S))$.
Then the perimeter satisfies $P(\conv(\polygons(S))) < P(S)$.
Let~$\mathcal S' \subset \solution$ be the set of regions that intersect~$\conv(\polygons(S))$ properly, and let~$X:= \bigcup_{S' \in \mathcal S'} S'$ be the union of those regions.
For the solution that replaces~$S$ and~$\mathcal S'$ with~$\conv(\polygons(S)) \cup X$, we observe that
\[
P(\conv(\polygons(S)) \cup X) < P(\conv(\polygons(S))) + P(X) < P(S) + P(X),
\]
since overlapping boundary parts between $X$ and $\conv(S)$ are removed, and no new boundary is introduced. Therefore, $\solution$ cannot be optimal.
\end{proof}

{Next, we discuss the general case~$\alpha \in (0,\infty)$.
We start by noting that this can be reduced to the case~$\alpha = 1$ by scaling the input appropriately.}

\observationScale*
\begin{proof}
Scaling both the input and the solution yields a valid problem instance and a valid solution for the scaled instance. The lengths of all curves increase by a factor of $\frac{1}{\alpha}$, and all areas increase by a factor of $\frac{1}{\alpha^2}$. Moreover, the objective value increases by $\alpha^2$ due to 
\begin{align*}
    g_\alpha(\solution)=\A(\solution)+\alpha\PP(\solution)
     &=\alpha^2\A(\sigma_{1/\alpha}(\solution))+\alpha^2\PP(\sigma_{1/\alpha}(\solution))=\alpha^2g_1(\sigma_{1/\alpha}(\solution)).\qedhere
\end{align*}
\end{proof}
Consequently, we only consider the cost function $g_1$ for the proofs in this section.
\Cref{obs:scale:paper} describes a crucial property of the problem formulation: the parameter~$\alpha$ corresponds directly to the scale at which the input map is viewed, in the sense that multiplying~$\alpha$ by some factor~$c$ is equivalent to scaling the map by a factor of~$c$.
It follows that the set of solutions that are optimal for at least one value of~$\alpha$ does not change when scaling the input.

We continue by investigating which curves locally optimize the area-perimeter trade-off.
\efficiency*
If we consider a solution that uses~$\overrightarrow{uv}$ as a free boundary piece such that the enclosed area is to the left of it, then~$e_{uv}(f)$ is the improvement in the objective value if we replace~$\overrightarrow{uv}$ with~$f$ (assuming that this yields a feasible solution).
We now characterize the curves that have maximum efficiency and show that any solution that uses a curve of non-maximum efficiency can be improved locally.

\begin{figure}
    \centering
    \includegraphics{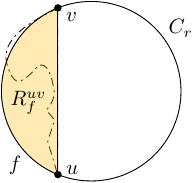}
    \caption{Visualization of~\Cref{lemma:arcs_enclose_much:paper}. For fixed $\LL(f)$, the enclosed area is maximized if~$f$ is circular.}
    \label{fig:ArcIsBest:full}
\end{figure}
\efficiencyCircular*
\begin{proof}
    Let $C_r$ be the circular arc with radius $r$ from $u$ to $v$ in the half-plane to the left of~$\overrightarrow{uv}$ 
    such that $\LL(C_r)+\LL(f)=2\pi r$ (see~\cref{fig:ArcIsBest:full}). 
    The curve~$f$ lies to the left of~$\overrightarrow{uv}$ since the efficiency is non-positive otherwise. Hence, $C_r$ together with $f$ is a closed curve with perimeter $P=2\pi r$. The isoperimetric inequality~\cite{Steiner1838, hurwitzisoperimetric} states that for every closed curve~$C$ with enclosed region~$R$, it holds that $4\pi \A(R)\leq \LL(C)^2$. Furthermore, equality holds if and only if $C$ is a circle.
    It follows that $C_r$ together with $f$ has to be a circle to maximize the enclosed area and in turn maximize the efficiency. Consequently, $f$ is a circular arc.
\end{proof}

\Cref{lemma:arcs_enclose_much:paper} restricts the curve~$f$ \done{such that it does not} intersect the line through~$u$ and~$v$ properly. We will see later (from~\Cref{lemma:local_structure:paper}) that the claim holds without this restriction.

In the remainder of this paper, we denote by $\sarc{r}{uv}$ the circular arc of radius $r$ and central angle $\theta \leq \pi$ that connects the points $u$ and $v$ and lies to the left of~$\overrightarrow{uv}$.
An analytical argument shows that the most efficient circular arcs have radius~$1$.

\begin{restatable}{lemma}{efficiencyRadius}
\label{lemma:arcs_have_radius_one}
    Let $u,v$ be two points with~$\LL(\overline{uv})<2$. Among all circular arcs from~$u$ to~$v$ with central angle $\theta<\pi$ that lie to the left of~$\overrightarrow{uv}$, the arc $\sarc{1}{uv}$ has the largest efficiency.
\end{restatable}
\begin{proof} 
\begin{figure}
    \centering
    \includegraphics{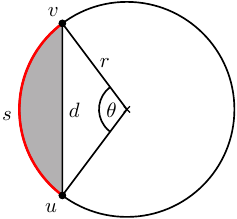}
    \caption{Visualization of the proof of~\Cref{lemma:arcs_have_radius_one}.}
    \label{fig:P2:full}
\end{figure}
Let $s_r$ be the length of the circular arc $\sarc{r}{uv}$ with radius $r$, let $\A_r$ be the area enclosed by $\sarc{r}{uv}$ and $\overline{uv}$ (gray area in \Cref{fig:P2:full}), and let $d=\LL(\overline{uv})$. Further,  $\theta=2\arcsin(d/2r)$ is the central angle of $\sarc{r}{uv}$. Since $d$ is fixed, the efficiency of the circular arc is a function~$h(r):=e_{{uv}}(\sarc{r}{uv})$ depending only on $r$. We have 
\begin{align*}
   h(r)&= -s_r+d+\A_r\\
        &=   -\theta r +d+ \left(\frac{r^2}{2} \theta- \frac{d}{2}\sqrt{r^2-(d/2)^2}\right)\\
        &= -2\arcsin\left(\frac{d}{2r}\right)r+ d+\left( \frac{r^2}{2} 2\arcsin\left(\frac{d}{2r}\right)- \frac{d}{2}\sqrt{r^2-(d/2)^2}\right)\\
        &=-\arcsin\left(\frac{d}{2r}\right)\cdot (2 r-r^2)-\frac{rd}{2}  \sqrt{1- (d/2r)^2 }    +  d.
\end{align*}
To find the maximum of this function, we compute its derivative
\begin{align*}
    \frac{\text{d}h}{\text{d}r}(r) =\frac{2(1-r)\left(d-r\arcsin(\frac{d}{2r}) \sqrt{4-(\frac{d}{r})^2} \right) }{r\sqrt{4-(\frac{d}{r})^2}}.
\end{align*}
A root of this function is at $r=1$. Note that for $r<d/2$, we cannot connect $\overline{uv}$ by a circular arc with radius $r$ and that for $r=d/2$ we have a singularity.
Next, we argue that $r=1$ is the only maximum in $(d/2,\infty)$.  We show that $\frac{\text{d}h}{\text{d}r}(r)>0$ for all $r<1$ and 
$\frac{\text{d}h}{\text{d}r}(r)<0$ for all $r>1$. It is sufficient to show that for all $r\in (d/2,\infty)$, it holds that
   \begin{align}\label{al:eq_min}
       d-r\arcsin(\frac{d}{2r}) \sqrt{4-\left(\frac{d}{r}\right)^2} >0.
   \end{align}
For all $x\in (0,1)$, it holds that $\arcsin(x)\sqrt{\frac{1}{x^2}-1}<1$. Therefore, since $\frac{d}{2r}\in (0,1)$, it follows that
\begin{align*}
    0<1-\arcsin\left(\frac{d}{2r}\right)\sqrt{\left(\frac{2r}{d}\right)^2-1}=
    1-\arcsin\left(\frac{d}{2r}\right)\frac{r}{d}\sqrt{4-\left(\frac{d}{r}\right)^2}.
\end{align*}
Multiplying by $d$ yields (\ref{al:eq_min}). Hence, $r=1$ is the only maximum of $h(r)$ for $r\in(d/2, \infty)$. Thus, $\sarc{1}{uv}$ is the most efficient connecting arc.
\end{proof}
The following corollary is a summary of the preceding lemmas and describes their implications for aggregation solutions.
\begin{corollary}
\label{cor:local_improvements_stronger}
    Let $\solution$ be a solution for $P_F^1$ and $h$ a non-closed curve on the boundary of~$\solution$ with starting point~$u$ and ending point~$v$ such that
    \begin{enumerate}
        \item $\LL(\overline{uv})<2$,
        \item $h$ does not intersect the line through $u$ and $v$ properly.
        \item the circular arc~$\sarc{1}{uv}$ does not intersect $\partial\solution\setminus h$ or $\polygons$ properly.
        \item the circular arc $C_\text{iso}(h)$ that is left of the line induced by $u$ and $v$ and has the same length as~$h$ has central angle $\theta < \pi$, 
    \end{enumerate}
    If~$h \neq \sarc{1}{uv}$, then there exists a solution $\solution'$ with a smaller objective value than~$\solution$.
\end{corollary}
Next we want to show that any free piece that is not a circular arc can be improved. For two distinct points $u,v \in \mathbb{R}^2$, we denote with $L(u,v)$  the directed line induced by $u$ and~$v$, oriented from $u$ to $v$. 
For this we need one more property that is satisfied by free pieces used in an optimal solution, but first, we show that local sub-curves of free boundary pieces are isolated.

\begin{lemma}[isolation lemma]\label{lemma:isolation}
    Let $h$ be a free piece of an optimal solution $\solution$. Then, for every interior point $p$ of $h$, there exists an $\varepsilon>0$ such that $h_\varepsilon=h\cap B_\varepsilon(p)$ is contiguous and $B_\varepsilon(p)\cap (\polygons\cup\partial\solution \setminus h_\varepsilon )=\emptyset$.
\end{lemma}
\begin{proof}
    Let $h$ be a free piece of an optimal solution~$\solution$ and $p\in h$ be an interior point. The lemma purports the existence of an $\varepsilon>0$ such that $h_l=h\cap B_\varepsilon(p)$ is a connected curve and $B_\varepsilon(p)\cap (\polygons\cup\partial\solution \setminus h_l )=\emptyset$. We construct the ball~$B_\varepsilon(p)$ in several steps. First, consider a ball $B_{\varepsilon'}(p)$, sufficiently small so that $B_{\varepsilon'}(p)\cap \polygons=\emptyset$ and the curve of $h\cap B_{\varepsilon'}$ through $p$
    connects two points $a$ and $b$ in $\partial B_{\varepsilon'}(p)$, that is, $h$ first encounters $a$, then $p$, then $b$. We may assume that $h_{ab}\cap\partial B_{\varepsilon'}(p)=\{a,b\}$, and preserve this property by updating the definition of $a$ and $b$ every time we shrink the ball.
    We note that free pieces that do not connect boundaries of input polygons are only beneficial in terms of the objective value if the perimeter of such an excluded hole
    is less than or equal to the area it encloses. For a given perimeter, the area enclosed is maximized by the isoperimetric inequality, and the area bounded by closed holes in an optimal solution is bounded from below by $2\pi$, meaning that each such component contributes a value large than some constant to the total perimeter.
    Because the total perimeter value is finite, there cannot exist more than finitely many of these components, so by shrinking the ball further, we can assume that $B_{\varepsilon'}(p)$ has empty intersection with such components of $\solution$.
    
    Therefore, if $B_{\varepsilon'}(p)$ intersects different components of free pieces of $\solution$, then these free pieces need to connect to points on $\polygons$. As a consequence, the ball $B_{\varepsilon'/2}(p)$ only intersects finitely many free pieces of $\solution$, since any such intersection means that the free piece has to travel from the boundary $\partial B_{\varepsilon'}(p)$  to $\partial B_{\varepsilon'/2}(p)$, which contributes at least $\varepsilon'/2$ to the length. Therefore, since each free piece is compact, we can shrink $B_{\varepsilon'}(p)$ so that it only intersects the free piece $h$. We denote by $h_{\varepsilon'}$ the subcurve of $h$ in $\partial B_{\varepsilon'}(p)$ through~$p$. Again because the solution set consists of compact curves, $(h\setminus h_{ab})\cup\{a,b\}\cap B_{\varepsilon'}(p)$ has a point of minimal distance to $p$. Since the free pieces are $C^1$ curves
    , every sufficiently small neighborhood of $p$ intersects $h$ only in a single curve (this is a well-known consequence of the implicit function theorem). From these two preceding facts, we see that we can shrink the ball to obtain $B_{\varepsilon}(p)$ as in the statement of the lemma.
\end{proof}

\begin{definition}[local left-sidedness]
A curve $h$ is locally left-sided if for every interior point $p \in h$, there exists a sub-curve $h_l \subseteq h$ with $p \in h_l$ in its interior and $\LL(h_l) > 0$ such that for any pair of distinct points $u,v \in h_l$ such that $v$ comes after $u$ on $h_l$, we have: 
\begin{itemize}
    \item the sub-curve $h_c \subseteq h_l$ from $u$ to $v$ lies entirely on the left side of $L(u,v)$, and
    \item the remaining part $h_l \setminus h_c$ lies entirely on the right side of $L(u,v)$,
\end{itemize} 
where points are allowed to lie on $L(u,v)$ in both cases.
\end{definition}


\begin{lemma}\label{lemma:local_onsidedness}
Every free piece of an optimal solution is locally left-sided.
\end{lemma}
\begin{proof}

    For a free piece $h$, take a small ball $B$ from \Cref{lemma:isolation} around $p$, and set $h_l=h\cap B$.
    Now let $u, v\in h_l$. Assume there exists a point $q\neq v$ on~$L(u,v)$ such that~$h_l$ goes to the right of $L(u,v)$ after $q$. 
    Let~$q'$ be the next point on~$L(u,v)$ after~$q$ that lies on~$h_l$. Then~$\overline{q q'}$ lies in the interior of~$\solution$ (except for the points $q$ and $q'$). We can therefore replace the subcurve of~$h_l$ between~$q$ and~$q'$ with~$\overline{qq'}$ to obtain a valid solution that decreases both the perimeter and the area.
    Hence, the first sub-property of local left-sidedness is fulfilled.
   Now, assume that the second sub-property does not hold for $u, v\in h_l$. Then, there exists a point $z\in h$ before $u$ or after $v$ that is to the left of $L(u,v)$. Without loss of generality, assume that $z$ comes after $v$ on $h_l$. Then $v$ is to the right of $L(u,z)$ and it follows that the first sub-property is violated for the subcurve of $h$ starting at $u$ and ending at $z$. This is a contradiction to the first sub-property. Hence, the second sub-property follows.
  \end{proof}

\localImprovements*

\begin{proof}
     Let $h$ be a free piece starting at $u$ and ending at $v$ that is locally left-sided and not the unit arc $\mathcal{C}_1^{uv}$. Let $p$ be an interior point of $h$ where $h$ is locally not a unit arc, i.e., every sub-curve of $h$ that contains $p$ in its interior is not a unit arc. Due to \Cref{lemma:isolation} and local left-sidedness, we can find a small ball $B_\varepsilon(p)$ such that:
\begin{enumerate}[i]
    \item $B_\varepsilon(p)\cap\left(\polygons \cup\partial\solution\right)$ is a contiguous piece $h_\varepsilon\subseteq  h$.
    \item It holds that $\varepsilon<1$.
    \item Let $u'$ and $v'$ be the endpoints of $h_\varepsilon$. The line $L(u',v')$ has $h_\varepsilon$ to its left or possibly touching or coinciding with $L(u',v')$.
\end{enumerate}
If $h_\varepsilon$ is a circular arc, we exchange an arbitrary subarc from $u''$ to $v''$ of $h_\varepsilon$ with central angle~$\leq \pi$ with $C_1^{u''v''}$. Otherwise, $h_\varepsilon$ is not a circular arc. If $C_{iso}(h_\varepsilon)\subset B_\varepsilon(p)$, then exchange $h_\varepsilon$ with $C_{iso}(h_\varepsilon)$ and improve the objective value by \Cref{lemma:arcs_enclose_much:paper}. Otherwise, exchange $h_\varepsilon$ with $C_\varepsilon^{u'v'}$. This has shorter length then $C_{iso}(h_\varepsilon)$ and in turn shorter length than $h_\varepsilon$. Additionally, it includes less area than $h_\varepsilon$, since $h_\varepsilon$ is  contained in $B_\varepsilon(p)$. Hence, we can improve the objective value in all cases and the lemma follows.
 \end{proof}
Using this insight, we can describe the free pieces as circular arcs with certain additional properties. Next we investigate how optimal free pieces connect to the boundaries of the input. To do so, we relate the tangents of the free pieces to the tangents of the boundaries to which they attach. Let $\gamma:[0,1]\rightarrow\mathbb{R}^2$ be a piece wise differentiable curve that is oriented counterclockwise. Then for every $p=\gamma(t_0)$ with $t_0\in [0,1]$ {(except for $t_0=0$ and $t_0=1$ if $\gamma$ is not closed)}, we have a left derivative $\frac{d\gamma}{dt}(t_0)^-$ and a right derivative $\frac{d\gamma}{dt}(t_0)^+$. Note that $\frac{d\gamma}{dt}(t_0)^-$ and $\frac{d\gamma}{dt}(t_0)^+$ may not agree at the corners of the differentiable pieces of $\gamma$.

With this we can define the forward and backward ray of $p=\gamma(t_0)$ on $\gamma$ as
\[
\forwardRay{p}{\gamma}= \left\{ \gamma(t_0) + \lambda \, \frac{d\gamma}{dt}(t_0)^+ \;\middle|\; \lambda \ge 0 \right\}\\
\backwardRay{p}{\gamma}= \left\{ \gamma(t_0) - \lambda \, \frac{d\gamma}{dt}(t_0)^- \;\middle|\; \lambda \ge 0 \right\}.
\]
 We denote the counterclockwise angle between rays $r_1,r_2$ with $\tangentangle{r_1}{r_2}$. With this we define the \emph{exterior region angle}  of $\polygons$ at $x\in B\in \polygons$ as $\beta_\polygons(x)=\tangentangle{\backwardRay{x}{\partial B}}{\forwardRay{x}{\partial B}}$.
Now let $f$ be a non-closed arc starting at $u\in B_u$ and ending at $v\in B_v$. 
We define the \emph{exterior arc angle} of $f$ at $u$ on $\polygons$ by $\beta_\polygons(f,u)=\tangentangle{\backwardRay{u}{f}}{\forwardRay{u}{\partial B_u}}$ and for the endpoint $v$ of~$f$ by $\beta_\polygons(f,v)=\tangentangle{\backwardRay{v}{\partial B_v}}{\forwardRay{v}{f}}$.  
See \Cref{fig:outer_arc_angle_all} for a visualization.

Using this, we can characterize the free boundary pieces that can appear in an optimal solution (provided that an optimal solution exists).

\begin{proposition}\label{prop:arc_properties:full}
    Let $\solution$ be an optimal solution of $\freeProblem$ \done{for~$\alpha \in (0,\infty)$}. Every free boundary piece $f$ with endpoints $u$ and $v$ has the following properties: 
    \begin{description}
    \item[P1:] The distance $d=\LL(\overline{uv})$ of the endpoints is at most $2\alpha$.  
    \done{\item[P2:] The piece $f$ is circular with radius $\alpha$.} 
    \item[P3:] The central angle $\theta$ of $f$ is at most $\pi$.
        \item[P4:] If \(x \in \{u,v\}\) lies on a boundary curve \(\partial B\) of an input region, then \(\beta_\polygons(f,x) \ge \pi\).
    \done{\item[P5:] The piece $f$ bends to the left (i.e., it is to the left of $\overrightarrow{uv}$).}
    \end{description}
\end{proposition}

By~\Cref{obs:scale:paper}, we assume $\alpha=1$ for all proofs.
\begin{proof}[Proof of P1]
    Let $f$ be a free piece connecting $u$ and $v$ with $\LL(\overline{uv})>2$. Then it cannot be a circular arc of radius~$1$ and~\Cref{lemma:local_structure:paper} implies that $f$ cannot be part of an optimal solution.
\end{proof}
\begin{proof}[Proof of \done{P2 and P5}]
This is an immediate consequence of~\Cref{lemma:local_structure:paper} because every free piece that does not have this property can be improved locally. 
\end{proof}
\begin{proof}[Proof of P3]
By \done{P2 and P5}, we know that~$f$ is a circular arc with radius $1$ that bends inwards to the region included in~$\solution$.
Assume for the sake of contradiction that the central angle~$\theta$ of~$f$ is greater than~$\pi$.
Let~$f'$ be the circular arc with central angle~$2\pi-\theta$ connecting $u$ and $v$ that bends in the same direction as~$f$ (see Figure~\ref{fig:P3:full}), and let~$\enclosed{f'}{f}$ denote the region enclosed by~$f$ and~$f'$.
If we replace~$f$ with the parts of~$f'$ that do not intersect the remaining boundary~$\partial\solution\setminus f$, we obtain a new solution~$\solution'$ whose objective value is upper bounded by
\[g_1(\solution') \leq g_1(\solution)+\A(\enclosed{f'}{f})-\LL(f)+\LL(f').\]
Note that equality may not hold because~$\enclosed{f'}{f}$ may intersect~$\solution$.
The region~$\enclosed{\overline{uv}}{f}$ enclosed by~$f'$ and $\overline{uv}$ has area~$\A(\enclosed{\overline{uv}}{f})=\frac{1}{2}((2\pi-\theta)-\sin(2\pi-\theta))$.
Because~$f'$ reflected along~$\overline{uv}$ completes $f$ to a circle, we have~$\A(\enclosed{f'}{f})=\pi-2\cdot \A(\enclosed{\overline{uv}}{f})= \theta-\pi+\sin(\theta)$.
Overall, we have
\[g_1(\solution') \leq g_1(\solution) + \theta-\pi+\sin(\theta) - \theta + (2\pi-\theta) = g_1(\solution) - \theta + \pi+\sin(\theta) < g_1(\solution).\]
In the special case~$\theta=2\pi$, the arc~$f'$ degenerates and~$\enclosed{f'}{f}$ becomes the circle enclosed by~$f$.
The objective value of the solution~$\solution'$ that removes~$f$ can be bounded by
\[g_1(\solution') \leq g_1(\solution) + \pi - 2\pi < g_1(\solution).\qedhere\]
\end{proof}

\begin{figure}
    \centering

    \begin{minipage}{0.48\textwidth}
        \centering
        \includegraphics[page=2]{images/new/P3.pdf}
        \caption{Property P3: replacing arc $f$ with $f'$.}\label{fig:P3:full}
    \end{minipage}
    \hfill
    \begin{minipage}{0.48\textwidth}
        \centering
        \includegraphics{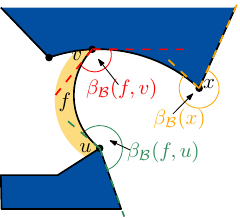}
        \caption{Two examples for the exterior arc angle and one for the exterior polygon angle.}\label{fig:outer_arc_angle_all}
    \end{minipage}

\end{figure}

\begin{proof}[Proof of P4]
\begin{figure*}[t]
    \centering

    \begin{minipage}[t]{0.45\linewidth}
        \centering
        \includegraphics{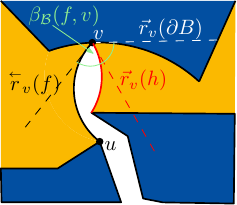}
        \caption{The tangent rays appear in counter-clockwise order $\backwardRay{v}{f}$, $\forwardRay{v}{h}$, $\forwardRay{v}{\partial B}$.}
        \label{fig:proof_local_improvement_arc}
    \end{minipage}
    \hfill
    \begin{minipage}[t]{0.45\linewidth}
        \centering
        \includegraphics{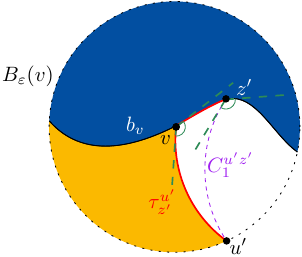}
        \caption{There exists a point $z'$ on $h$ such that the objective value decreases when using $\sarc{\alpha}{u'z'}$.}
        \label{fig:proof_local_improvement}
    \end{minipage}

\end{figure*}
Let $\solution$ be a solution, $S\in\solution$ a region, and let $f=\sarc{1}{uv}$ be a free piece in $\partial S$ starting at $u$ and ending at $v$. Let $B_u$, respectively $B_v$, denote the input region with $u$, respectively $v$, on its boundary. For a contradiction, assume without loss of generality that $\beta_\polygons(f, v) < \pi$; the case $\beta_\polygons(f, u) < \pi$ can be handled symmetrically. Let $h$ denote the next boundary piece after $f$ on $\partial S$.

We first show that $\tangentangle{\backwardRay{v}{f}}{\forwardRay{u}{h}} \leq \beta_\polygons(f, v) < \pi$. If $h$ is a constrained piece, this follows because $\tangentangle{\backwardRay{v}{f}}{\forwardRay{v}{h}} = \beta_\polygons(f, v)$. If $h$ is free, the ray $\forwardRay{h}{v}$ must lie between $\backwardRay{f}{v}$ and $\forwardRay{B_v}{v}$ in counter-clockwise order (see \Cref{fig:proof_local_improvement_arc}), since otherwise $h$ would intersect $B_v$ properly or $S$ would not lie to the left of both $f$ and $h$. Consequently, $\tangentangle{\backwardRay{f}{v}}{\forwardRay{h}{v}} < \beta_\polygons(f, v) < \pi$.

For a point $x \in f$ and a point $y \in h$, we denote by $\tau^{x}_y$ the boundary piece consisting of the sub-arc of $f$ from ${x}$ to $v$, concatenated with the sub-piece of $h$ from $v$ to $y$. We show that there is a choice of ${x}$ and $y$ such that $\tau^{x}_y$ can be improved locally.

Since $f$ and all input boundaries are piecewise differentiable curves, a straightforward adaptation of \cref{lemma:isolation} to points $u$ on the input boundary $\partial\polygons$ implies that there exists a ball $B_\varepsilon(v)$ centered at $v$ such that:
\begin{enumerate}[i]
    \item $\varepsilon < 1$, \label{prop:label:i}
    \item $B_\varepsilon(v) \cap \partial\polygons$ is a contiguous piece $b_v$ of $\partial B_v$,
    \item $B_\varepsilon(v) \cap \partial\solution = \tau^{u'}_z$, where $\tau^{u'}_z$ is a contiguous subcurve with $u' \in f$ and $z \in h$, and
    \item $\tau^{u'}_z$ is locally left-sided (here we also use $\tangentangle{\backwardRay{v}{f}}{\forwardRay{u}{h}} < \pi$).
\end{enumerate}

For some $z' \in h \cap \tau^{u'}_z$ we want to apply \cref{cor:local_improvements_stronger} to $\tau^{u'}_{z'}$ . Preconditions~1 and~2 follow immediately from (i) and~(iv) for all sub-curves of $\tau^{u'}_z$. To satisfy precondition~3, we only need to ensure that the arc $\sarc{1}{u'z'}$ does not intersect $b_v$; properties (ii), (iii), and (iv) ensure that the solution boundary and the rest of the input boundary is not intersected.

For fixed $u'$, as $z'$ moves from $v$ along $h$, the arc $\sarc{1}{u'z'}$ rotates about the anchor point~$u'$ (see \Cref{fig:proof_local_improvement}), so $\backwardRay{\sarc{1}{u'z'}}{z'}$ varies continuously and the arc can only properly intersect $b_v$ to the right of $f$, i.e., on $h$. Since $h$ is differentiable, the tangent $\forwardRay{h}{z'}$ also changes continuously. The arc $\sarc{1}{u'z'}$ can properly intersect $h$ only if $\tangentangle{\backwardRay{\sarc{1}{u'z'}}{z'}}{\forwardRay{h}{z'}} > \pi$, but at $z' = v$ this angle equals $\tangentangle{\backwardRay{f}{v}}{\forwardRay{h}{v}} < \pi$ and varies continuously, so precondition~3 is satisfied for $z'$ sufficiently close to $v$.

For precondition~4, note that $\sarc{1}{u'v}$ has central angle smaller than $\pi$. Since the central angle of $C_\text{iso}(\tau^{u'}_{z'})$ varies continuously in $z'$, it is also smaller than $\pi$ for $z'$ sufficiently close to~$v$. Hence, there exists $\tau^{u'}_{z'}$ for which \cref{cor:local_improvements_stronger} applies, contradicting the optimality of~$\solution$.

\end{proof}
From now on we say that $f$ is \emph{tangential} on $\partial B$ at $x$, if \(\beta_\polygons(f,x) = \pi\). Note that for any differentiable point $x$ this coincides with the natural definition of tangentiality, i.e., the tangents of $f$ and $\partial B$ at $x$ coincide. Then P4 of \Cref{prop:arc_properties:full} implies that a free piece $f$ of an optimal solution has to be tangential in the interior of the differentiable pieces of the boundary curves.

\begin{corollary}\label{cor:disjointness:full}
    \done{Let $\solution$ be an optimal solution of $\freeProblem$ for~$\alpha \in (0,\infty)$. The inclusion-maximal free pieces of~$\solution$} are disjoint except possibly at endpoints. 
\end{corollary}
\begin{proof}
Assume otherwise, so that there exists an inclusion-maximal free boundary piece $f$ with endpoints $u$ and $v$ and another inclusion-maximal free boundary piece $f'$ with endpoints $u'$ and $v'$ such that $x\in f\cap f'$. By \done{P2, P3 and P5} of \Cref{prop:arc_properties:full}, $f$ and~$f'$ are circular arcs and the region between them is part of the solution. Let $h_1$ be the concatenation of the subcurve of $f'$ from $u'$ to $x$ and then $f$ from $x$ to $v$. Similarly, let $h_2$ be composed of the part of $f$ from $u$ to $x$, followed by the subcurve of $f'$ from $x$ to $v'$. Replacing $f$ and $f'$ with $h_1$ and $h_2$ in $\solution$ does not change the objective value. However, since $h_1$ and $h_2$ are not circular arcs, $\solution$ cannot be an optimal solution by P2 of~\Cref{prop:arc_properties:full}.
\end{proof}
\subsection{Existence of an Optimal Solution}
We next show that our problem always has an optimal solution in which all boundary curves are piecewise differentiable.
\begin{theorem}\label{prop:there_exists_a_solution:full}
    Let $\polygons$ be an instance of the aggregation problem $P_F^1$. There exists a solution $\solution$ that minimizes $g_1$, such that the free pieces comprising the boundary $\partial \solution$ form a collection of differentiable ($C^1$) curves. 
\end{theorem}

We will show that the existence and piecewise $C^1$ regularity of the optimal solution of the unrestricted problem follows from deep results in geometric measure theory concerning approximations of perimeters that locally converge in measure. There are other approaches that could be used for the problem, similar to the classical isoperimetric inequality, including the Arzel\` a-Ascoli theorem. The overall strategy will be to first show the existence (\Cref{prop:existenceminimizer}) of a minimizer as a bounded set of finite perimeter, and to subsequently deduce (\Cref{prop:regularsolutions}) the $C^1$ regularity of the free pieces, so the theorem then follow from combining these two results. 

First, we replace the perimeter inside an open set by a more versatile variant requiring less regularity.
\begin{definition}[Relative perimeter]
Let $U\subset\mathbb{R}^2$ be open and let $E\subset\mathbb{R}^2$ be measurable. The perimeter of $E$
relative to $U$ is
\[
\Per(E;U)
:=
\sup\left\{
\int_E \operatorname{div}\varphi\,dx
:
\varphi\in C_c^1(U,\mathbb{R}^2),\ \|\varphi\|_{L^\infty}\le 1
\right\}.
\]
If $\Per(E;\mathbb{R}^2)<\infty$, we say that $E$ is a set of finite perimeter.
If $\partial E$ is piecewise $C^1$, then 
\[
\Per(E;U)=\LL(\partial E\cap U).
\]
\end{definition}
We only give the definition for completeness' sake and will ultimately only need some abstract properties of the relative perimeter, and not use the exact definition explicitly. 
The following is a standard tool of modern geometric measure theory used to deduce regularity results for $E$. The idea is to capture properties similar to lower-semicontinuity in a measure theoretic sense. We write $M \Subset N$ for subsets of $\mathbb{R}^2$ whenever both $\overline{M}\subset N$ and $\overline{M}$ is compact (bounded).

\begin{definition}[$(\Lambda,r_0)$-perimeter minimizer]
Let $U\subset\mathbb{R}^2$ be open, let $\Lambda\ge 0$, and let $r_0>0$. A measurable set
$E\subset\mathbb{R}^2$ is called a \emph{$(\Lambda,r_0)$-perimeter minimizer in $U$} if for every
ball $B_r(x)\Subset U$ with $0<r<r_0$, and for every measurable set $F\subset\mathbb{R}^2$
such that $E\Delta F \Subset B_r(x)$, one has
\[
\Per(E;B_r(x))
\le
\Per(F;B_r(x))+\Lambda\,A(E\Delta F).
\]
\end{definition}
The following two theorems represent the main ingredients for our regularity results, each collecting a series of established results in the literature. 
\begin{theorem}[{\cite[Proposition~3.38, Theorem~3.39]{AFP}}]\label{thm:BVcompactness}
Let $(E_h)$ be a sequence of measurable subsets of a compact set $K\subset\mathbb{R}^2$ with $\sup_h \bigl( A(E_h)+\Per(E_h;\mathbb{R}^2) \bigr)<\infty$.
Then there exist both a subsequence $(E_{h_k})$ and a finite-perimeter measurable set $E\subset K$ such that
\[
A(E_{h_k}\Delta E)\to 0 \text{ and } \Per(E;\mathbb{R}^2)\le \liminf_{k\to\infty}\Per(E_{h_k};\mathbb{R}^2).
\]
\end{theorem}

\begin{theorem}
\label{thm:quasimin-regularity}
Let $U\subset\mathbb{R}^2$ be open, such that there exists a $(\Lambda,r_0)$-perimeter minimizer $E$ in $U$. Then the free pieces of the boundary of $E$ are $C^1$. 
\end{theorem}
The exact statement of \Cref{thm:quasimin-regularity} comes from two places. First, there is a statement for only part of the boundary, the so-called reduced boundary away from singular points~\cite[Theorems~26.3,~26.5]{Maggi2012}. The fact that there are no singular points in the planar case follows from the fact that the singular set is empty in dimensions $<8$~{\cite[Theorem 28.1]{Maggi2012}}. 

\begin{proposition}\label{prop:existenceminimizer}
There exists a feasible minimizer $\solution_\ast$ of $g_1$.
\end{proposition}

\begin{proof}
Consider a sequence of feasible solutions $\solution_h$, which we treat as regions, such that $g_1(\solution_{h})$ converges to the infimum of $g_1$. 
By passing to a subsequence we can assume $\sup_h \bigl( A(\solution_h)+\Per(\solution_h;\mathbb{R}^2)\bigr)<\infty$. Moreover, we can assume that every set in $\solution_h$ is contained in the convex hull $\conv(\polygons)$ of the input polygons. Indeed, it suffices to note that replacing $\solution_h$ with the solution obtained from intersecting the sets in $\solution_h$ with $\conv(\polygons)$ decreases both the area and the perimeter. In particular, there exists a fixed compact set $K\subset\mathbb{R}^2$ such that $\solution_h\subset K$ for all $h$.

\Cref{thm:BVcompactness} implies that there exists a subsequence together with a finite-perimeter set
$\solution_\ast\subset K$ such that $A(\solution_h\Delta \solution_\ast)\to 0$ as $h\to\infty$, and $\Per(\solution_\ast;\mathbb{R}^2)\le \liminf_{h\to\infty}\Per(\solution_h;\mathbb{R}^2)$.

Since $\polygons\subset \solution_h$ for every $h$, the characteristic function $\chi_{\solution_h}$ of the set $\solution_h$ satisfies $\chi_{\solution_h}=1$ almost everywhere (a.e.) on $S$. Along with $A(\solution_h\Delta \solution_\ast)\to 0$, this implies that $\chi_{\solution_\ast}=1$ a.e. on $S$. After changing $\solution_\ast$ on a null set if necessary, we have $\polygons\subset \solution_\ast$, so $\solution_\ast$ is a feasible solution.

Observe that $A(\solution_h)\to A(\solution_\ast)$ because $A(\solution_h\Delta \solution_\ast)\to 0$, so we obtain (again using \Cref{thm:BVcompactness})
\[
g_1(\solution_\ast)
=
A(\solution_\ast)+\Per(\solution_\ast;\mathbb{R}^2)
\le
\liminf_{h\to\infty}\bigl(A(\solution_h)+\Per(\solution_h;\mathbb{R}^2)\bigr)
=
\inf_{\solution} g_1(\solution),
\]
implying equality and concluding the proof.
\end{proof}

\begin{proposition}\label{prop:regularsolutions}
Let $\solution_\ast$ be a minimizer of $g_1$. Then every connected component of $\partial \solution_\ast\cap (\mathbb{R}^2\setminus \polygons)$ is of class $C^1$.
\end{proposition}

\begin{proof}
Let $U\Subset \Omega=(\mathbb{R}^2\setminus \polygons)$, so $\overline U\cap S=\emptyset$ and $d:=\operatorname{dist}(\overline U,S)>0$. Setting $r_0:=\frac d2$, we have $B_r(x)\Subset \Omega$ for $0<r<r_0$ and $x\in U$. We will show that $\solution_\ast\cap U$ is a perimeter minimizer for $U$, so that \Cref{thm:quasimin-regularity} concludes the proof as $U$ is arbitrary. To this end, let $F$ be a measurable set with $\solution_\ast\Delta F \Subset B_r(x)$. We have $F=\solution_\ast$ except in $B_r(x)$, and $B_r(x)\cap \polygons=\emptyset$, so $\polygons\subset F$, meaning $F$ is a feasible solution for the minimization problem. By minimality of $\solution_\ast$, $A(\solution_\ast)+\Per(\solution_\ast;\mathbb{R}^2)\le A(F)+\Per(F;\mathbb{R}^2)$, and thus 
\[
A(\solution_\ast\cap B_r(x))+\Per(\solution_\ast;B_r(x))
\le
A(F\cap B_r(x))+\Per(F;B_r(x)).
\]
This in turn yields
\[
\Per(\solution_\ast;B_r(x))
\le
\Per(F;B_r(x)) + \big(A(F\cap B_r(x))-A(\solution_\ast\cap B_r(x))\big)\le \Per(F;B_r(x)) + A(\solution_\ast\Delta F),
\]
since $\bigl|A(F\cap B_r(x))-A(\solution_\ast\cap B_r(x))\bigr|
\le
A(\solution_\ast\Delta F)$.

Thus, $\solution_\ast$ is a $(1,r_0)$-perimeter minimizer in $U$. Since $U$ was arbitrary, \Cref{thm:quasimin-regularity} applies at
every point of $\partial \solution_\ast\cap \Omega$, so every connected component of
$\partial \solution_\ast\cap \Omega$ is of class $C^1$.
\end{proof}



\subsection{A Polynomial Time Algorithm for Circular Input Polygons}
    After establishing the characterization of the free pieces and the existence of a solution in the most general setting, we now restrict ourselves to the input class of circular polygons (cf. \cref{sec:problem_def}), which is also the setting considered in the main body of this work. This restriction is needed since for general differentiable curves, the number of arcs that satisfy Proposition~\ref{prop:arc_properties:full} can be infinite. However, to develop a polynomial-time algorithm for~$\freeProblem$, we require a polynomial bound (in the input size) on the number of candidate arcs.
    For practical purposes, this is not really a restriction since cartography applications and fencing problems usually only consider polygonal inputs. The extension to circular polygons is only needed for the preprocessing algorithm given in \Cref{sec:preprocessing:appendix},
    {which successively replaces parts of the input~$\polygons$ with optimal solutions for subsets $\polygons'\subseteq\polygons$.
    Because these solutions contain circular arcs by \Cref{prop:arc_properties:full}, the input becomes a set of circular polygons.}

For the class of circular polygons, the solutions for $\alpha=\infty$ have a simple structure.
\begin{lemma}\label{lem:infinity_edges}
    {For an instance~$\polygons$ of~$P_F^\infty$ given by circular polygons, every maximal free piece~$f$ in an optimal solution~$\solution$ is a straight-line segment between two polygon vertices~$u,v \in V(\polygons)$ that does not intersect~$\polygons$ properly.}
\end{lemma}
\begin{proof}
By \cref{lem:infinity}, every optimal solution region $S$ is the convex hull of the polygons assigned to it.
For any optimal solution region $S$ with the set of contained polygons $\polygons(S)$ we have
\(
\conv\!\Big(\bigcup_{B\in\polygons(S)} B\Big)
=
\conv\!\Big(\bigcup_{B\in\polygons(S)} \conv(B)\Big).
\)
Since all arcs of the circular polygons bend inwards, $\conv(B)$ is a convex polygon whose vertices are contained in $V(B)$.
Hence, every extreme point of 
\(
\conv\!\Big(\bigcup_{B\in\polygons(S)} \conv(B)\Big)
\)
lies in $V(\polygons(S))$. Consequently, all edges of $S$ and in particular all maximal free pieces of $S$, are straight-line segments between vertices of $V(\polygons(S)) \subseteq V(\polygons)$.
\end{proof}
Next, we limit the number of inclusion-maximal free pieces that have to be inspected to find an optimal solution. For this, we need one more lemma that handles ambiguous cases where an infinite number of different arcs have the same objective value and connect the same boundary objects.
\begin{restatable}{lemma}{distanceTwo}
    \label{lemma:distance2}
    \done{For an instance~$\polygons$ of~$\freeProblem$ with~$\alpha \in (0,\infty)$}, there exists an optimal solution in which no free boundary piece
    \begin{enumerate}
        \item connects the interiors of two parallel straight-line edges, or
        \item connects a vertex $u\in V(\polygons)$ with the interior of a circular boundary edge that is a sub-curve of a circle centered at $u$.
    \end{enumerate}

\end{restatable}%
\begin{figure}
    \centering
    \begin{subfigure}[b]{0.32\textwidth}
        \centering
        \includegraphics[page=2]{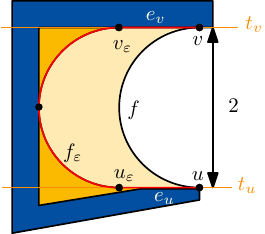}
    \end{subfigure}
    \begin{subfigure}[b]{0.32\textwidth}
        \centering
        \includegraphics[page=4]{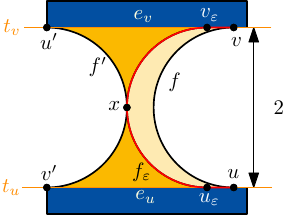}
    \end{subfigure}
    \begin{subfigure}[b]{0.32\textwidth}
        \centering
        \includegraphics[page=6]{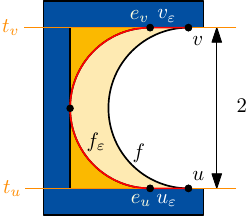}
    \end{subfigure}
    \caption{Visualization of the proof of~\Cref{lemma:distance2}.
    We move~$f$ inwards until it touches another free piece (middle) or an input polygon (left, right).
    The left figure corresponds to case a), the middle to case b), and the right to case c).
    In the first two cases, we show that~$f$ cannot be part of an optimal solution. 
    On the right, both~$f$ and~$f_\varepsilon$ are optimal for~$\alpha=1$. In particular, for both arcs the objective value is $6-\frac{3}{2}\pi(\thickapprox1.3)$ less than if we add nothing at all.}
    \label{fig:Distance2Possible:full}
\end{figure}

\begin{proof}[Proof of 1] 
    By~\Cref{obs:scale:paper}, we assume that $\alpha=1$.
    Let~$\solution$ be an optimal solution, which exists by~\Cref{prop:there_exists_a_solution:full}, and let~$f=\sarc{1}{uv}$ a free piece with endpoints $u$ and $v$ in the interior of two parallel input straight line edges $e_u,e_v$. Due to P4 of~\cref{prop:arc_properties:full} the only way for $f$ to be optimal is if the central angle of $f$ is exactly $\pi$ and $\ell(\overline{uv})=2$.
    The tangents $t_u$ and $t_v$ of the circle of radius $1$ through $u$ and $v$ contain $e_u$ and $e_v$, respectively.
    
    For every $\varepsilon>0$, we define $u_\varepsilon$ (resp.~$v_\varepsilon$) as the point on $t_u$ (resp.~$t_v$) with distance $\varepsilon$ to $u$ (resp.~$v$)  
    in the direction of the arc, at $u$ (resp.~$v$).
    Further, let $f_\varepsilon$ be the curve consisting exactly of $\overline{u u_\varepsilon}$, the circular arc $C_1^{u_\varepsilon v_\varepsilon}$, and $\overline{v_\varepsilon v}$. See Figure~\ref{fig:Distance2Possible:full} for an example. 
    Now, let $\varepsilon>0$ be the maximum value such that $f_\varepsilon$ touches $\partial\solution\setminus f$ and does not intersect it properly. Note that $\varepsilon$ is properly larger than zero, since $u$ and $v$ are in the interior of straight line edges.
    Let~$\solution'$ denote the solution that replaces $f$ with $f_{\varepsilon}$.
    This increases the perimeter by $2\varepsilon$ and decreases the area by the same amount, so the objective value does not change.
    Hence, $\solution$ is optimal if and only if~$\solution'$ is optimal.

    \begin{enumerate}[a)]
        \item Consider the case that $\overline{u u_\varepsilon}\not\subset \partial\polygons$ or $\overline{v v_\varepsilon}\not\subset \partial\polygons$. Then $\solution'$ contains a free boundary piece that is not a circular arc. Hence, it follows by P2 of~\Cref{prop:arc_properties:full} that $\solution$ cannot be optimal.
        \item Consider the case that $f_\varepsilon$ touches another free boundary piece.
        By~\Cref{cor:disjointness:full}, the free boundary pieces are disjoint except possibly at endpoints. Hence, $\solution$ cannot be optimal.
        \item The remaining case is that the circular arc between $u_\varepsilon$ and $v_\varepsilon$ touches $\partial\polygons$.
        Then the circular arc has been split into two free pieces (separated by a constrained piece consisting of a single point), where the corresponding new free pieces do not connect the interior of parallel edges.
    \end{enumerate}
    We can handle each free boundary piece with endpoints on parallel edges in this way.
    This way, we either show that $\solution$ is not optimal or we obtain an optimal solution of $P_F^\alpha$ where the endpoints of every free boundary piece are not on parallel edges.
\end{proof}
\begin{proof}[Proof of 2]
By~\Cref{obs:scale:paper}, we assume that $\alpha=1$.
Let $u$ be a boundary vertex and $h$ a circular boundary edge such that the inducing circle has $u$ as its center point. Let $f$ be a circular arc that is part of an optimal solution and connects $u$ and the interior of $h$. Consequently, it satisfies \cref{prop:arc_properties:full}. Thus, the center $c$ of $f$ has distance $1$ to $u$ as well as $h$. However, since the center point of $h$ is also $u$, this implies that $f$ is a semicircle with central angle $\pi$ and radius $1$. Since $f$ is part of an optimal solution, its interior does not touch any other boundary piece. Hence, $f$ can be rotated around $u$ by a small angle $\beta$ in either direction, producing a new arc $f'$ that still connects $u$ and $h$.
Rotating $f$ changes the objective value by $\frac{\beta}{2\pi}(4\pi -4\pi\cdot1^2)=0$ in one direction and by $\frac{\beta}{2\pi}(-4\pi  + 4\pi\cdot1^2)=0$ in the other direction. 
Hence, the objective value remains unchanged and we can rotate $f$ until it reaches an endpoint of $h$ or touches another boundary piece.
Thus, there exists an optimal solution that does not connect the interior of $h$ with the vertex $u$.
\end{proof}

With this we can define a set of free pieces such that there \emph{exists} an optimal solution, that uses only maximal free pieces from this set and, using ~\Cref{prop:arc_properties:full}, \Cref{lemma:distance2} and~\Cref{lem:infinity_edges}, we can show that it has polynomial size.
\begin{definition}\label{def:arcs}
    \done{Let~$\polygons$ be an instance of~$\freeProblem$. For~$\alpha\in(0,\infty)$}, let~$\allowedArcs$ denote the set of circular arcs~$f=\sarc{\alpha}{uv}$ such that
    \begin{itemize}
        \item the endpoints~$u$ and~$v$ lie on the input boundary~$\partial \polygons$, 
        \item {$u$ and $v$ are not in the interior of two parallel straight line edges,
        \item if $u \in V(\polygons)$ and $v$ lies in the interior of a circular arc $h$ (or vice versa), then $u$(or $v$) is not the center point of $h$,}
        \item the arc~$f$ fulfills properties P1--P5 of~\Cref{prop:arc_properties:full},
    
        \item and the arc~$f$ does not intersect~$\polygons$ properly.
    \end{itemize}
    \done{We define~$F_0(\polygons)=\emptyset$ and~$F_\infty(\polygons)$ as the set of straight-line segments~$\overline{uv}$ with~$u,v \in V(\polygons)$ that do not intersect~$\polygons$ properly.}
    Let~$\subdivision_C$ denote the subdivision of~$\conv(\polygons)\setminus\polygons$ induced by~$\allowedArcs$ and~$E(\polygons)$, \done{and let~$|\subdivision_C|$ denote the total number of vertices, edges and cells in~$\subdivision_C$.}
\end{definition}

{Note that~$F_0(\polygons)$ and~$F_\infty(\polygons)$ are the limits of~$\allowedArcs$ for~$\alpha\to0$ and~$\alpha\to\infty$, respectively.
For~$\alpha=0$, circular arcs become impossible to form because they would have to have radius~$0$.
Conversely, the limit of a circular arc~$\sarc{\alpha}{uv}$ for~$\alpha\to\infty$ is the straight-line segment~$\overline{uv}$.
If~$u$ or~$v$ lies in the interior of an edge~$e$, then~$\overline{uv}$ can only be tangential on~$e$ if $\overline{uv}$ intersects the corresponding input polygon properly.
Hence, $u$ and~$v$ must be vertices.}
We {now} show that~$\subdivision_C$ has polynomial size and is sufficient to construct an optimal solution.
\begin{lemma}
    For an instance $\polygons$ of $\freeProblem$ that consists of circular polygons with~$|V(\polygons)|=n$, the set~$\allowedArcs$ has size~$\bigO{n^2}$.
    \label{lemma:relevantEdges:full}
\end{lemma}
\begin{proof}
Consider an arc~$f=\sarc{\alpha}{uv} \in {F_\alpha}(\polygons)$.
Let~$C_f$ be the circle inducing~$f$, and let~$c$ be its center.
If~$v$ is an input boundary vertex, then~$c$ must lie on the circle with radius~$\alpha$ centered at~$v$.
Otherwise, $v$ lies in the interior of an edge~$e$, which is either a line segment or a circular arc.
Then the radius~$\overrightarrow{vc}$ must be orthogonal to~$e$ (because~$f$ is tangential on~$e$) and point to the right of~$e$ (because the interior of the circular polygon is to the left of~$e$).
Furthermore, $c$ must have distance~$\alpha$ to~$v$.
Thus, the possible locations of~$c$ form an open line segment or an open circular arc to the right of~$e$ at distance~$\alpha$.
The same argument applies to the starting point~$u$.
Consequently, the center~$c$ must lie at the intersection of two curves, each of which is either a circle, a circular arc, or a line segment.
If the two primitives are different, that is, if one is a line segment and the other is a circle or circular arc, then there are at most two possible locations for~$c$.
Thus, we only need to consider the following special cases in which the two center-defining curves may overlap:
(1) straight-line with straight-line,
(2) boundary vertex with circular arc, and
(3) circular arc with circular arc.

In the first two cases, there can be an infinite amount of valid positions, but by \cref{lemma:distance2} all of them can (as already done in the definition of $\allowedArcs$) be safely ignored. Next, we show that the third case is not possible at all.
Let $h$ and $h'$ be the circular arcs and $t$ and $t'$ their corresponding center-position arcs.
If $t$ and $t'$ overlap, then they both lie on the same circle $C$ with center $z$ and share a common angular interval.
By construction of $t$ and $t'$, the arcs $h$ and $h'$ also have the same center. Consequently, they also both lie on some circle centered in $z$ and overlap over the same angular interval as $t$ and $t'$. This implies that the corresponding input regions overlap, which is impossible.

Overall it follows that the number of free arcs in $\allowedArcs$  connecting any pair of boundary objects (i.e., arcs, edges, and vertices) is constant.
Because the number of boundary objects is in~$\bigO{n}$, it follows that~${F_\alpha}(\polygons) \in \bigO{n^2}$.
\end{proof}
Note that there may be two arcs connecting the same pair of boundary objects $b_1,b_2$, but they are uniquely defined by their direction, since one of the arcs starts at $b_1$ and the other one at $b_2$. Thus, For input edges \( e, h \in E(\polygons) \), we can write \( \sarc{r}{eh} \) to denote the circular arc~$\sarc{r}{uv}$ such that~$u$ is tangential on and lies in the interior of~$e$, and~$v$ is tangential on and lies in the interior of~$h$.

For $\alpha$-circular polygons the arcs that need to be considered for $F_\alpha(\polygons)$ can be further reduced. Specifically, if we optimize an $\alpha$-circular polygon for some $\alpha'\geq\alpha$ we do not need to consider any arcs ending in the interior of  circular boundary arcs.
\begin{observation}\label{obs:no_arc_in_interior_of_arc}
    Let $\mathcal{B}$ be a set of $\alpha$-circular polygons. Then in an optimal solution for~$\alpha'\geq\alpha$, no arc ends in the interior of a circular boundary arc.
\end{observation}
\begin{proof}
Let $e$ be a boundary arc of a polygon $B\in\polygons$. Then every e arc $f\in {F_\alpha}(\polygons)$ that ends in the interior of $e$ has to have \(\beta_\polygons(f,e) = \pi\). Since it has to be to the left of $e$ and has radius $\alpha'\geq\alpha$, the circle induced by $f$ completely contains $e$ (possibly on the boundary for $\alpha=\alpha'$). Thus, either the arcs $f$ and $e$ coincide or $e$ intersects $B$ properly and $e$ cannot be part of~${F_\alpha}(\polygons)$.
\end{proof}

\freeSubdivision*
\begin{proof}
    By~\Cref{lemma:relevantEdges:full}, the set~$\allowedArcs$ has size~$\bigO{n^2}$.
    Each pair of arcs in~$\allowedArcs \cup E(\polygons)$ intersects properly at most a constant number of times.
    With~$|E(\polygons)|=n$, it follows that the number of vertices and cells in~$\subdivision_C$ is in $\bigO{n^4}$.

    By~\Cref{prop:arc_properties:full}\done{, \Cref{lemma:distance2} and~\Cref{lem:infinity}}, as well as the definition of inclusion-maximal free pieces, there is an optimal solution~$\solution$ for~$\freeProblem$ whose inclusion-maximal free pieces are contained in~$\allowedArcs$.
    By construction, every boundary piece in~$\solution$ is a path in~$\subdivision_C$.
    Hence, $\solution$ is a solution for~$\problem{\subdivision_C}$.
    Every optimal solution for~$\problem{\subdivision_C}$ is also a solution for~$\freeProblem$.
    To be optimal for~$\problem{\subdivision_C}$, it must have the same objective value as~$\solution$, so it is also optimal for~$\freeProblem$.
\end{proof}

{Leveraging planar multi-source-multi-sink minimum cut algorithms \cite{DBLP:journals/siamcomp/BorradaileKMNW17,DBLP:conf/icalp/GawrychowskiK18} and a graph transformation adapted from~\cite{DBLP:journals/dcg/AbrahamsenGLR20}, we show in \cref{sec:Rottmann_transformation} that for any subdivision $\mathcal{D}$, the problem $\problem{\subdivision}$ can be solved in quasilinear-time in the size of $|\mathcal{D}|$. Using these results we get a polynomial-time algorithm for \freeProblem we call \emph{unconstrained polygon aggregation optimizer} algorithm~$\optimizer(\polygons)$ which is summarized in the following:

\begin{enumerate}
    \item Compute the arcs \( F_\alpha(\polygons) \).
    \item Build the subdivision \( \subdivision_C \) from \( F_\alpha(\polygons) \) and \( E(\polygons) \).
    \item Build the auxiliary graph from \( \subdivision_C \) (see \cref{sec:Rottmann_transformation}).
    \item Compute the minimum cut MC.
    \item Derive the regions of \solution from MC and $\subdivision_C$.
\end{enumerate}

In our preprocessing algorithms in Section~\ref{sec:preprocessing:appendix}, the relevant arcs ~$F_\alpha(\polygons)$ sometimes are already precomputed. Thus, we denote with~$\optimizer(\polygons,F_\alpha(\polygons))$ the algorithm that skips the first step and uses the pre-computed arcs.

\algorithm*
\begin{proof}
\done{Because every polygon vertex appears in~$\subdivision_C$, we have~$|\subdivision_C|\in\Omega(n)$.
Hence, the subdivision can be constructed in~$\bigO{|\subdivision_C| \log n} \subseteq \mathcal{O}\left(|\subdivision_C| \frac{\log^3|\subdivision_C|}{\log^2 \log |\subdivision_C|}\right)$ time by computing all the intersection points of~$\allowedArcs$ and~$E(\polygons)$.
Then, using \Cref{theorem:multisorucesink} in \Cref{sec:Rottmann_transformation}, an optimal solution~$\solution$ for~$\problem{\subdivision_C}$ can be computed in~$\mathcal{O}\left(|\subdivision_C| \frac{\log^3|\subdivision_C|}{\log^2 \log |\subdivision_C|}\right)$ time.
By~\Cref{lemma:freeSubdivision:paper}, this is an optimal solution for~$\freeProblem$.}
\end{proof}
The preprocessing algorithm presented in \cref{sec:preprocessing:appendix} sometimes computes intermediate instances in which some subproblems contain properly intersecting circular polygons. We therefore show that optimal solutions can still be computed in this setting without increasing the running time. 
To see that this holds even though the union of \(\polygons\) may contain a quadratic number of vertices and edges, we rely on one more insight derived from P4~in Proposition~\ref{prop:arc_properties:full}. 
\begin{observation}
    \label{cor:reflex_vertex_is_boring}
    Let $B \in \polygons$ be a circular input polygon and let $v$ be a vertex of $B$. If the exterior region angle satisfies $\beta_\polygons(v) < \pi$, then no arc $f \in F_\alpha(\polygons)$ ends at $v$.
\end{observation}
\begin{proof}
For every arc $f$ in $F_\alpha(\polygons)$ that ends  at a point $v$ on the input region~$B$, it holds that $\beta_\polygons(v)\geq\beta_\polygons(f, v)$. Hence, $\beta_\polygons(v) < \pi$ implies $\beta_\polygons(f,v) < \pi$. The claim then follows directly from Proposition~\ref{prop:arc_properties:full}.
\end{proof}
\begin{corollary}[Properly Intersecting Circular Polygons]\label{theorem:circ_polygons}
    Let $\polygons$ be a set of properly intersecting circular polygons  with~$|V(\polygons)|=n$.
    The subdivision~$\subdivision_{C}$ induced by the arcs $F_\alpha(\polygons)$ has size~$|\subdivision_{C'}|\in\bigO{n^4}$ and an optimal solution can be computed in~$\bigO{|\subdivision_{C'}| \frac{\log^3|\subdivision_{C'}|}{\log^2 \log |\subdivision_{C'}|}}\subseteq \bigO{n^4 \frac{\log^3n}{\log^2 \log n}}$ time.
\end{corollary}
\begin{proof}
Note that computing the union of the polygons takes time  $\bigO{n^2\log n}$. 
Let $\hat{\polygons}$ denote the union of the polygons in $\polygons$. It is sufficient to show $|F_\alpha(\hat{\polygons})|\in \bigO{n^2}$.  It may happen that $|V(\hat{\polygons})| = \bigO{n^2}$ and $|E(\hat{\polygons})| = \bigO{n^2}$, which could lead to $\bigO{n^4}$ possible arcs, but we show that the number of arcs in $F_\alpha(\hat{\polygons})$ is still in $\bigO{n^2}$.
First, observe that every newly introduced vertex $v \in V(\hat{\polygons})$ on some polygon $\hat{B}$ that does not coincide with an original input vertex must result from the intersection of the interiors of two edges $e_1\in E( B_1)$ and $e_2\in E( B_2)$ of some polygons $B_1$ and $B_2$. Since $B_1$ and $B_2$ overlap $\beta_{{\polygons}}(v)$ is not necessarily well defined for all points on $B_1$ and $B_2$. Thus, we abuse notation slightly and denote with $\beta_{B_i}(v)$ the exterior region angle of $v$ on $B_i$.
Then, we have $\beta_{\hat{\polygons}}(v)\leq\beta_{B_1}(v)$ as well as  $\beta_{\hat{\polygons}}(v)\leq \beta_{B_2}(v)$.
Since we have $\beta_{B_2}(v)=\beta_{B_1}(v)=\pi$, we get $\beta_{\hat{\polygons}}(v)\leq\pi$ and by Corollary~\ref{cor:reflex_vertex_is_boring} the only arcs in $F_\alpha(\hat{\polygons})$ that can end in $v$ have $\beta_{\hat{\polygons}}(v)=\pi$, i.e., they are tangential on $e_1$ or $e_2$.

A segment~$e \in E(\polygons)$ may be split into a set~$\hat{E}$ of $\bigO{n}$ edges in $E(\hat{\polygons})$.
For an arc $f$ to end in the interior of an edge~$\hat{e} \in \hat{E}$, it must be tangential to~$\hat{e}$.
Tangentiality to $\hat{e}$ is equivalent to being tangent to the inducing object of $e$, i.e., a circle or line.

Hence, it suffices to perform the tangentiality test for this object, which covers all edges in~$\hat{E}$ at once.
For a fixed other endpoint of~$f$ (either a vertex or an edge), there can only be a constant number of such arcs.
Thus, overall, we only obtain $\bigO{n^2}$ arcs of the types vertex-vertex, vertex-edge, and edge-edge.
\end{proof}
\subsection{Additional Structural Properties}

Finally, we discuss two nesting properties, which enable many engineering techniques. The first nesting properties considers how solutions change if we increase $\alpha$. Here we only present the result, since the proof follows directly from the proof of Lemma~2 in~\cite{rottmann2024bicritshapes}, which is applicable as long as intersections and unions of solutions also are solutions.

\alphaNestedness*
Moreover, our problem exhibits an additional nesting property. To prove it, we use the following lemma.
\begin{lemma}\label{lem:boundary-nestedness}
    {For pairwise interior-disjoint regions~$X,Y,Z$, it holds that~$\partial X \cap \partial Y \subseteq \partial X \cap \partial (Y \cup Z)$.}
\end{lemma}
\begin{proof}
    {We observe that~$\partial (Y \cup Z)$ contains all points in~$\partial Y$ except for the interior points of the shared boundary~$\partial Y \cap \partial Z$.
    If an interior point of~$\partial Y \cap \partial Z$ also lies on~$\partial X$, then~$X$ is not interior-disjoint with either~$Y$ or~$Z$.}
\end{proof}

\subsetNestedness*
\begin{proof}
    {Let $\hat{\solution}$ be an arbitrary but fixed solution for $\polygons$ minimizing $g_\alpha$.
    We divide~$\solution_{\polygons'}$ and~$\hat{\solution}$ into the pairwise interior-disjoint regions~$X:=\solution_{\polygons'} \setminus \hat{\solution}$, $Y:= \solution_{\polygons'} \cap \hat{\solution}$, and $Z:=\hat{\solution} \setminus \solution_{\polygons'}$ (see~\Cref{fig:SubsetNestedness:full}).
    Then we have~$\solution_{\polygons'}=X \cup Y$ and~$\hat{\solution}=Y \cup Z$.
    We show that the solution~$\solution_{\polygons} = X \cup Y \cup Z$, which obviously satisfies~$\solution_{\polygons'} \subseteq \solution_{\polygons}$, is optimal for~$\polygons$.}
        
    {For two interior-disjoint regions~$R,R'$, we have~$g_\alpha(R \cup R') = g_\alpha(R) + g_\alpha(R') - 2\alpha \cdot |\partial R \cap \partial R'|$.
    Hence, $g_\alpha(R \cup R') > g_\alpha(R)$ is equivalent to~$g_\alpha(R') > 2\alpha \cdot |\partial R \cap \partial R'|$.
    Assume that~$\solution_{\polygons}$ is not optimal, i.e., $g_\alpha(\solution_\polygons) > g_\alpha(\hat{\solution})$.
    Using~\Cref{lem:boundary-nestedness}, we obtain
    \begin{align*}
    &g_\alpha(X \cup Y \cup Z) = g_\alpha(\solution_\polygons) > g_\alpha(\hat{\solution}) = g_\alpha(Y \cup Z)\\
    \Rightarrow\quad&g_\alpha(X) > 2\alpha \cdot |\partial X \cap \partial (Y \cup Z)| > 2\alpha \cdot |\partial X \cap \partial Y|\\
    \Rightarrow\quad&g_\alpha(\solution_{\polygons'}) = g_\alpha(X \cup Y) > g_\alpha(Y),
    \end{align*}
    which contradicts the assumption that $\solution_{\polygons'}$ minimizes~$g_\alpha$.
    }
\end{proof}
\begin{figure}
    \centering
    \includegraphics{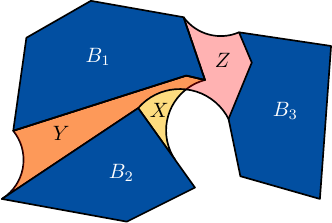}
    \caption{{Visualization of the proof of \Cref{lemma:subset_nestedness:paper} with $\polygons = \{B_1,B_2,B_3\}$ and $\polygons'=\{B_1, B_2\}$.
    It is assumed that~$\solution_{\polygons'} = X \cup Y$ is an optimal solution for~$\polygons'$ and~$\hat{\solution} = Y \cup Z$ is an optimal solution for~$\polygons$.
    If~$\solution_\polygons = X \cup Y \cup Z$ is not an optimal solution for~$\polygons$, i.e., removing~$X$ improves the objective value, then removing~$X$ from~$X \cup Y$ also improves the objective value and thus~$\solution_{\polygons'}$ is not optimal.}}
    \label{fig:SubsetNestedness:full}
\end{figure}
Subset-nestedness also holds in the case where an input polygon grows, because this is equivalent to adding new polygons that touch the old one. This property is also useful when polygons are added sequentially (e.g., in historical data where buildings are added over time).
In this case, subset-nestedness implies nestedness over time.

\section{A Local Preprocessing Algorithm (Detailed Overview)}\label{sec:preprocessing:appendix}
The main contributor to the runtime of \optimizer is the size of the subdivision $\subdivision_C$, which is quadratic in the size of the arc set~$F_\alpha(\polygons)$. Preliminary experiments show that the number of arcs in $F_\alpha(\polygons)$ is manageable for small $\alpha$, but increases significantly with increasing $\alpha$. Leading to subdivisions that cannot be handled in a reasonable time frame for sufficiently large $\alpha$ values. This growth is due to the fact that more and more pairs of potential endpoints fall below the distance limit of~$2\alpha$ imposed by property~P1.
We observe that this is counteracted by another trend: due to $\alpha$-nestedness, the regions of an optimal solution grow with~$\alpha$, and therefore more and more arcs become irrelevant because their endpoints already lie in the interior of a region or because they intersect another region.
Although the exact regions are not known until the instance is solved, we can exploit subset-nestedness to approximate them from below.

In this section, we introduce a preprocessing algorithm that iteratively selects sets of nearby polygons as a sub-instance, solves it, and replaces the polygons with the computed optimal representative regions.
If it succeeds in grouping large parts of the instance, then many arcs that connect distant polygons are never generated.
The following corollary of subset-nestedness forms the cornerstone of our preprocessing algorithm:

\begin{corollary}\label{cor:subset_optimality_paper}
    Let $\polygons=X_1\sqcup \dots\sqcup X_k $ be a set of (intersecting) circular polygons decomposed into a disjoint union of subsets. Let $\solution_1,\dots, \solution_k$ be optimal solutions to $\freeProblem$ for $X_1,\dots,X_k$. Then every optimal solution of $\freeProblem$  for $\solution_1\cup\dots\cup\solution_k$ is also optimal for $\polygons$. 
\end{corollary}
Building on this decomposition property, our preprocessing algorithm aims to identify subsets of the input polygons such that optimizing the induced sub-instance reduces the number of regions that need to be considered. We say that a subset $\polygons' \subseteq \polygons$ can be merged if the solution $\mathcal{S} = \optimizer(\polygons',F_\alpha(\polygons'))$ consists of exactly one representative region.
It is not clear how to find subsets $\mathcal{B}'$ that can be merged.
Our approach derives a family of candidate subsets and optimizes each one individually.
If a subset is successfully merged during optimization, we replace it in the overall instance with the resulting optimal representative.

\subparagraph{Selecting Subsets.} 
For the sake of efficiency, we focus on subsets for which the induced sub-instances have small subdivisions.  
First, we consider the singleton \( \{B\} \) for every polygon \( B \in \polygons \), which may already simplify~$B$.
Then, we consider all pairs \( \{B_i, B_j\} \) of polygons such that the edge \( \{C_i, C_j\} \) appears in the Delaunay triangulation \( DT(\mathcal{C}) \) of the centroids of the polygons \(\polygons\). This is motivated by the assumption that polygons in close proximity are more likely to be merged. To avoid unnecessary overhead for small \(\alpha\) values, we only consider pairs satisfying \( \LL(\overline{C_iC_j}) \leq 2\alpha \).
\subparagraph{Adapting \textsc{UPA-Opt}.}
In addition to the solution~$\solution$, we ensure that \( \optimizer(\polygons', F_\alpha(\polygons')) \) returns a flag indicating whether \( \polygons' \) was merged, and, if true, the set \( A =F_\alpha(\solution)\) of useful arcs for the instance~$\solution$.
The latter can be computed in a single pass over the subdivision. The min-cut algorithm labels each cell of the subdivision as either part of the solution or free. Each arc \( f \in F_\alpha(\polygons') \) has one of three types:
(1) every (half-)edge of \( f \) separates a solution cell from a free cell, (2) all associated edges are between free cells, or (3) at least one edge is between two solution cells.
Type-1 arcs form the boundary of \( \solution \), type-2 arcs are fully free and connect edges of \( \solution \), and type-3 arcs intersect \( \solution \); see Figure~\ref{fig:arc_types_paper}. Thus, \( F_\alpha(\solution) \) consists exactly of the type-2 arcs.
\begin{figure}[t]
    \centering
    \includegraphics{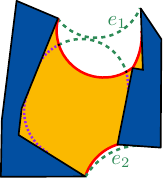}
    \captionsetup{textformat=simple}
    \caption{The different types of arcs. We have $F_\alpha(\solution)=\{e_1,e_2\}$.}
    \label{fig:arc_types_paper}
\end{figure}

\begin{algorithm2e}
\caption{Local preprocessing algorithm}
\label{algo:preprocessing_paper}
\KwIn{Set of polygons $\polygons$, parameter $\alpha$}
\KwOut{Set of (merged) representative circular polygons $\mathcal{R}$}
\textbf{Initialize:} Union-find $\mathfrak{U}$, arc sets $\mathcal{A}$, ignore sets $\Phi$, Delaunay triangulation~$DT(\mathcal{C})$\;
\ForEach{$B \in \polygons$\label{algo:start_singlepp_paper}}{
    $F_\alpha \gets \ComputeUsefulArcs(B,\alpha)$\;
    $\solution,\; {A},\; \text{succeeded} \gets \optimizer(B,F_\alpha)$\;
    $\mathfrak{U}.\makeSet(\{B\}, \solution)$, 
    $\mathcal{A}(\solution) \gets A$, 
    $\Phi(\solution) \gets \emptyset$\;\label{algo:end_singlepp_paper}
}
\ForEach{edge $\{C_i, C_j\} \in DT(\mathcal{C})$ \label{algo:start_allpp_paper}}{
    $B_i \gets \GetPoly(C_i)$, 
    $B_j \gets \GetPoly(C_j)$\;
    $R_1 \gets \mathfrak{U}.\find(
     B_i)$, 
     $R_2 \gets \mathfrak{U}.\find( B_j)$\;
    \lIf{$(R_1 \in \Phi(R_2) \textbf{ and } R_2 \in \Phi(R_1))$ \textbf{or} $R_1 = R_2$} {\Continue}
    ${F_\alpha} \gets \ComputeUsefulArcs(R_1,R_2,\mathcal{A}(R_1),\mathcal{A}(R_2),\alpha)$\;
    $\solution,\;{A},\;\text{succeeded}\gets \optimizer(\{R_1, R_2\},{F_\alpha})$\;
    \If{\textbf{not} succeeded}{
        $\Phi(R_1) \gets \Phi(R_1) \cup \{R_2\}$, 
        $\Phi(R_2) \gets \Phi(R_2) \cup \{R_1\}$\;
    }
    \Else{
        $\mathfrak{U}.\unionFn( R_1, R_2, \solution)$, 
        $\mathcal{A}(\solution) \gets A$, 
        $\Phi(\solution) \gets \emptyset$\;
    \label{algo:end_allpp_paper}}
}
\Return representatives of~$\mathfrak{U}$
\end{algorithm2e}

\subparagraph{Processing Scheme.}
We update the representative regions on the fly by tracking merged polygons with a union-find data structure~$\mathfrak{U}$.
Each set in~$\mathfrak{U}$ is a subset~$\polygons'$ of input polygons for which the optimal solution consists of a single region~$R$, which we use as the representative for~$\polygons'$ in~$\mathfrak{U}$.
Additionally, we maintain the set \( \mathcal{A}(R) = F_\alpha(\{R\}) \) and the \emph{ignore set} \( \ignoreset{R} \) of representatives with which merge attempts have already failed.
The processing scheme is given in~\Cref{algo:preprocessing_paper}. We start with the singletons (lines~\ref{algo:start_singlepp_paper}--\ref{algo:end_singlepp_paper}): For each polygon~\( B \), we compute \( F_\alpha(B) \) and run \( \optimizer \). The result~\( \solution \) becomes the representative of \( B \) in~$\mathfrak{U}$, with \( \mathcal{A}(\solution) \) derived from \optimizer. Initially, the ignore set of~$\solution$ is empty.

Afterwards, we process the edges of the Delaunay triangulation \( DT(\mathcal{C}) \) in increasing order of length (lines~\ref{algo:start_allpp_paper}--\ref{algo:end_allpp_paper}). For each edge \( \{C_i, C_j\} \), we retrieve the representatives $R_1$ and~$R_2$ of~$B_i$ and~$B_j$ in~$\mathfrak{U}$. If \( R_1 = R_2 \) or both are in each other's ignore sets, we skip the pair.
Otherwise, we compute the useful arcs \( F_\alpha \) and call \( \optimizer \).
Note that the call to \( \FuncSty{CompUArcs} \) computes the useful arcs between \( R_1 \) and \( R_2 \) and performs intersection tests, as all other arcs are already precomputed in \( \mathcal{A}(R_1) \) and \( \mathcal{A}(R_2) \). 
If the merge fails, we add~\( R_1 \) to~\( \ignoreset{R_2} \) and vice versa. If it succeeds with result~\( \solution \), we merge the components in~$\mathfrak{U}$, set~\( \solution \) as the new representative, update \( \mathcal{A}(\solution) \), and reset \( \ignoreset{\solution} \). We repeat until all edges have been processed.
The correctness of this scheme is summarized in the following proposition. 
\begin{proposition}
Let \( R_1, \dots, R_k \) be the final representatives and let \( \mathcal{B}_R \) be the set of polygons merged into the representative \( R \). Then \( \polygons = \mathcal{B}_{R_1} \sqcup \dots \sqcup \mathcal{B}_{R_k} \), and each \( R_i \) is an optimal solution of \freeProblem for \( \mathcal{B}_{R_i} \).
\end{proposition}
With \cref{cor:subset_optimality_paper}, it follows that optimizing the representatives still leads to optimal solutions for the initial input $\polygons$.
Despite only performing local merges of neighboring polygons, the algorithm has global effects by continuously updating and comparing representatives instead of original polygons. For large~\( \alpha \), preprocessing alone can already yield (near-)optimal representative sets.
Note that, in the worst case, the preprocessing may involve~\( \bigO{n} \) candidate sets, each with subdivision complexity~\( \bigO{n^4} \), increasing the runtime by a factor of~\( n \). However, in practice, the subproblems are small and solved efficiently, as shown in our experiments.

\section{Solving the Parametric Problem (with Formal Definitions and Proofs)}\label{sec:parametric:appendix}
So far, we have only studied the problem~$\freeProblem$ for fixed~$\alpha$.
For the subdivision-restricted variant, Rottmann et al.~\cite{rottmann2024bicritshapes} also study the parametric problem, where the objective is to find an optimal solution for each~$\alpha\in[0,\infty]$.
They show that due to~$\alpha$-nestedness, there can only be~$\mathcal O(n)$ optimal solutions, which allows the problem to be solved in~$\bigtO{n^2}$ time.
Although $\alpha$-nestedness also holds in the unrestricted variant, this does not imply a bound on the number of optimal solutions.
Because our definition of solutions is geometric, even an infinitesimal change in~$\alpha$ causes a change in the optimal solution because the radius of the circular arcs changes.
Hence, the number of optimal geometric solutions is infinite. 

We address the problem by deriving a description of a solution and its components that does not depend on any geometry and is strictly given combinatorially with respect to the polygon edges (and vertices).
For example, given two edges~$e,h$ and two values~$\alpha,\alpha'$, the circular arcs~$\sarc{\alpha}{eh}$ and~$\sarc{\alpha'}{eh}$ are different geometric realizations of the same combinatorial arc~$\sarc{}{eh}$.
This allows us to define a parametric variant~$\freeProblem$, which asks for a set of combinatorial solutions.
We then show that the number of combinatorial solutions in an optimal solution is in~$\mathcal{O}(n^2)$.
We use this to design a dichotomic scheme for the parametric problem, which runs in polynomial time under the assumption that we have an oracle that computes the (unique) intersection point of two specific functions.
We also present a fully polynomial-time approximation scheme~(FPTAS) that does not require an oracle.
Given a parameter~$\varepsilon>0$, it computes a~$(1+\varepsilon)$-approximation of the optimal solution for every value of~$\alpha$ and runs in polynomial time in the input size and $1/\varepsilon$. 

For the remainder of this paper, we restrict the input to straight-line polygons, which are the most relevant input type for applications. While a few of the more technical results rely on this assumption, we believe that they can be extended to more general settings, e.g., circular polygons.

\subsection{Combinatorial Solutions}
We build our definition of combinatorial solutions incrementally, starting with combinatorial representations of arcs and other boundary primitives.
Next, we define the combinatorial counterpart to a closed boundary curve.
Building on this, we define combinatorial regions and solutions.
For each combinatorial object, we explain how a geometric realization can be obtained for a fixed value of $\alpha$, and we show that the set of~$\alpha$ values for which a valid realization exists is a single open interval.

For an instance~$\polygons$ of~$\freeProblem$, we denote by~$\possibleArcs$ the set of circular arcs~$f=\sarc{\alpha}{uv}$ such that the endpoints~$u$ and~$v$ lie on the input boundary,~$f$ fulfills properties~P1--P5 of~\Cref{prop:arc_properties:full} and does not connect parallel edges.
Note that this definition is identical to~\Cref{def:arcs}, except for the fact that we do not require that~$f$ does not intersect~$\polygons$ properly.
Hence, we have~$F_\alpha(\polygons) \subseteq \possibleArcs$.
For $x,y \in V(\polygons) \cup E(\polygons)$, we define the \emph{combinatorial arc} $\sarc{}{xy}$ between $x$ and $y$ as a representative for the family of arcs~$\sarc{\alpha}{xy} \in \possibleArcs$ for each~$\alpha\in(0,\infty)$. We denote with $\hat{F}(\polygons)$ the set of all combinatorial arcs.
We denote its \emph{combinatorial endpoints} on $x$ and $y$, respectively, with~$u(x,y)$ and~$v(x,y)$.
A \emph{combinatorial vertex} is either a combinatorial endpoint or an input vertex $v\in V(\polygons)$.
A \emph{combinatorial edge segment} between combinatorial vertices~$u$ and~$v$ that lie on the same polygon edge~$e$ is the (directed) subsegment of~$e$ from~$u$ to~$v$.
A \emph{combinatorial boundary piece} is either a combinatorial arc or a combinatorial edge segment.

For fixed $\alpha$, we define the \emph{realizations} of combinatorial objects (which may or may not exist): (1) For a combinatorial arc~$\sarc{}{xy}$ with combinatorial endpoints~$u$ and~$v$, its realization is the unique arc $\sarc{\alpha}{xy} \in \possibleArcs$ connecting $x$ and $y$ (if it exists), and the realizations~$u_\alpha$ and~$v_\alpha$ are the respective endpoints of the realized arc $\sarc{\alpha}{xy}$, (2) the realization~$v_\alpha$ of a polygon vertex~$v$ is~$v$ itself, and (3) the realization of a combinatorial edge segment between combinatorial vertices~$x$ and~$y$ is the line segment~$\overrightarrow{x_\alpha y_\alpha}$ connecting the realizations of the vertices (if they exist).
To determine the values of~$\alpha$ for which a combinatorial arc has a realization, we first characterize how the realized arc and the positions of its endpoints vary depending on $\alpha$.
Given a pair~$X,Y$ of geometric objects (e.g., points, lines, circles), we denote their Euclidean distance by \(d(X, Y) = \min_{x \in X,\, y \in Y} \LL(\overline{xy})\).

\begin{lemma}\label{lem:arc-movement}
    Let~$f=\sarc{}{xy}$ be a combinatorial arc with combinatorial endpoints~$u$ and~$v$, let~$\alpha\in(0,\infty)$ a value for which the realization~$f[\alpha]$ exists, and let~$\theta(f[\alpha])$ denote the central angle of~$f[\alpha]$.
    If~$x$ (resp. $y$) is a polygon edge~$\overrightarrow{ab} \in E(\polygons)$, then~$d(a,u[\alpha])$ (resp. $d(b,v[\alpha])$) is strictly decreasing as~$\alpha$ increases.
    The conditions~$d(u[\alpha],v[\alpha]) < 2\alpha$ and~$\theta(f[\alpha])<\pi$ hold for any value of~$\alpha$ such that~$f$ has a realization and the following condition is satisfied:
    \begin{itemize}
        \item If~$x$ and~$y$ are both vertices: $d(x,y) < 2\alpha$.
        \item If~$x$ is a vertex and~$y$ is an edge: $d(x,L_y) < 2\alpha$, where~$L_y$ is the line through~$y$,
        \item If~$x$ is an edge and~$y$ is a vertex: $d(L_x,y) < 2\alpha$, where~$L_x$ is the line through~$x$,
        \item If~$x$ and~$y$ are both edges: no additional condition.
    \end{itemize}
\end{lemma}
\begin{proof}
    If~$x$ and~$y$ are both vertices, then the claim is trivially true because~$u[\alpha]=x$ and~$y[\alpha]=y$.
    Assume w.l.o.g.\ that~$x$ is an edge~$\overrightarrow{ab}$.
    We rotate, move and mirror our coordinate system such that the line~$L_x$ through~$x$ coincides with the~$x$-axis and~$\overrightarrow{ab}$ goes from right to left.
    This means that~$f$ lies above it and is oriented clockwise.
    Let~$z$ be the value such that~$u[\alpha]=(z,0)$.
    Then, the center point of~$f[\alpha]$ is given by~$C=(z,\alpha)$.
    We derive a formula for~$z$ and show that it is strictly increasing with~$\alpha$.
    This means that~$u[\alpha]$ moves to the right, towards~$a$.

    Consider the case that~$y=v$ is a vertex.
    We move the coordinate system such that~$v=(0,d)$, where~$d=d(L_x,v)$.
    Because~$f[\alpha]$ is oriented clockwise, the central angle is smaller than~$\pi$ if and only if~$z > 0$.
    Furthermore, the arc must satisfy~$d(C,v)=\alpha$.
    With~$d(C,v)=\sqrt{z^2 - (\alpha - d)^2}$, we obtain the positive solution~$z=\sqrt{2 \alpha d - d^2}$, which is strictly increasing with~$\alpha$.
    The distance between the two endpoints is~$d(u[\alpha],v)=\sqrt{2 \alpha d}$ and the central angle of~$f[\alpha]$ is given by~$\theta(f[\alpha])=2\arcsin{\frac{\sqrt{2 \alpha d}}{2\alpha}} =2\arcsin{\sqrt{\frac{d}{2\alpha}}}$.
    Hence, we have~$d(u[\alpha],v[\alpha]) < 2\alpha$ and~$\theta(f[\alpha]) < \pi$ iff~$d(L_x,v) < 2\alpha$.

    Consider now the case that~$y$ is an edge.
    Let~$\gamma \in (0,\pi)$ be the angle between the lines~$L_x$ and~$L_y$ through~$x$ and~$y$, respectively.
    Note that~$L_x$ and~$L_y$ cannot be parallel because this would imply that~$f[\alpha]$ has central angle~$\pi$ and is therefore not contained in~$\possibleArcs$.
    Because~$f[\alpha]$ must lie to the right of both edges, it follows that~$z>0$.
    We move the coordinate system such that the intersection point of the two lines is the origin~$O=(0,0)$.
    Then the normal vector of~$L_y$ is given by~$\vec{n}_y=\vecTwo{-\sin \gamma}{\cos \gamma}$.
    The arc must satisfy~$d(C,L_y)=\alpha$.
    With
    \begin{align*}
        d(C,L_y) &= \frac{|\vec{C} \cdot \vec{n}_v|}{||\vec{n}_v||} = \left|\vecTwo{z}{\alpha} \cdot \vecTwo{-\sin\gamma}{\cos\gamma}\right|\\
        &= |\alpha\cos\gamma - z\sin\gamma|,
    \end{align*}

    we obtain the solutions
    \begin{align*}
    z_1 &= -\alpha \cdot \frac{1-\cos\gamma}{\sin\gamma} = -\alpha \tan(\gamma/2),\\
    z_2 &= \alpha \cdot \frac{1+\cos\gamma}{\sin\gamma} = \alpha \cot(\gamma/2).       
    \end{align*}
    Only~$z_2$ is positive, so we have~$z=\alpha \cot(\gamma/2)$, which is strictly increasing with~$\alpha$.
    Now, consider the quadrilateral formed by~$O$, $C$, $u[\alpha]$ and $v[\alpha]$.
    The inner angles at~$u[\alpha]$ and~$v[\alpha]$ are~$\pi/2$ because~$f[\alpha]$ is tangential on~$x$ and~$y$, whereas the inner angle at~$O$ is~$\gamma$ and the inner angle at~$C$ is~$\theta(f[\alpha])$.
    Hence, we have~$\theta(f[\alpha]) = \pi - \gamma \in (0,\pi)$, which is independent of~$\alpha$.
    The distance of the endpoints is the chord length of~$f[\alpha]$, which is given by~$d(u[\alpha],v[\alpha])=2\alpha\sin(\theta(f[\alpha])/2) < 2\alpha$.
\end{proof}

Using this, we can characterize the~$\alpha$ values for which a combinatorial arc has a realization.

\begin{lemma}\label{prop:arc-validity-range}
    For a combinatorial arc~$\sarc{}{xy}$ with~$x,y \in V(\polygons) \cup E(\polygons)$, the set~$\mathcal{I}(\sarc{}{xy})\subseteq(0,\infty)$ of values for which~$\sarc{}{xy}$ has a realization~$\sarc{\alpha}{xy} \in \possibleArcs$ is a single open interval.
\end{lemma}
\begin{proof}
    Let~$u$ and~$v$ denote the combinatorial endpoints of~$\sarc{}{xy}$.
    The following conditions must be satisfied in order for the realization~$\sarc{\alpha}{xy}$ to exist:
    \begin{itemize}
        \item If~$x$ is a polygon edge, then~$u$ must lie in the interior of~$x$ (and likewise for~$y$ and~$v$). By~\Cref{lem:arc-movement}, the movement of~$u$ on~$x$ is strictly monotonous in~$\alpha$, so there is a single interval in which this condition is satisfied. The interval is open because the interior of~$x$ is open.
        \item P1: The distance between the endpoints must be smaller than~$2\alpha$. By~\Cref{lem:arc-movement}, the set of~$\alpha$ values for which this is the case is an open interval.
        \item P3: The central angle must be smaller than~$\pi$. By~\Cref{lem:arc-movement}, this holds iff~P1 is satisfied.
    \end{itemize}
   Conditions~P2 and~P4 are satisfied by definition.
\end{proof}

Now that we have established the combinatorial counterparts to constrained and free boundary pieces, we can connect them to form combinatorial boundary curves.

\begin{definition}[Combinatorial boundary]
 A combinatorial boundary $B$ is a cyclic sequence $v_0,v_1,\dots,v_{n-1},v_n=v_0$ of combinatorial vertices such that for each~$0 \leq i < n$, there is a combinatorial boundary piece~$(v_i,v_{i+1})$.  The combinatorial boundary~$B$ is realizable for~$\alpha$ if all boundary pieces are realizable. In that case, the realization of~$B$ is given by the concatenation of the realized combinatorial boundary pieces.
\end{definition}

A closed region can be described via an outer boundary curve and a set of inner boundary curves that enclose holes.
We can define a combinatorial region in an analogous manner, using combinatorial boundary curve.
However, this leads to issues because the realized boundaries may form a region for some values of~$\alpha$ but become self-intersecting or intersect each other for others.
For the following proofs, we require a notion analogous to a region that is well-defined even in those cases.

A \emph{pseudo-region} is a set of closed curves, one of which is designated as the \emph{outer boundary} and is oriented counterclockwise, and rest as \emph{inner boundaries}, which are oriented clockwise. The pseudo-region forms a region if none of the boundaries are self-intersecting, no pair of boundaries intersects properly, the inner boundaries are enclosed by the outer boundary, and none of the inner boundaries enclose each other.
\begin{definition}[Winding number and interior]
For a point~$p$ and a closed curve~$f$, the \emph{winding number}~$\wind(p,f)$ is defined as follows. Let~$r$ be a ray that extends from~$p$ into an arbitrary direction. Then~$\wind(p,f)$ is the number of times that~$f$ crosses~$r$ from right to left, minus the number of times that~$f$ crosses~$r$ from left to right.
The definition extends naturally to sets of closed curves and thereby (pseudo-)regions.
The interior of pseudo-region~$R$ is the set of points~$p$ with~$\wind(p,R)>0$.
\end{definition}
It is a direct consequence of the Jordan curve theorem~\cite{jordan} that the interior of a region~$R$ is contiguous and has winding number 1.
The exterior has winding number 0 but is not necessarily contiguous because it also includes the holes enclosed by the inner boundary curves.
In a pseudo-region, there may be points with winding number~$>1$ and~$<0$.
Intuitively, the former are "included multiple times" in the pseudo-region, whereas the latter are "excluded multiple times".
\begin{definition}[Combinatorial region]
A combinatorial region $R$ is a set of combinatorial boundaries, of which one is designated as the outer boundary curve and the rest as the inner boundary curves. The realization of~$R$ for~$\alpha$ is the pseudo-region given by the set of realized boundary curves.
The combinatorial region $R$ is \emph{valid} for~$\alpha$ if
    \begin{enumerate}[(i)]
        \item all combinatorial boundaries of~$R$ are realizable,
        \item the realization of~$R$ is a region and
        \item none of the boundaries touch each other.
    \end{enumerate}
\end{definition}

Using the concepts of pseudo-regions and winding number, we can derive a nesting property for realizations of the same combinatorial region.

\begin{lemma}\label{lem:winding-number-increasing}
    Let~$R$ be a combinatorial region and~$(\amin,\amax)$ the interval in which all combinatorial boundary arcs of~$R$ have realizations.
    For~$\alpha_1 \leq \alpha_2 \in (\amin,\amax)$ and any point~$p$ in the plane, it holds that~$\wind(p,R[\alpha_1]) \leq \wind(p,R[\alpha_2])$.
\end{lemma}
\begin{proof}
    We show that if we gradually transform~$R[\alpha_1]$ into~$R[\alpha_2]$ by moving the combinatorial arcs of~$R$ one at a time, the winding number does not decrease.
    Consider a combinatorial arc~$f$ with combinatorial endpoints~$u$ and~$v$.
    If~$u$ is a polygon vertex, then~$u[\alpha_1]$ and~$u[\alpha_2]$ coincide at~$u$.
    If~$u$ lies on an edge, then it follows from~\Cref{lem:arc-movement} that~$u[\alpha_2]$ lies to the right of~$u[\alpha_1]$ with respect to~$f[\alpha_1]$.
    The same holds for~$y$ and~$v$.
    Let~$h$ denote the concatenation of~$\overrightarrow{u[\alpha_2]u[\alpha_1]}$, $f[\alpha_1]$ and $\overrightarrow{v[\alpha_1]v[\alpha_2]}$.
    Because~$f[\alpha_2]$ has a larger radius that~$f[\alpha_1]$, it follows that~$h$ and~$f[\alpha_2]$ enclose a contiguous region~$R^*$ that is to the right of~$h$ and to the left of~$f[\alpha_2]$.
    Now, consider a ray~$r$ originating in~$p$ that does not cross~$h$.
    If~$p$ lies inside~$R^*$, then~$f[\alpha_2]$ crosses~$r$ from right to left.
    Hence, the winding number increases by~1.
    If~$p$ lies outside of~$R^*$, then~$f[\alpha_2]$ does not cross~$r$, so the winding number is unchanged.
\end{proof}

\begin{corollary}\label{cor:combinatorial-nested}
    Let~$R$ be a combinatorial region and~$(\amin,\amax)$ the interval in which all combinatorial boundary arcs of~$R$ have realizations.
    The realizations of~$R$ are nested for increasing~$\alpha$, i.e., $R[\alpha_1] \subseteq R[\alpha_2]$ for~$\alpha_1 \leq \alpha_2 \in (\amin,\amax)$.
\end{corollary}
\begin{proof}
This follows from~\Cref{lem:winding-number-increasing} and the fact that the interior of a realization consists of the points with winding number $>0$.
\end{proof}

We use these results to show that if a combinatorial region becomes invalid when increasing or decreasing~$\alpha$, it will stay invalid as~$\alpha$ increases or decreases further.

\begin{lemma}\label{prop:region-validity-range}
    For a combinatorial region~$R$, the set~$\mathcal{I}(R)\subseteq(0,\infty)$ of values for which~$R$ is valid is a single open interval.
\end{lemma}
\begin{proof}
    By~\Cref{prop:arc-validity-range}, the set of values for which all combinatorial boundary arcs have realizations is a single open interval, which we denote by~$(\amin,\amax)$.
    We fix a value~$\alpha^* \in (\amin,\amax)$ for which~$R[\alpha^*]$ is valid.
    Let~$\alpha_\text{high}$ be the lowest value in~$(\alpha^*,\amax)$ for which two boundary curves~$e$ and~$h$ intersect in some point~$x$, or~$\alpha_\text{high}=\amax$ if no such value exists.
    It follows from the strictly monotonous movement of the combinatorial arcs shown in~\Cref{lem:arc-movement} that~$x$ is in the exterior of~$R[\alpha]$ for~$\alpha\in(\alpha^*,\alpha_\text{high})$ and moves to the interior with respect to both~$e$ and~$h$ as~$\alpha$ crosses~$\alpha_\text{high}$.
    Hence, $\wind(x,R[\alpha])$ increases from~$0$ to~$2$ and~$R[\alpha]$ becomes invalid.
    By~\Cref{lem:winding-number-increasing}, the increase of~$\wind(x,R[\alpha])$ is monotonous in~$\alpha$, so~$R[\alpha]$ is invalid for~$[\alpha_\text{high},\amax)$.

    A symmetrical argument can be applied for decreasing~$\alpha$.
    Let~$\alpha_\text{low}$ be the highest value in~$(\amin,\alpha^*)$ for which two boundary curves~$e$ and~$h$ intersect in some point~$x$, or~$\alpha_\text{low}=\amin$ otherwise.
    Then~$\wind(x,R[\alpha])$ decreases from~$1$ to~$-1$ as~$\alpha$ crosses~$\alpha_\text{low}$, and thus~$R[\alpha]$ is invalid for~$(\amin,\alpha_\text{low}]$.
    Overall, we have~$\mathcal{I}(R) = (\alpha_\text{low},\alpha_\text{high})$.
\end{proof}

\begin{figure}
\centering
\includegraphics{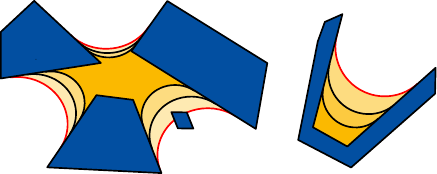}
\caption{Three realizations of the same combinatorial solution for different~$\alpha$ values. The arcs move with~$\alpha$ but still connect the same edge pairs. The realization marked in red is invalid because it intersects a polygon.}
\label{fig:combinatorial}
\end{figure}

The definition of a combinatorial solution is straightforward; see~\Cref{fig:combinatorial}.
To distinguish combinatorial solutions from their realizations, we also refer to the latter as \emph{geometric solutions}.

\begin{definition}[Combinatorial solution]
A combinatorial solution is a set of combinatorial regions. It is \emph{valid} for~$\alpha$ if all combinatorial regions are valid, the realized regions form a solution for~$\freeProblem$ and no realized boundary arc intersects~$\polygons$ properly.
\end{definition}

\solutionValidityRange*
\begin{proof}
    By~\Cref{prop:region-validity-range}, the set of values for which all combinatorial regions of~$\combinatorialsol$ have valid realizations is a single open interval, which we denote by~$(\amin,\amax)$.
    Inside this interval, the realized regions form a solution if they are pairwise disjoint, they cover all polygons, and no realized boundary arc intersects~$\polygons$ properly.
    Let~$\alpha_r$ be the smallest value in~$(\amin,\amax)$ for which two realized regions~$R_1$ and~$R_2$ touch, or~$\alpha_r=\amax$ if there is no such value. By~\Cref{cor:combinatorial-nested}, the overlapping area of~$R_1$ to~$R_2$ grows monotonously in~$\alpha$, so the realized regions of~$\combinatorialsol$ are pairwise disjoint exactly in~$(\amin,\alpha_r)$.

    If a region~$R$ covers a polygon~$B$ only partially, then there is a boundary arc that intersects~$B$ properly. Therefore, each input polygon~$B$ must be fully covered by a single region.
    We can interpret~$B$ as a combinatorial region, which is valid in~$(0,\infty)$. Hence, the portion of~$B$ that is covered by a combinatorial region~$R$ grows with increasing~$\alpha \in \mathcal{I}(R)$ by~\Cref{cor:combinatorial-nested}.
    Let~$\alpha_p$ be the smallest value in~$(\amin,\alpha_r)$ such that every polygon is fully covered by a single region in~$\combinatorialsol[\alpha_p]$, or~$\alpha_p=\alpha_r$ if there is no such value.
    Because~$\alpha_p$ is chosen as the smallest such value, there must be at least one polygon~$B$ that is touched by a boundary arc of the region in~$\combinatorialsol[\alpha_p]$ that covers~$B$.
    By~\Cref{lem:arc-movement}, the movement of combinatorial arcs is strictly monotonous and outward with increasing~$\alpha$.
    Hence, $\combinatorialsol[\alpha_p]$ is invalid, but for~$\alpha\in(\alpha_p,\alpha_r)$, no boundary arc in~$\combinatorialsol[\alpha]$ intersects a covered polygon.
    In particular, no boundary arc intersects a polygon that is covered by a different region, because this would imply that the two regions are not disjoint.
    Hence, we have~$\mathcal{I}(\combinatorialsol) = (\alpha_p,\alpha_r)$.
\end{proof}

\subsection{The Parametric Problem}
Using the combinatorial solutions, we define the parametric problem.

\begin{problem_description}[The Exact Parametric Problem]\label{problem:problem_2}
    Let~$\polygons$ be a set of interior-disjoint polygons in the plane. A solution to the parametric unrestricted polygon aggregation problem $\biProblem$ is a set $\mathfrak{K}$ of combinatorial solutions such that for every $\alpha\in[0,\infty]$, there exists a combinatorial solution $\combinatorialsol\in\mathfrak{K}$ such that $\combinatorialsolevaluated{\alpha}$ is an optimal solution for $\freeProblem$. 
\end{problem_description}
Using our combinatorial view and exploiting the $\alpha$-nestedness property (see~\cref{prop:alpha-nested:paper}), we can show that the set $\mathfrak{K}$ has polynomial size.
\parametricSolutions*
\begin{proof}
Let~$\combinatorialsol_1,\dots,\combinatorialsol_k$ be the solutions in~$\mathfrak{K}$, sorted in ascending order of the lowest~$\alpha$ value for which they are optimal.
Note that~$k$ is finite because the number of possible combinatorial solutions is finite.
When switching from a combinatorial solution~$\combinatorialsol_i$ to its successor~$\combinatorialsol_{i+1}$, at least one combinatorial arc must be removed or added.
The total number of combinatorial arcs is in~$\mathcal O(n^2)$, and each of them can appear for the first time only once.
Below, we show that if a combinatorial arc from~$\combinatorialsol_i$ does not appear in~$\combinatorialsol_{i+1}$ but in another combinatorial solution~$\combinatorialsol_j$ with~$j > i+1$, then~$\combinatorialsol_j$ consists of strictly fewer regions than~$\combinatorialsol_i$.
Because every region in an optimal solution covers at least one polygon, the number of regions in~$\combinatorialsol_1$ is in~$\mathcal O(n)$.
Due to~$\alpha$-nestedness (\Cref{prop:alpha-nested:paper}), it cannot increase between subsequent solutions.
Hence, the removal and reappearance of a combinatorial arc can occur at most~$\mathcal O(n)$ times, summed across all combinatorial arcs.
Hence, $\mathfrak{K}$ has size~$\mathcal O(n^2)$.

For~$\alpha\in[0,\infty]$, let~$\combinatorialsol(\alpha)$ denote the combinatorial solution in~$\mathfrak{K}$ that is optimal for~$\alpha$, and~$\solution(\alpha)$ its realization for~$\alpha$.
Consider a combinatorial arc~$f \in \hat{F}(\polygons)$ with endpoints~$u$ and~$v$ as well as three values~$\alpha_1<\alpha_2<\alpha_3$ such that~$f$ is included in~$\combinatorialsol(\alpha_1)$ and~$\combinatorialsol(\alpha_3)$ but not in~$\combinatorialsol(\alpha_2)$.
Because~$f$ is included in~$\combinatorialsol(\alpha_1)$, its endpoints~$u$ and~$v$ belong to the same combinatorial region.
By~\Cref{prop:alpha-nested:paper}, this is still the case in~$\combinatorialsol(\alpha_2)$ and~$\combinatorialsol(\alpha_3)$.
Furthermore, because~$u$ and~$v$ lie on the boundary of~$\combinatorialsol(\alpha_1)$ and~$\combinatorialsol(\alpha_3)$, they must also lie on the boundary of~$\combinatorialsol(\alpha_2)$, again by~\Cref{prop:alpha-nested:paper}.
The boundary segment between~$u$ and~$v$ in~$\combinatorialsol(\alpha_2)$ is less efficient than the valid arc~$f[\alpha_2]$ (cf.~\Cref{def:efficiency:paper}), but~$\solution(\alpha_2)$ is optimal, so~$f[\alpha_2]$ must be excluded because it intersects a region~$R_I^2$ of~$\solution(\alpha_2)$ properly.
Let~$R_{uv}^2$ denote the region containing~$u$ and~$v$ in~$\solution(\alpha_2)$.
We distinguish between three cases:
\begin{enumerate}
    \item $R_{uv}^2 \neq R_I^2$:
    By~\Cref{lem:arc-movement} and~\Cref{prop:alpha-nested:paper}, both~$f[\alpha_2]$ and~$R_I^2$ are included in~$\solution(\alpha_3)$.
    Because they intersect properly and~$\solution(\alpha_3)$ is valid, they must be part of the same region in~$\solution(\alpha_3)$.
    Hence, $\combinatorialsol(\alpha_3)$ has strictly fewer regions than~$\combinatorialsol(\alpha_2)$.
    \item $R_{uv}^2 = R_I^2$ and the intersected boundary piece of~$R_{uv}^2$ lies between~$u$ and~$v$. Because all arcs bend inwards by P5 of~\Cref{prop:arc_properties:full}, the part of~$R_{uv}^2$ that lies to the right of~$f[\alpha_2]$ must contain at least part of a polygon~$B$.
    Let~$R_B^1$ be the region containing~$B$ in~$\solution(\alpha_1)$ and let~$R_{uv}^1$ be the region containing~$u$ and~$v$ in~$\solution(\alpha_1)$.
    If~$R_B^1 = R_{uv}^1$, then~$f[\alpha_1]$ intersects that region properly and therefore~$\solution(\alpha_1)$ is not valid. If~$R_B^1 \neq R_{uv}^1$, then~$\combinatorialsol(\alpha_1)$ has strictly fewer regions than~$\combinatorialsol(\alpha_2)$.
    \item $R_{uv}^2 = R_I^2$ and the intersected boundary piece of~$R_{uv}^2$ lies between~$v$ and~$u$. Then it follows from~\Cref{prop:alpha-nested:paper} that~$f[\alpha_3]$ intersects the region of~$\solution(\alpha_3)$ that contains~$R_I^2$ properly. Hence, $\solution(\alpha_3)$ is not valid.$\,$
\end{enumerate}
\end{proof}


To solve the parametric problem, we consider the objective values of geometric and combinatorial solutions as functions parameterized in~$\alpha$.
For a geometric solution~$\solution$, the parametric objective function given by~$g(\solution)(\alpha):=g_\alpha(\solution)$ is linear and its slope is the perimeter~$\PP(\solution)$.
For a combinatorial solution~$\combinatorialsol$, we define
\[
g(\combinatorialsol)(\alpha):=\begin{cases}
g_\alpha(\combinatorialsol[\alpha]) & \text{if } \alpha\in\mathcal{I}(\combinatorialsol),\\
\infty & \text{otherwise.}
\end{cases}
\]
It follows from property P2 of~\Cref{prop:arc_properties:full} that~$\combinatorialsol[\alpha]$ has a smaller objective value for~$\freeProblem$ than any other realization~$\combinatorialsol[\alpha']$ with~$\alpha'\in\mathcal{I}(\combinatorialsol)$.
Due to~$\alpha$-nestedness (\Cref{cor:combinatorial-nested}), the area~$\A(\combinatorialsol[\alpha])$ of the realizations is increasing with~$\alpha$, so the perimeter, which is the slope of~$g(\combinatorialsol[\alpha])$, must be decreasing.
Hence, $g(\combinatorialsol)$ is the lower envelope of the (infinitely many) functions~$g(\combinatorialsol[\alpha])$ for~$\alpha\in\mathcal{I}(\combinatorialsol)$ and its slope is monotonically decreasing.
Moreover, $g(\combinatorialsol[\alpha])$ is the tangent of~$g(\combinatorialsol)$ at~$\alpha$.

The objective of~$\biProblem$ can be stated as finding the lower envelope of the functions~$g(\combinatorialsol)$ for all possible combinatorial solutions~$\combinatorialsol$.
It is also the lower envelope of the linear functions~$g(\solution)$ for the infinitely many geometric solutions~$\solution$, and therefore continuous.
For the subdivision-restricted problem variant, the number of geometric solutions is finite and it can be solved using a \emph{dichotomic scheme}~\cite{rottmann2024bicritshapes}.
Given an interval~$[\amin,\amax]$ and optimal solutions~$\sollow$ at~$\amin$ and~$\solup$ at~$\amax$, the dichotomic scheme explores the interval by recursive bisection.
At the crossing point~$\anew$ of~$g(\sollow)$ and~$g(\solup)$, it computes an optimal solution~$\solnew$ for~$\anew$.
If~$\solnew$ is better than~$\sollow$ and~$\solup$ for~$\anew$, the algorithm adds~$\solnew$ to the solution set and recurses on the intervals~$[\amin,\anew]$ and~$[\anew,\amax]$.
Otherwise, it terminates with the solution set~$\{ \sollow, \solup \}$.

Adapting the dichotomic scheme to our scenario poses several challenges.
The objective functions of the combinatorial solutions are not linear and they are not fully continuous because they jump to~$\infty$ outside of the validity interval.
Hence, it is not immediately clear that a function cannot appear on the lower envelope more than once, which is necessary for the correctness of the termination condition.
To show that this property is indeed preserved in our case, we first show that if a combinatorial solution is fully contained in another combinatorial solution for one value of~$\alpha$, then this holds for all other values as well.

\begin{lemma}
\label{lemma:valid_solution_containment}
Let~$\combinatorialsol_1,\combinatorialsol_2$ be two combinatorial solutions and~$(\alpha_1,\alpha_2)= \validSet{\combinatorialsol_1}\cap\validSet{\combinatorialsol_2}$ the interval in which both are valid.
If there is an~$\alpha\in(\alpha_1,\alpha_2)$ for which~$\combinatorialsol_1[\alpha] \subseteq \combinatorialsol_2[\alpha]$, then~$\combinatorialsol_1[\alpha'] \subseteq \combinatorialsol_2[\alpha']$ holds for every~$\alpha'\in(\alpha_1,\alpha_2)$.
\end{lemma}
\begin{proof}
We describe points on the boundary of a solution combinatorially.
Consider a combinatorial arc~$f(a,b)$.
For a parameter~$\beta\in[0,1]$, we define the combinatorial boundary point~$p(a,b,\beta)$ as follows:
For~$\alpha\in[0,\infty)$, the realization~$p(a,b,\beta)[\alpha]$ is the point~$p$ on~$f(a,b)[\alpha]$ such that the circular arc from~$u(a,b)[\alpha]$ to~$p$ has length~$\beta \cdot |f(a,b)[\alpha]|$.
Note that~$u(a,b)=p(a,b,0)$ and~$v(a,b)=p(a,b,1)$.

Assume that there are values~$\alpha,\alpha'\in(\alpha_1,\alpha_2)$ such that~$\combinatorialsol_1[\alpha] \subseteq \combinatorialsol_2[\alpha]$ but~$\combinatorialsol_1[\alpha'] \setminus \combinatorialsol_2[\alpha'] \neq \emptyset$.
Let~$\alpha^*$ denote the largest value in $[\alpha,\alpha')$ such that~$\combinatorialsol_1[\alpha^*] \subseteq \combinatorialsol_2[\alpha^*]$.
Because the movement of the solution boundaries is continuous in~$\alpha$, there must be a combinatorial boundary point~$u_2$ of~$\combinatorialsol_2$ that lies on the boundary of~$\combinatorialsol_1[\alpha^*]$ but in the interior of~$\combinatorialsol_1[\alpha^*+\varepsilon]$ and the exterior of~$\combinatorialsol_1[\alpha^*-\varepsilon]$ for any~$\varepsilon>0$.
Let~$u_1$ denote the combinatorial boundary point of~$\combinatorialsol_1$ for which~$u_1[\alpha^*]=u_2[\alpha^*]$.
Note that~$u_1 \neq u_2$ because the two points do not coincide for~$\alpha^*-\varepsilon$ or~$\alpha^*+\varepsilon$.

If~$u_1$ is an interior point of a combinatorial arc~$f_1$ of~$\combinatorialsol_1$, then~$u_2[\alpha^*]$ cannot lie on an input polygon boundary because this would imply that~$\combinatorialsol_1[\alpha^*]$ is invalid.
Hence, $u_2$ must be an interior point of a combinatorial arc~$f_2$ of~$\combinatorialsol_2$.
It follows from~$\combinatorialsol_1[\alpha^*]\subseteq\combinatorialsol_2[\alpha^*]$ that~$f_1[\alpha^*]$ is fully enclosed by~$\combinatorialsol_2[\alpha^*]$.
Because~$f_1[\alpha^*]$ and~$f_2[\alpha^*]$ have the same radius and share at least one point, they must be part of the same circle.
Because~$f_1=f_2$ would contradict $u_1 \neq u_2$, there must be an endpoint of~$f_1[\alpha^*]$ that lies in the interior of~$f_2[\alpha^*]$ (or vice versa).
However, this implies that one of the arcs intersects a polygon properly and therefore the corresponding solution is invalid.

Assume therefore that~$u_1$ lies on an input polygon edge~$e$.
If~$u_2$ is an interior point of a combinatorial arc~$f_2$ of~$\combinatorialsol_2$, then~$f_2[\alpha^*]$ is invalid.
Hence, $u_2$ must also lie on~$e$.
Because~$u_1 \neq u_2$, at least one of them must be an endpoint of a circular arc.
If~$u_2$ is not an endpoint of a circular arc but~$u_1$ is an endpoint of a circular arc~$f_1$, then part of the enclosed region bounded by~$f_1[\alpha^*]$ is in~$\combinatorialsol_1[\alpha^*] \setminus \combinatorialsol_2[\alpha^*]$, which contradicts our choice of~$\alpha^*$.
Assume therefore that~$u_2$ is an endpoint of a circular arc~$f_2$.
If~$u_1$ is not an endpoint of a circular arc, then~$u_2[\alpha^*+\varepsilon]$ lies on the boundary of~$\combinatorialsol_1[\alpha^*+\varepsilon]$, which contradicts our choice of~$\alpha^*$.
Assume therefore that~$u_1$ is an endpoint of a circular arc~$f_1$.
Because~$u_1$ and~$u_2$ are both tangential on~$e$, the arcs~$f_1$ and~$f_2$ must be part of the same circle.
As shown above, this implies a contradiction.
\end{proof}

Using the~$\alpha$-nestedness of optimal solutions, this containment property can be chained across multiple consecutive pieces of the lower envelope, which implies that each combinatorial solution appears at most once.

\lowerEnvelope*
\begin{proof}
Let~$\combinatorialsol$ be a combinatorial solution with validity interval~$\mathcal{I}(\combinatorialsol)=(\amin,\amax)$.
To simplify this proof, we define realizations of~$\combinatorialsol$ at~$\amin$ and~$\amax$, such that the validity interval becomes closed.
Ordinarily, these realizations are not defined because there is at least one realized arc that touches a polygon edge or has distance~$2\alpha$ between its endpoints and radius~$\pi$.
In the former case, the realized arc is split into two free pieces, so we associate the realization with a different combinatorial solution.
In the latter case, the realization does not satisfy~\Cref{prop:arc_properties:full}.
However, it follows from the proof of~\Cref{prop:solution-validity-range_paper:paper} that the distance does not exceed~$2\alpha$, the radius does not exceed~$\pi$, and the realized arc does not intersect a polygon or a region properly.
Hence, there are valid geometric solutions that we can associate with~$\combinatorialsol[\amin]$ and~$\combinatorialsol[\amax]$.

We use the following observation:
Let~$\combinatorialsol_1,\combinatorialsol_2$ be two combinatorial solutions, $[\alpha_1,\alpha_2]$ an interval in which~$\combinatorialsol_1$ is optimal, and~$(\alpha_2,\alpha_3]$ an interval in which~$\combinatorialsol_1$ is not optimal but~$\combinatorialsol_2$ is.
For every~$0 < \varepsilon < \alpha_3-\alpha_2$, we have~$\combinatorialsol_1[\alpha_2] \subseteq \combinatorialsol_2[\alpha_2+\varepsilon]$ due to the~$\alpha$-nestedness of optimal solutions (\Cref{lem:boundary-nestedness}). Because the change in the realized solutions is continuous in~$\alpha$ by~\Cref{lem:arc-movement}, it follows that~$\combinatorialsol_1[\alpha_2] \subseteq \combinatorialsol_2[\alpha_2]$.

Now, assume there are two values~$\amin < \amax$ such that~$\combinatorialsol[\amin]$ and~$\combinatorialsol[\amax]$ are optimal but~$\combinatorialsol[\alpha]$ is not optimal for~$(\amin,\amax)$.
Then there is a sequence~$\langle\combinatorialsol=\combinatorialsol_0, \combinatorialsol_1,\dots,\combinatorialsol_{k-1},\combinatorialsol_k=\combinatorialsol\rangle$ of combinatorial solutions and a sequence~$\amin=\alpha_0<\alpha_1<\dots<\alpha_k<\alpha_{k+1}=\amax$ such that~$\combinatorialsol_{i-1}$, $\combinatorialsol_i$ and the intervals~$[\alpha_{i-1},\alpha_i]$ and~$(\alpha_i,\alpha_{i+1}]$ fulfill the preconditions of our observation for every~$1 \leq i \leq k$.
Then we have~$\combinatorialsol_{i-1}[\alpha_i] \subseteq \combinatorialsol_i[\alpha_i]$. Furthermore, because~$\combinatorialsol=\combinatorialsol_0$ is valid for~$\alpha_i$, it follows from~\Cref{lemma:valid_solution_containment} and an inductive argument that~$\combinatorialsol[\alpha_i] \subseteq \combinatorialsol_i[\alpha_i]$.
In particular, we have~$\combinatorialsol[\alpha_{k-1}] \subseteq \combinatorialsol_{k-1}[\alpha_{k-1}]$ and~$\combinatorialsol_{k-1}[\amax] \subseteq \combinatorialsol[\amax]$. Because $\combinatorialsol_{k-1}[\amax]$ is not optimal but $\combinatorialsol[\amax]$ is, the second inclusion is proper. Therefore, this contradicts~\Cref{lemma:valid_solution_containment}.
\end{proof}

For most applications, it is not necessary to compute the complete set of solutions for the parametric problem. Instead, an approximation of the set is sufficient.
\begin{definition} [Approximate Parametric Solution]
Let~$\polygons$ be a set of interior-disjoint polygons in the plane and $\varepsilon>0$. A $(1+\varepsilon)$-approximate solution for the parametric unrestricted polygon aggregation problem $\biProblem$ is a set $\mathfrak{K}$ of combinatorial solutions such that for every $\alpha\in[0,\infty]$, there exists an $\alpha'$ and a combinatorial solution $\combinatorialsol\in\mathfrak{K}$ such that $g_\alpha(\combinatorialsol[\alpha']) \leq (1+\varepsilon) \cdot \text{OPT}(\alpha)$, where~$\text{OPT}(\alpha)$ is an optimal solution to~$\freeProblem$.
\end{definition}

Now that the necessary tools have been established, we present three variants of a dichotomic scheme for the parametric problem.
We begin with an exact algorithm that runs in polynomial time, provided there exists an oracle that, given two combinatorial solutions $\combinatorialsol_1$ and $\combinatorialsol_2$, computes the intersection point of $g(\combinatorialsol_1)$ and $g(\combinatorialsol_2)$.
It is unclear whether such an oracle exists, so we propose a second algorithm that approximates the intersection point via binary bisection.
We show that this yields an FPTAS for the parametric problem: for any $\varepsilon > 0$, it computes a $(1+\varepsilon)$-approximate solution in time polynomial in the input size and~$1/\varepsilon$.
Finally, we present a third algorithm that uses a Newton-like method to approximate the intersection point.
Although this does not provide a formal runtime guarantee, we show in~\Cref{sec:appendix:chord_scheme_comparison} that it performs well in practice.

\subsection{Exact Dichotomic Scheme}\label{sec:dichotomic:exact}
\begin{algorithm2e}
    \caption{Exact dichotomic scheme for~$\biProblem$ in the interval~$[\amin,\amax]$.}
    \label{alg:chord:exact}
    $\comblow \gets$ Combinatorial solution optimal at $\amin$\;
    $(\alowinf,\alowsup) \gets$ Validity interval $\mathcal{I}(\comblow)$\;
    Report~$\comblow$\;
    \BlankLine
    $\combup \gets$ Combinatorial solution optimal at $\amax$\;
    $(\aupinf,\aupsup) \gets$ Validity interval $\mathcal{I}(\combup)$\;
    $\Recurse(\comblow,[\amin,\alowsup) ,\combup, (\aupinf, \amax])$\;
    Report~$\combup$\;
    \BlankLine
    \myproc{$\Recurse(\comblow,[\amin,\alowsup) ,\combup, (\aupinf, \amax])$}{
        \If{$\aupinf >\amin$}{
            $\combnew \gets$ Combinatorial solution optimal at $\aupinf$\;
            $(\anewinf,\anewsup) \gets$ Validity interval $\mathcal{I}(\combnew)$\;
            \If{$g(\combnew)(\aupinf) \neq g(\comblow)(\aupinf)$}{
                $\Recurse(\comblow,[\amin,\alowsup) ,\combnew, (\anewinf, \aupinf])$\;
                Report~$\combnew$\;
                $\Recurse(\combnew,[\aupinf,\anewsup) ,\combup, (\aupinf, \amax])$\;
            }
            \Else{
                $\Recurse(\comblow,[\aupinf,\alowsup) ,\combup, (\aupinf, \amax])$\;
            }
        }
        \ElseIf{$\alowsup <\amax$}{
            $\combnew \gets$ Combinatorial solution optimal at $\alowsup$\;
            $(\anewinf,\anewsup) \gets$ Validity interval~$\mathcal{I}(\combnew)$\;
            \If{$g(\combnew)(\alowsup) \neq g(\combup)(\alowsup)$}{
                $\Recurse(\comblow,[\amin,\alowsup) ,\combnew, (\anewinf, \alowsup])$\;
                Report~$\combnew$\;
                $\Recurse(\combnew,[\alowsup,\anewsup) ,\combup, (\aupinf, \amax])$\;
            }
            \Else{
                $\Recurse(\comblow,[\amin,\alowsup) ,\combup, (\aupinf, \alowsup])$\;
            }
        }
        \Else{
        \lIf{$g(\comblow)$ \text{and} $g(\combup)$ \text{do not cross in} $(\amin,\amax)$}{\Return\label{alg:low-high-equal}}
        $\anew\gets \text{Crossing point of } g(\comblow) \text{ and } g(\combup)$\;
        $\combnew \gets$ Combinatorial solution optimal at $\anew$\;
        $(\anewinf,\anewsup) \gets$ Validity interval~$\mathcal{I}(\combnew)$\;
            \If{$g(\combnew)(\anew)\neq g(\comblow)(\anew)$}{
                $\Recurse(\comblow,[\amin,\alowsup) ,\combnew, (\anewinf, \anew])$\;
                Report~$\combnew$\;
                $\Recurse(\combnew,[\anew,\anewsup) ,\combup, (\aupinf, \amax])$\;
            }
        }   
    }
\end{algorithm2e}
The exact variant of the dichotomic scheme for~$\freeProblem$, shown in~\Cref{alg:chord:exact}, requires an oracle that computes the crossing point of~$g(\combinatorialsol_1)$ and~$g(\combinatorialsol_2)$ for two combinatorial solutions~$\combinatorialsol_1,\combinatorialsol_2$.
Although these functions are elementary, they involve non-trivial interactions between polynomials and trigonometric terms, so it is unclear whether their crossing points can be determined analytically.
However, they can be approximated with numerical methods (e.g., with Newton's method).

The main challenge that needs to be solved is that the objective functions are not valid for all values of~$\alpha$.
Whenever the algorithm computes an optimal geometric solution for some value~$\alpha$ using~\optimizer, it retrieves the associated combinatorial solution and computes its validity interval.
The recursive step for the interval~$[\amin,\amax]$ takes as inputs a combinatorial solution~$\comblow$ that is optimal for~$\amin$ and valid below~$\alowsup$, and a combinatorial solution~$\combup$ that is optimal for~$\amax$ and valid above~$\aupinf$.
If~$\comblow$ and~$\combup$ are both valid in~$(\amin,\amax)$, the algorithm proceeds normally by computing the crossing point~$\anew$ of the objective functions, computing an optimal combinatorial solution~$\combnew$ at~$\anew$, and recursing if~$\combnew$ improves upon~$\comblow$ and~$\combup$, which are equivalent at~$\anew$.
If~$\aupinf > \amin$, the algorithm computes~$\combnew$ at~$\aupinf$ instead.
If~$\combnew$ is better than~$\comblow$ at~$\anew$, the algorithm recurses in both~$[\amin,\aupinf]$ and~$[\aupinf,\amax]$.
If it is equivalent to~$\comblow$, it is not necessary to recurse in~$[\amin,\aupinf]$.
The algorithm still recurses in~$[\aupinf,\amin]$, but it has now cut off the part of the interval in which~$\combup$ is invalid, which ensures that it makes progress.
The case for~$\alowsup < \amax$ is symmetrical.

\begin{theorem}
    When given access to an oracle that computes the crossing point between two functions, \Cref{alg:chord:exact} reports a  solution for~$\biProblem$ in the interval~$[\amin,\amax]$ in~$\bigO{n^6 \frac{\log^3n}{\log^2 \log n}}$ time.
\end{theorem}
\begin{proof}
The algorithm upholds the following invariants for each~$\Recurse$ call: The combinatorial solutions~$\comblow$ and~$\combup$ are valid and optimal for~$\amin$ and~$\amax$, respectively.
The supremum of the validity interval of~$\comblow$ is given by~$\alowsup$ and the infimum of the validity interval of~$\combup$ by~$\aupinf$.
We show by induction over the recursion tree that the set of combinatorial solutions reported by~$\Recurse$, plus~$\comblow$ and~$\combup$, is a solution for~$\biProblem$ the interval~$(\amin,\amax)$.
The algorithm distinguishes between three cases.
Unless additional conditions are met, each case computes a combinatorial solution~$\combnew$ that is optimal at some value~$\anew\in(\amin,\amax)$, reports it and recurses on the interval~$[\amin,\anew]$ and~$[\anew,\amax]$ while upholding the invariants.
By the induction step, this yields a solution for~$(\amin,\amax)$.

Regarding the additional conditions, we first consider the case~$\aupinf > \amin$. If~$\combnew$ and~$\comblow$ have the same value at~$\aupinf$, then~$\comblow$ is optimal at both~$\amin$ and~$\aupinf$. Then it follows from~\Cref{lem:lower-envelope:paper} that it is optimal within the entire interval~$[\amin,\aupinf]$.
Hence, it is sufficient to recurse on~$[\aupinf,\amax]$, using~$\comblow$ as the left input solution.
The case~$\alowsup < \amax$ is symmetrical.
In the third case, we know that~$\comblow$ and~$\combup$ are both valid in~$(\amin,\amax)$.
If~$g(\comblow)$ and~$g(\combup)$ do not cross in~$(\amin,\amax)$, then one of the two solutions is optimal at both~$\amin$ and~$\amax$.
If~$\comblow$ and~$\combnew$ have the same value at~$\anew$, then all three solutions are optimal at~$\anew$.
In both cases, $\{ \comblow, \combup\}$ is already a solution for~$\biProblem$ in~$(\amin,\amax)$ by~\Cref{lem:lower-envelope:paper}.

The runtime for a single~$\Recurse$ call (excluding the recursive calls) is dominated by the time for finding the optimal solution~$\combnew$, which is in~$\tilde{\mathcal O}(n^4)$ by~\Cref{theorem:algorithm:paper}.
The other potentially expensive step is computing the validity interval of~$\combnew$.
Here, the most expensive step is testing for pairwise intersection between combinatorial boundary pieces, which requires~$\mathcal{O}(n^2)$ time.

To obtain the overall runtime bound, we show that the number of~$\Recurse$ calls is proportional to the number of reported solutions, which is in~$\mathcal O(n^2)$ by~\Cref{prop:parametric-solutions:paper}. Whenever the algorithm makes two recursive calls, it reports a new solution. When it makes one recursive call without reporting a new solution, it leaves all inputs unchanged except for~$\amin$ (or~$\amax$, respectively). This ensures that the condition~$\aupinf > \amax$ (or~$\alowsup < \amax$, respectively), cannot be satisfied anymore. Thus, the algorithm enters the else-case after at most two calls, where it either reports a solution or terminates without another recursive call.
\end{proof}
\subsection{An Approximation Scheme}
For our approximation scheme, we replace the crossing point oracle with a binary search and use it to recursively bisect the interval~$[\amin,\amax]$.
Given combinatorial solutions $\combinatorialsol_L$ optimal at $\alpha_\text{min}$ and $\combinatorialsol_U$ optimal at $\alpha_\text{max}$, the bisection point would ordinarily be $\alpha_\text{N} = (\alpha_\text{max} + \alpha_\text{min})/2$.
However, this is not well defined for $\alpha_\text{max}=\infty$.
To circumvent this, we switch to a different but equivalent formulation of the objective function:
\begin{align*}
    f_\lambda(\solution) = (1-\lambda) \cdot \A(\solution) + \lambda \cdot \PP(\solution)
\end{align*}
Since the domain of $f$~is~$[0,1]$, the bisection point~$\lnew= (\lmax - \lmin)/2$, is always well defined.
To see that the formulations are equivalent, note that for~$\alpha=\frac{\lambda}{1-\lambda}$, we have~$(1-\lambda) \cdot g_\alpha(\solution) = (1-\lambda) \cdot \A(\solution) + \lambda \cdot \PP(\solution) = f_\lambda(\solution)$.
Moreover, we note that approximating~$f$ and~$g$ gives the same solution set because~$f_\lambda(\solution_1) \leq (1+\varepsilon) \cdot f_\lambda(\solution_2)$ is equivalent to~$(1-\lambda) \cdot g_\alpha(\solution_1) \leq (1-\lambda) \cdot (1+\varepsilon) \cdot g_\alpha(\solution_2)$.
Note that this formulation is almost identical to the original formulation by Rottmann et al.~\cite{rottmann2024bicritshapes}, except that the role of area and perimeter is switched.
This is done for ease of exposition, as it ensures that~$f_\alpha$ is decreasing in the parameter, just as $g_\alpha$.
Because the parameter~$\lambda$ is strictly increasing in~$\alpha$, it follows that the area of a combinatorial solution is decreasing in~$\lambda$ and the perimeter is increasing.
This also means that~$g_\alpha(\solution_1) < g_\alpha(\solution_2)$ iff $f_\lambda(\solution_1) < f_\lambda(\solution_2)$.
The slope of the function is~$f'(\combinatorialsol)(\lambda) = \PP(\combinatorialsol[\lambda]) - \A(\combinatorialsol[\lambda])$, which is decreasing in~$\lambda$.

Unlike in the exact dichotomic scheme, $\comblow$ and~$\combup$ do not have the same objective value at~$\lnew$ because it is not the actual crossing point.
Hence, our algorithm, which is depicted in~\Cref{alg:chord:heuristic:new}, must perform two independent comparisons and may end up recursing on only one side. Since we are not guaranteed to reach the exact crossing point with binary bisection, we introduce an approximate optimality test: After computing the optimal solution $\combnew$ at~$\lnew$, we compare it to the evaluated geometric solutions~$\comblow[\amin]$ and~$\combup[\amax]$ at~$\anew$.
If one of them is within an~$(1+\varepsilon)$-factor of~$\combnew$, we do not recurse on that side. 

In the remainder of this section, we show that our algorithm computes a $(1+\varepsilon)$-approximate solution in polynomial time, with a logarithmic dependence on $1/\varepsilon$ and a quantity that describes the range of distances that occur in~$\polygons$.
A standard measure for describing the range of distances in a point set $\mathcal{S}$ is the \emph{spread}, defined as $\frac{\max_{u,v \in \mathcal{S}} \LL(\overline{uv})}{\min_{u,v \in \mathcal{S}} \LL(\overline{uv})}.$ However, in algorithmic settings involving polygons, this quantity fails to adequately capture the structural properties of the input. In particular, regardless of the spread, one can construct instances in which the polygons become degenerate, for example by having arbitrarily small area. To address this limitation, we introduce a refined measure that also considers the structure of the polygons.
Given an instance~$\polygons$, the \emph{diameter}~$\text{diam}(\polygons)$ is the maximum distance between any pair of vertices of~$\polygons$.
The \emph{segment separation}~$\delta_\text{seg}(\polygons)$ is the minimum distance between any polygon vertex~$v$ and any polygon edge that does not contain~$v$.
This quantity captures both the shortest distance between vertices and the thickness of the polygons.
The ratio between the two is the \emph{segment spread}~$\Phi:=\frac{\text{diam}(\polygons)}{\delta_\text{seg}(\polygons)}$.
In the following proofs, we use the quantity~$K:=\max\{\frac{\PP(\polygons)}{\A(\polygons)},\frac{\A(\conv(\polygons))}{\PP(\conv(B_\text{min}))}\}$, where $B_\text{min} \in \polygons$ denotes the polygon whose convex hull has the smallest perimeter.

\begin{lemma}
\label{lem:size-of-K}
For a set~$\polygons$ of polygons, the quantity~$K=\max\{\frac{\PP(\polygons)}{\A(\polygons)},\frac{\A(\conv(\polygons))}{\PP(\conv(B_\text{min}))}\}$ satisfies~$\log K\in\bigO{\log n +  \log \Phi}$.
\end{lemma}
\begin{proof}
    The perimeter of~$\polygons$ consists of at most~$n$ line segments of length at most~$\text{diam}(\polygons)$, i.e., we have~$\PP(\polygons)\leq n \cdot \text{diam}(\polygons)$.
    On the other hand, the perimeter of every polygon in~$\polygons$ includes at least one line segment of length at least~$\delta_\text{seg}(\polygons)$, so we have~$\PP(\conv(B_\text{min}))\geq \delta_\text{seg}(\polygons)$.
    We have~$\A(\conv(\polygons))\leq \text{diam}(\polygons)^2$, as the convex hull of~$\polygons$ is contained within a bounding box of side length~$\text{diam}(\polygons)$.
    Let~$B$ be any polygon of~$\polygons$ and let~$T$ be any triangle in an arbitrary triangulation of~$B$.
    Let $g$ be the longest side of $T$, $p$ the opposite vertex, and the line segment given by the orthogonal projection of~$p$ onto the line through~$g$.
    Because~$g$ is the longest side, $h$ touches~$g$.
    Since~$p$ is not incident to~$g$, it follows that~$\LL(h) \geq \delta_{\mathrm{seg}}$.
    Similarly, we have~$\LL(g) \geq \delta_{\mathrm{seg}}$.
    Hence, we have
\begin{equation*}
    A(\polygons) \geq A(B)\geq A(T) = \tfrac{1}{2}\LL(g)\cdot \LL(h) \geq \tfrac{1}{2}\delta_{\mathrm{seg}}(\polygons)^2.
\end{equation*}
     It follows that $K \leq \max \{ \frac{2n \cdot \text{diam}(\polygons)}{\delta_\text{seg}(\polygons)^2}, \frac{\text{diam}(\polygons)^2}{\delta_\text{seg}(\polygons)} \}$ and therefore $\log K \in \bigO{\log n  + \log \Phi}$.
\end{proof}

We show that computing the (unique) intersection point of~$f(\combinatorialsol_1)$ and~$f(\combinatorialsol_2)$ for two combinatorial solutions~$\combinatorialsol_1$ and~$\combinatorialsol_2$ within an interval of length~$\Delta$ is equivalent to finding the root of a $4K$-Lipschitz function within this interval, and that binary search $\varepsilon$-approximates this root in~$\bigO{\log \frac{K}{\varepsilon} + \log \Delta}$ time.
Hence, although our binary bisection scheme does not necessarily find the crossing point exactly, it comes sufficiently close in a logarithmic number of steps that the approximate optimality test succeeds.

\begin{algorithm2e}
    \caption{Approximative dichotomic scheme for~$\biProblem$.}\label{alg:chord:heuristic:new}
    $\comblow \gets$ Combinatorial solution optimal for~$\lambda = 0$\;
    Report~$\comblow$\;
    \BlankLine
    $\combup \gets$ Combinatorial solution optimal for~$\lambda = 1$\;
    $\Recurse(\comblow, 0, \combup, 1)$\;
    Report~$\combup$\;
    \BlankLine
    \myproc{$\Recurse(\comblow, \lmin, \combup, \lmax)$}{
        \lIf{$\comblow =\combup$ \KwOr $f(\comblow[\lmin]) = f(\combup[\lmax])$\label{alg:chord:heuristic:new:identical}}{\Return}
    $\lnew \gets (\lmax - \lmin)/2$\;
    $\combnew \gets$ Comb. solution optimal for $\lnew$\label{alg:chord:heuristic:new:optimize}\;
    \tcp{Note $f(\combinatorialsol)(\lambda)=\infty$ if not valid}
    $\text{dom}_L\gets f(\combnew)(\lnew) < \min(f(\comblow)(\lnew), f(\comblow[\lmin])(\lnew)/(1+\varepsilon))$\;
    $\text{dom}_U\gets f(\combnew)(\lnew) < \min(f(\combup)(\lnew), f(\combup[\lmax])(\lnew)/(1+\varepsilon))$\;
    \lIf{$\text{dom}_L$}{
        $\Recurse(\comblow, \lmin, \combnew, \lnew)$
    }
    \lIf{$\text{dom}_L$ $\KwAnd$ $\text{dom}_U$}{Report~$\combnew$}
     \lIf{$\text{dom}_U$}{
        $\Recurse(\combnew, \lnew, \combup, \lmax)$
    }
    }
\end{algorithm2e}

\begin{lemma}
    \label{lem:lipschitz}
    Let~$\solution$ be a geometric solution that is optimal for some~$\lambda^* \in [0,1]$.
    Then the function~$h(\solution)$ with~$h(\solution)(\lambda) := \log f(\solution)(\lambda)$ is~$K$-Lipschitz.
\end{lemma}
\begin{proof}
    For any~$\lambda \in [0,1]$, we show that~$|h'(\solution)(\lambda)| \leq K$, where~$h'$ denotes the first derivative of~$h$ in~$\lambda$.
    We have
    \[
    h'(\solution)(\lambda) = \frac{f'(\solution)(\lambda)}{f(\solution)(\lambda)} = \frac{\PP(\solution) - \A(\solution)}{f(\solution)(\lambda)}.
    \]
    If~$\PP(\solution) < \A(\solution)$, then we have
    \[
    |h'(\solution)(\lambda)| = \frac{\A(\solution) - \PP(\solution)}{f(\solution)(\lambda)} \leq \frac{\A(\solution)}{\PP(\solution)} \leq \frac{\A(\conv(\polygons))}{\PP(\conv(B_\text{min}))} \leq K.
    \]
    Because~$\solution$ is optimal for some~$\lambda^*$ and~$\A(\solution) \geq \A(\polygons)$, we have~$\PP(\solution) \leq \PP(\polygons)$.
    Hence, if~$\A(\solution) < \PP(\solution)$, we have
    \[
    |h'(\solution)(\lambda)| = \frac{\PP(\solution)-\A(\solution)}{f(\solution)(\lambda)} \leq \frac{\PP(\solution)}{\A(\solution)} \leq \frac{\PP(\polygons)}{\A(\polygons)} \leq K.
    \]
\end{proof}

\begin{lemma}
    \label{lem:perimeter-bound}
    Let~$\combinatorialsol$ be a combinatorial solution and~$\alpha^* \in \validSet{\combinatorialsol}$ such that~$\combinatorialsol[\alpha^*]$ is optimal.
    Then for every~$\alpha \in \validSet{\combinatorialsol}$, we have~$\PP(\combinatorialsol[\alpha]) \leq 3 \cdot \PP(\polygons)$.
\end{lemma}
\begin{proof}
    We have~$\PP(\combinatorialsol[\alpha^*]) \leq \PP(\polygons)$ because~$\combinatorialsol[\alpha^*]$ is optimal.
    Due to~$\alpha$-nestedness (\Cref{cor:combinatorial-nested}), the area of~$\combinatorialsol$ is increasing with~$\alpha$, so the perimeter is decreasing.
    Thus, the remaining case is~$\alpha < \alpha^*$.
    Let~$C_{vv}(\combinatorialsol)$ denote the set of vertex-vertex arcs in~$\combinatorialsol$, $C_{ve}(\combinatorialsol)$ the set of vertex-edge arcs, $C_{ee}(\combinatorialsol)$ the set of edge-edge arcs, and~$E_\polygons(\combinatorialsol[\alpha])$ the set of input polygon edge segments that appear in~$\combinatorialsol[\alpha]$.
    Then we have
    \[
    \PP(\combinatorialsol[\alpha]) = \sum_{c \in C_{vv}(\combinatorialsol)} \ell(c[\alpha]) + \sum_{c \in C_{ve}(\combinatorialsol)} \ell(c[\alpha]) + \sum_{c \in C_{ee}(\combinatorialsol)} \ell(c[\alpha]) + \sum_{e \in E_\polygons(\combinatorialsol[\alpha])} \ell(e).
    \]
    It is clear that~$\sum_{e \in E_\polygons(\combinatorialsol[\alpha])} \ell(e) \leq \PP(\polygons)$.
    By~\Cref{lem:arc-movement}, the vertex-edge and edge-edge arcs become shorter as~$\alpha$ decreases.
    For~$c \in C_{vv}(\combinatorialsol)$ with endpoints~$u,v$, we have~$d(u,v) \leq \ell(c[\alpha]) \leq \pi/2 \cdot d(u,v)$.
    This is because~$c[\alpha]$ is a straight line for~$\alpha=\infty$ and a half-circle for the smallest~$\alpha$ such that~$c$ is valid.
    It follows that~$\ell(c[\alpha]) \leq 2 \cdot \ell(c[\alpha^*])$.
    Overall, we have~$\PP(\combinatorialsol[\alpha]) \leq 2 \cdot \PP(\combinatorialsol[\alpha^*]) + \PP(\polygons) \leq 3 \cdot \PP(\polygons)$.
\end{proof}

It is clear that the area of a combinatorial solution is always within $[\A(\polygons),\A(\conv(\polygons))]$ because the arcs bend inwards and do not touch polygons except at endpoints.
It also clear that the perimeter of any solution is at least $\PP(\conv(B_\text{min}))$.

\begin{lemma}
    Let~$\combinatorialsol$ be a combinatorial solution and~$\lambda^* \in [0,1]$ such that~$\combinatorialsol[\lambda^*]$ is optimal.
    Then the function~$h(\combinatorialsol)$ with~$h(\combinatorialsol)(\lambda) := \log f(\combinatorialsol)(\lambda)$ is~$3K$-Lipschitz.
\end{lemma}
\begin{proof}
    For any~$\lambda \in [0,1]$, we show that~$|h'(\combinatorialsol)(\lambda)| \leq 3K$, where~$h'$ denotes the first derivative of~$h$ in~$\lambda$.
    We have
    \[h'(\combinatorialsol)(\lambda) = \frac{f'(\combinatorialsol)(\lambda)}{f(\combinatorialsol)(\lambda)} = \frac{\PP(\combinatorialsol[\lambda]) - \A(\combinatorialsol[\lambda])}{f(\combinatorialsol)(\lambda)}.\]

    If~$\PP(\combinatorialsol[\lambda]) < \A(\combinatorialsol[\lambda])$, then we have
    \[
    |h'(\combinatorialsol)(\lambda)| = \frac{\A(\combinatorialsol[\lambda]) - \PP(\combinatorialsol[\lambda])}{f(\combinatorialsol)(\lambda)} \leq \frac{\A(\combinatorialsol[\lambda])}{\PP(\combinatorialsol[\lambda])} \leq \frac{\A(\conv(\polygons))}{\PP(\conv(B_\text{min}))} \leq K.
    \]
    If~$\A(\combinatorialsol[\lambda]) < \PP(\combinatorialsol[\lambda])$, then using~\Cref{lem:perimeter-bound} and the fact that switching between the objectives $g_\alpha$ and $f_\lambda$ does not change the perimeter of the solution, we have
    \[
    |h'(\combinatorialsol)(\lambda)| = \frac{\PP(\combinatorialsol[\lambda])-\A(\combinatorialsol[\lambda])}{f(\combinatorialsol)(\lambda)} \leq \frac{\PP(\combinatorialsol[\lambda])}{\A(\combinatorialsol[\lambda])} \leq \frac{3\PP(\polygons)}{\A(\polygons)} \leq 3K.
    \]
\end{proof}

\begin{corollary}
    \label{cor:lipschitz-diff}
     Let~$\combinatorialsol$ be a combinatorial solution and~$\lambda^* \in [0,1]$ such that~$\combinatorialsol[\lambda^*]$ is optimal.
     Furthermore, let~$\solution$ be a geometric solution that is optimal for some~$\lambda'$.
     Then the function~$h(\solution) - h(\combinatorialsol)$ is~$4K$-Lipschitz.
\end{corollary}

\begin{lemma}
    \label{lem:lipschitz-binary}
    Let~$f$ be a monotone~$c$-Lipschitz function that has exactly one root in the interval~$[L,U]$. Then binary search finds a value~$x \in [L,U]$ with~$|f(x)| \leq \varepsilon$ in~$\mathcal O(\log\frac{c}{\varepsilon} + \log \Delta)$ time, where~$\Delta = U - L$.
\end{lemma}
\begin{proof}
    Let~$x^*$ denote the root of~$f$ in~$[L,U]$.
    Then we have~$|f(x)| = |f(x) - f(x^*)| \leq c |x - x^*|$.
    After~$k$ steps of the binary search, the considered interval has length~$\frac{\Delta}{2^k}$, so we have
    \[ |f(x_k)| \leq c  |x_k - x^*| \leq c \frac{\Delta}{2^k}. \]
    For~$k = \lceil\log_2\frac{c}{\varepsilon} + \log_2 \Delta\rceil$, we have~$|f(x_k)| \leq \varepsilon$.
\end{proof}

\begin{theorem}
    \Cref{alg:chord:heuristic:new} computes a~$(1+\varepsilon)$-approximate solution to~$\biProblem$ in~$\mathcal O(n^6 \frac{\log^3 n}{\log^2 \log n} (\log n + \log \Phi \log \frac{1}{\varepsilon}))$ time.
\end{theorem}
\begin{proof}
    We first show that the algorithm terminates in the given time.
    To simplify the analysis, we consider a modified version of the algorithm that calls~$\Recurse(\combnew,\lnew,\combnew,\lnew)$ if~$\text{dom}_L$ and~$\text{dom}_U$  are both false and~$f(\combnew)(\lnew) = f(\comblow)(\lnew) = f(\combup)(\lnew)$ holds. 
    Because the call terminates immediately in line~\ref{alg:chord:heuristic:new:identical}, this variant has the same output and the same asymptotic runtime as~\Cref{alg:chord:heuristic:new}.
    
    If~$f(\combnew)[\lnew]$ is strictly smaller than both~$f(\comblow)(\lnew)$ and~$f(\combup)(\lnew)$, then the algorithm has discovered a new part of the lower envelope, even if it does not report it.
    This can happen at most~$\mathcal O(n^2)$ times by~\Cref{prop:parametric-solutions:paper}.
    If~$f(\combnew)[\lnew]$ is equal to at least one of~$f(\comblow)(\lnew)$ and~$f(\combup)(\lnew)$, the algorithm makes at most a single~$\Recurse$ call.
    We show that after~$\mathcal O(\log \frac{K}{\varepsilon})$ recursive calls in which no new part of the lower envelope is discovered, the algorithm terminates because both~$\text{dom}_L$ and~$\text{dom}_U$ are false.
    Then it follows that the total number of $\Recurse$ calls is in $O(n^2 \log \frac{K}{\varepsilon})$.
    The time for each call is dominated by line~\ref{alg:chord:heuristic:new:optimize}, which takes~$\mathcal O(n^4 \frac{\log^3 n}{\log^2 \log n})$ time, so the runtime claim follows with~\Cref{lem:size-of-K}.
    
    We show that~$(1+\varepsilon) f(\combnew)(\lnew) \geq f(\comblow[\lmin])(\lnew)$ holds after~$\mathcal O(\log \frac{K}{\varepsilon})$ recursive calls; the argument for~$\text{dom}_U$ is analogous.
    If the algorithm did not find a new part of the lower envelope in the current call, then we have~$f(\combnew)(\lnew) = \min\{f(\comblow)(\lnew),f(\combup)(\lnew)\}$.
    Thus, the condition to be met is
    \[
    \frac{f(\comblow[\lmin])(\lnew)}{\min\{f(\comblow)(\lnew),f(\combup)(\lnew)\}} \leq 1 + \varepsilon,
    \]
    which is equivalent to
    \[
    \max\{\frac{f(\comblow[\lmin])(\lnew)}{f(\comblow)(\lnew)},\frac{f(\comblow[\lmin])(\lnew)}{f(\combup)(\lnew)}\} \leq 1 + \varepsilon.
    \]
    Once again, we restrict ourselves to the case
    \[
    \frac{f(\comblow[\lmin])(\lnew)}{f(\combup)(\lnew)} \leq 1 + \varepsilon,
    \]
    as the other one is analogous.
    With~$h(\lambda) := h(\comblow[\lmin])(\lambda) - h(\combup)(\lambda) \leq \log(1+\varepsilon)$, this is equivalent to finding a value~$\lambda^*\in[\lmin,\lmax]$ such that~$h(\lambda^*) \leq \log(1+\varepsilon)$.
    By~\Cref{cor:lipschitz-diff}, $h$ is~$4K$-Lipschitz.
    For~$0 < \varepsilon \leq 1$, we have~$\log(1+\varepsilon) \geq \varepsilon/2$ (note that~$\log$ denotes the natural logarithm).
    By~\Cref{lem:lipschitz-binary}, binary search finds a value~$\lambda^* \in [\lmin,\lmax]$ with~$h(\lambda^*) \leq |h(\lambda^*)| \leq \varepsilon/2 \leq \log(1+\varepsilon)$ in~$\mathcal O(\log\frac{K}{\varepsilon} + \log \Delta)$ time, where~$\Delta = \lmax - \lmin \leq 1$.

    To show that the algorithm computes a $(1+\varepsilon)$-approximation, consider the case that the recursion in~$[\lmin,\lnew]$ is skipped because~$\text{dom}_L$ is false (the case for the other recursion is symmetrical).
    We show that for every~$\lambda\in[\lmin,\lnew]$, there is a~$\lambda'\in[0,1]$ such that~$f(\comblow[\lambda'])(\lambda) \leq (1+\varepsilon) \cdot f(\combnew)(\lambda)$.
    Consider the case that~$\text{dom}_L$ is false because~$f(\combnew) \geq f(\comblow)(\lnew)$ holds.
    Because~$\combnew$ is optimal for~$\lnew$, it follows that~$\comblow$ is optimal as well. Then~$\comblow$ is optimal within~$[\lmin,\lnew]$ by~\Cref{lem:lower-envelope:paper} and the claim follows with~$\lambda'=\lambda$.
    
    Now, consider the other case that~$(1+\varepsilon) \cdot f(\combnew)(\lnew) \geq f(\comblow[\lmin])(\lnew)$.
    We write this as~$z(\lnew) \geq 0$ with the functions~$f_1:=(1+\varepsilon) \cdot f(\combnew)$, $f_2:=f(\comblow[\lmin])$ and~$z:=f_1-f_2$.
    Because~$\comblow$ is optimal at~$\lmin$, we also have~$h(\lmin) \geq 0$.
    We show that~$z(\lambda) \geq 0$ for~$\lambda\in[\lmin,\lnew]$.
    The slope of~$f_1$ is decreasing, whereas the slope of~$f_2$ is constant, so the slope of~$z$ is decreasing.
    If~$z(\lambda) < 0$, then the slope of~$z$ must be negative at~$\lambda$.
    However, this implies that it remains negative, which contradicts~$z(\lnew) \geq 0$.
    Hence, the claim follows for~$\lambda'=\lmin$.
\end{proof}

\subsection{Heuristic Dichotomic Scheme}\label{sec:dichotomic:heuristic}

The binary bisection scheme of \cref{alg:chord:heuristic:new} searches for the intersection point in an uninformed manner, while the exact scheme of \cref{alg:chord:exact} requires an oracle to determine it. In contrast, we present a heuristic dichotomic scheme (given in \cref{alg:chord:heuristic}) that aims to estimate the intersection point in a more informed way.
To this end, it approximates the functions~$g(\comblow)$ and~$g(\combup)$ by their tangents~$g(\comblow[\amin])$ and~$g(\combup[\amax])$, corresponding to the objective functions of the geometric solutions. The intersection point~$\anew$ of these tangents is then used as an estimate for the true crossing point~$\alpha^*$ of the combinatorial solutions (see~\Cref{fig:chord}).
Otherwise, the algorithm is identical to \cref{alg:chord:heuristic:new}.

In the remainder of this section, we show that \cref{alg:chord:heuristic} terminates and return an $(1+\varepsilon)$-approximate solution.

\begin{figure}
\centering
\includegraphics{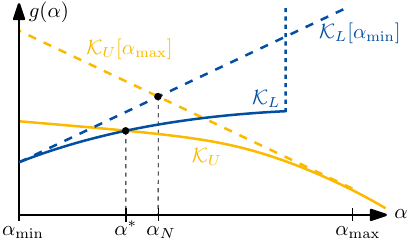}
\caption{A~\texttt{Recurse} call in the dichotomic scheme for~$\biProblem$. The combinatorial solutions~$\comblow$ and~$\combup$ are approximated via their tangents~$\comblow[\amin]$ and~$\combup[\amax]$. A new optimal solution is computed at the crossing point~$\anew$ of the tangents, which differs from the crossing point~$\alpha^*$ of the combinatorial solutions. For illustration purposes, we depict the slope of~$\combup$ as negative.}
\label{fig:chord}
\end{figure}

\begin{algorithm2e}
    \caption{Heuristic dichotomic scheme for~$\biProblem$ in the interval~$[\amin,\amax]$.}\label{alg:chord:heuristic}
    $\comblow \gets$ Combinatorial solution optimal at $\amin$\;
    Report~$\comblow$\;
    \BlankLine
    $\combup \gets$ Combinatorial solution optimal at $\amax$\;
    $\Recurse(\comblow, \amin, \combup, \amax)$\;
    Report~$\combup$\;
    \BlankLine
    \myproc{$\Recurse(\comblow, \amin, \combup, \amax)$}{
        \lIf{$\comblow =\combup$ \KwOr $g(\comblow[\amin]) = g(\combup[\amax])$\label{alg:chord:heuristic:identical}}{\Return}
    $\anew \gets$ Crossing point of ${g}(\comblow[\amin])$ and ${g}(\combup[\amax])$\;
    $\combnew \gets$ Comb. solution optimal at $\anew$\;
    \tcp{Note $g(\combinatorialsol)(\alpha)=\infty$ if not valid}
    $\text{dom}_L\gets g(\combnew)(\anew) < \min(g(\comblow)(\anew), g(\comblow[\amin])(\anew)/(1+\varepsilon))$\;
    $\text{dom}_U\gets g(\combnew)(\anew) < \min(g(\combup)(\anew), g(\combup[\amax])(\anew)/(1+\varepsilon))$\;
    \lIf{$\text{dom}_L$}{
        $\Recurse(\comblow, \amin, \combnew, \anew)$
    }
    \lIf{$\text{dom}_L$ $\KwAnd$ $\text{dom}_U$}{
    Report~$\combnew$}
    
     \lIf{$\text{dom}_U$}{
        $\Recurse(\combnew, \anew, \combup, \amax)$
    }
    }
\end{algorithm2e}

\begin{theorem}
    For~$\varepsilon>0$, \Cref{alg:chord:heuristic} reports a~$(1+\varepsilon)$-approximated solution for~$\biProblem$ in the interval~$[\amin,\amax]$.
\end{theorem}
\begin{proof}
We first show that the algorithm terminates.
To simplify the analysis, we consider a modified version of the algorithm that calls~$\Recurse(\combnew,\anew,\combnew,\anew)$ if~$\text{dom}_L$ and~$\text{dom}_U$  are both false and~$g(\combnew)(\anew) = g(\comblow)(\anew) = g(\combup)(\anew)$ holds. 
Because the call terminates immediately in line~\ref{alg:chord:heuristic:identical}, this variant has the same output and the same asymptotic runtime as~\Cref{alg:chord:heuristic}.

If~$g(\combnew)[\anew]$ is strictly smaller than both~$g(\comblow)(\anew)$ and~$g(\combup)(\anew)$, then the algorithm has discovered a new part of the lower envelope, even if it does not report it.
This can happen at most~$\mathcal O(n^2)$ times by~\Cref{prop:parametric-solutions:paper}.
If~$g(\combnew)[\anew]$ is equal to at least one of~$g(\comblow)(\anew)$ and~$g(\combup)(\anew)$, the algorithm makes at most a single~$\Recurse$ call.
We show that in a chain of recursive calls, the interval length converges to~$0$.

Because the tangents~$g(\comblow[\amin])$ and~$g(\combup[\amax])$ are optimal at~$\amin$ and~$\amax$, respectively, they are either identical or have a crossing point in~$[\amin,\amax]$.
If they are identical, the algorithm terminates in line~\ref{alg:chord:heuristic:identical}.
If the crossing point is in~$(\amin,\amax)$, then the interval length strictly decreases in the next~$\Recurse$ call.
If the crossing point is~$\amin$ or~$\amax$, then~$\text{dom}_L$ and~$\text{dom}_U$ are both false and the interval has length~$0$ in the next~$\Recurse$ call.
Because the interval length is bounded from below by~$0$, it must converge towards a limit that is a fixed point, i.e., a~$\Recurse$ call in which the interval length does not change.
The only fixed point is for~$\amin=\amax$, so the interval length converges towards~$0$.

We use this to show that after a finite number of steps, the algorithm terminates because both~$\text{dom}_L$ and~$\text{dom}_U$ are false.
Let~$\sollow:=\comblow[\amin]$, $\solnew:=\combnew[\anew]$ and~$\solup:=\combup[\amax]$.
Because~$\anew - \amin$ and~$\amax - \anew$ converge towards~$0$, they decrease below any constant threshold after a finite number of steps.
Hence, we eventually have
\begin{align*}
\anew - \amin \leq &\varepsilon \cdot \frac{\A(\polygons)}{\PP(\polygons)} \leq \varepsilon \cdot \frac{\A(\sollow)}{\PP(\sollow)},\\
\amax - \anew \leq &\frac{\varepsilon \cdot \A(\polygons)}{(1+\varepsilon) \cdot \PP(\polygons)} \leq \frac{\varepsilon \cdot \A(\solup)}{(1+\varepsilon) \cdot \PP(\polygons) - \PP(\solup)}.
\end{align*}
Here we use that~$\sollow$ is optimal for~$\amin$ and~$\solup$ is optimal for~$\amax$ and therefore~$\A(\polygons) \leq \A(\sollow)$, $A(\polygons) \leq \A(\solup)$ and~$\PP(\polygons) \geq \PP(\sollow)$.
Furthermore, the denominator of the final inequality is strictly positive because~$\polygons \subset \solup$.
We rewrite these inequalities as
\begin{align}
\anew \cdot \PP(\sollow) \leq &\ \varepsilon \cdot \A(\sollow) + \amin \cdot \PP(\sollow),\label{eqn:perimeter-lower}\\
\anew \cdot \PP(\solup) \leq &\ \varepsilon \cdot \A(\solup) - (\amax - \anew) \cdot (1+\varepsilon) \cdot \PP(\polygons)\label{eqn:perimeter-upper}\\
&- \amax \cdot \PP(\solup).\notag
\end{align}
We also use the following observation about the function~$\text{OPT}$ that describes the value of the optimal solution dependent on~$\alpha$.
This function is the lower envelope of linear functions whose slope is between~$0$ and~$\PP(\polygons)$.
Hence, for~$\alpha_1 < \alpha_2$ we have
\begin{align}
    \text{OPT}(\alpha_1) &\leq \text{OPT}(\alpha_2),\label{eqn:opt-1}\\
    \text{OPT}(\alpha_2) &\leq \text{OPT}(\alpha_1) + (\alpha_2 - \alpha_1) \cdot \PP(\polygons)\label{eqn:opt-2}.
\end{align}
Assuming~(\ref{eqn:perimeter-lower}) holds, and using~(\ref{eqn:opt-1}), we obtain
\begin{align*}
    &g(\sollow)(\anew)\\
    =&\ \A(\sollow) + \anew \cdot \PP(\sollow)\\
    \leq&\ \A(\sollow) + \varepsilon \cdot \A(\sollow) + \amin \cdot \PP(\sollow)\\
    \leq&\ (1+\varepsilon) \cdot (\A(\sollow) + \amin \cdot \PP(\sollow))\\
    =&\ (1+\varepsilon) \cdot g(\sollow)(\amin) = (1+\varepsilon) \cdot \text{OPT}(\amin)\\
    \leq&\ (1+\varepsilon) \cdot \text{OPT}(\anew) = (1+\varepsilon) \cdot g(\solnew)(\anew).
\end{align*}
Similarly, assuming (\ref{eqn:perimeter-upper}) holds, and using~(\ref{eqn:opt-2}), we obtain
\begin{align*}
    &\ g(\solup)(\anew)\\
    = &\ \A(\solup) + \anew \cdot \PP(\solup)\\
    \leq &\ \A(\solup) + \varepsilon \cdot \A(\solup) - (\amax - \anew) \cdot (1+\varepsilon) \cdot \PP(\polygons)\\
    &\ + \amax \cdot \PP(\solup)\\
    \leq &\ (1+\varepsilon) \cdot (\A(\solup) - (\amax - \anew) \cdot \PP(\polygons)\\
    &\ + \amax \cdot \PP(\solup))\\
    = &\ (1+\varepsilon) \cdot (g(\solup)(\amax) - (\amax - \anew) \cdot \PP(\polygons))\\
    = &\ (1+\varepsilon) \cdot (\text{OPT}(\amax) - (\amax - \anew) \cdot \PP(\polygons))\\
    \leq &\ (1+\varepsilon) \cdot \text{OPT}(\anew) = (1+\varepsilon) \cdot g(\solnew)(\anew).
\end{align*}
Hence, both~$\text{dom}_L$ and~$\text{dom}_U$ become false after a finite number of steps.

To show that the algorithm computes a $(1+\varepsilon)$-approximation, consider the case that the recursion in~$[\amin,\anew]$ is skipped because~$\text{dom}_L$ is false (the case for the other recursion is symmetrical).
We show that for every~$\alpha\in[\amin,\anew]$, there is an~$\alpha'\in[0,\infty]$ such that~$g(\comblow[\alpha'])(\alpha) \leq (1+\varepsilon) \cdot g(\combnew)(\alpha)$.
Consider the case that~$\text{dom}_L$ is false because~$g(\combnew) \geq g(\comblow)(\anew)$ holds.
Because~$\combnew$ is optimal for~$\anew$, it follows that~$\comblow$ is optimal as well. Then~$\comblow$ is optimal within~$[\amin,\anew]$ by~\Cref{lem:lower-envelope:paper} and the claim follows with~$\alpha'=\alpha$.

Now, consider the other case that~$(1+\varepsilon) \cdot g(\combnew)(\anew) \geq g(\comblow[\amin])(\anew)$.
We write this as~$h(\anew) \geq 0$ with the functions~$f:=(1+\varepsilon) \cdot g(\combnew)$, $f':=g(\comblow[\amin])$ and~$h:=f-f'$.
Because~$\comblow$ is optimal at~$\amin$, we also have~$h(\amin) \geq 0$.
We show that~$h(\alpha) \geq 0$ for~$\alpha\in[\amin,\anew]$.
The slope of~$f$ is decreasing, whereas the slope of~$f'$ is constant, so the slope of~$h$ is decreasing.
If~$h(\alpha) < 0$, then the slope of~$h$ must be negative at~$\alpha$.
However, this implies that it remains negative, which contradicts~$h(\anew) \geq 0$.
Hence, the claim follows for~$\alpha'=\amin$.
\end{proof}

The runtime of the algorithm depends on the rate of convergence of the estimated crossing point towards the actual one.
It is unclear whether this is polynomial in~$1/\varepsilon$ because the second derivatives of the objective functions can become unbounded.
However, we observe in~\Cref{sec:appendix:chord_scheme_comparison} that the algorithm converges quickly in practice.

\subsection{Exploiting $\alpha$-Nestedness}
Independent of the chosen algorithm, during a call of
\(
\Recurse(\comblow, \amin, \combup, \amax),
\)
we exploit the fact that we already know the optimal solutions $\comblow$ and $\combup$ at the boundaries $\amin$ and $\amax$.
Let~$\sollow:=\comblow[\amin]$, $\solnew:=\combnew[\anew]$ and~$\solup:=\combup[\amax]$.
Due to the~$\alpha$-nestedness of optimal solutions, we have~$\sollow\subseteq\solnew\subseteq\solup$.
For the subdivision-restricted version of the problem, it is known that this can be exploited by contracting the problem instance that is solved for~$\anew$~\cite{beines2024}.
This optimization does not carry over directly to~$\freeProblem$, but we can nevertheless exploit~$\alpha$-nestedness in three ways:
\begin{enumerate}
    \item Since~$\sollow \subseteq \solnew$ and $\optimizer$ can handle circular polygons as inputs, we can run the optimization at~$\anew$ with~$\sollow$ as the input instance instead of $\polygons$. Note that $\sollow$ is an $\amin$-circular polygon with $\amin \leq \anew$. Thus, by \cref{obs:no_arc_in_interior_of_arc}, all new arcs still connect only to the edges of $\polygons$, and consequently, the combinatorial solutions are still defined with respect to the original input $\polygons$.
    Especially for larger~$\alpha$ values, this can significantly reduce the complexity of the input.
    \item Let $\polygons(S)$ denote the polygons contained in a region $S$.
    It follows from~$\solnew \subseteq \solup$ that for every region~$S \in \solnew$, there exists a region $S' \in \solup$ with $\polygons(S) \subseteq \polygons(S')$.
    Therefore, $\solnew$ can be computed by solving the subproblem with input~$\polygons(S')$ for each region~$S' \in \solup$ independently. This allows us to discard arcs between different regions of~$\solup$, leading to smaller subdivisions.
    \item Consider an arbitrary but fixed region $S_L \in \sollow$, and let $S_U \in \solup$ be the region such that $S_L \subseteq S_U$. If $S_L$ and $S_U$ contain the same polygons, then we can find the corresponding region~$S_N \in \solnew$ by solving a simplified sub-instance for every free arc~$c=\sarc{}{xy}$ in~$S_U$ (see~\Cref{fig:meshes}).
    Let~$s=\langle x=v_0,v_1,\dots,v_k=y\rangle$ be the sequence of combinatorial vertices between~$u$ and~$v$ on the boundary of~$S_L$.
    The boundary segment between~$x$ and~$y$ in~$S_N$ must lie fully within the region enclosed by~$c$ and~$s$, and the only vertices within this region are those in~$s$. 
    Thus, it suffices to consider arcs between vertices in~$s$ for this sub-instance.
    In the special case~$s=\langle x,y \rangle$, the sub-instance can be solved directly without invoking~$\optimizer$ because the only possible arc is~$c$.
    Especially once~$\amin$ and $\amax$ are close together, this optimization substantially reduces the amount of work.
\end{enumerate}
\begin{figure}
    \centering 
    \includegraphics[width = 0.45\textwidth]{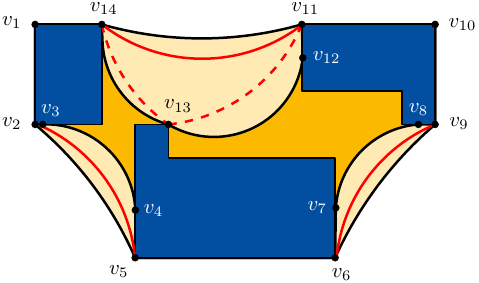}
    \caption{An example of optimization 3 for an input polygon shown in blue, $\sollow$ in dark yellow and $\solup$ in light yellow. By exploiting $\alpha$-nestedness, we can isolate the vertex sequences $\langle v_2,v_3,v_4,v_5 \rangle$, $\langle v_6,v_7,v_8,v_9 \rangle$, and $\langle v_{11},v_{12},\dots,v_{14} \rangle$ of~$\sollow$ as sub-instances, where arcs in $\solnew$ may only connect vertices from the respective sequence. In the example, this makes the first two sub-instances trivial to solve, as there is only one possible arc of radius $\anew$ each, marked in red. For the third sub-instance, three possible arcs exist.}
    \label{fig:meshes}
\end{figure}

\section{Data and Additional Experimental Results}
\label{sec:data_and_adittional_experiments}
In~\Cref{tab:dataset_overview}, we provide an overview of the datasets used in our experimental evaluation.
We report characteristics of the datasets that have an influence on the performance of the exact algorithm for the unrestricted problem.
\subsection{Additional Experiments for the Fixed Parameter Problem}\label{sec:appendix:Fixed_parameter_experiments}

We consider the two instance Andernach and Euskirchen in more detail. The runtime peak for Andernach is at $\alpha=260$ and $\alpha=500$ for Euskirchen. Using the detailed measurements for the two medium-size instances given in~\Cref{tab:single_parameter_runs:paper}, we now discuss the runtime fluctuation for different values $\alpha$. At $\alpha=100$, few arcs exist and the preprocessing performs few successful merges, so the preprocessed instance and the solution contain many polygons. At $\alpha=3000$, over $90\%$ of merges succeed, and the number of regions after preprocessing is close to the final solution. Additionally, for sub-instances, on average, each \optimizer~call involves only $4$--$5$ cells, making most of them trivial to solve and the final global subdivision is also close to trivial. In the range $\alpha\in[200,500]$, solutions are the most volatile: some polygons already merge into large representatives, but still many merges can fail, resulting in larger subdivisions during the preprocessing and in particular for the final $\optimizer$ call after the preprocessing. Hence, different factors can lead to the peak in runtime: For Andernach, it is due to many generated free arcs and failed merges; for Euskirchen, it is mainly due to the size of the final arrangement, though free arcs and average subdivision size also contribute.

Overall, we observe that the preprocessing drastically reduces the number of generated arcs that need to be tested for intersection with $\polygons$ (for $\alpha=3000$ this reduction is by more then one order of magnitude). In addition, for larger $\alpha$, the number of regions after preprocessing is already close to that in the final solution and the preprocessing solution is already very close to optimal.

\begin{table}[tbh]
\caption{
    Performance for $\alpha\in\{10,3000\}$ and the $\alpha$ value corresponding to the maximum runtime.
    We measure the number of useful arcs before ($|\hat{{F}}_\alpha(\polygons)|$) and after ($|{{F}_\alpha(\polygons)}|$) testing for intersection with polygons, both with and without preprocessing.
    In the preprocessing case, the values are accumulated over all sub-problems as well as the final one.
    Merges (Tested) is the number of calls to \texttt{UPA-Opt} during the preprocessing and Succ is the percentage of successful merges.
    The average number of cells in the subdivision across all \texttt{UPA-Opt} calls in the preprocessing is listed under Cells (PP) and the number of cells in the final \texttt{UPA-Opt} call is given in Cells (Final).
    Finally, we list the number of regions after the preprocessing (PP) and in the solution (Sol).
}
\label{tab:single_parameter_runs:paper}
\begin{tabular*}{\textwidth}{@{\,}l
@{\extracolsep{\fill}}r@{\extracolsep{\fill}}r
@{\extracolsep{\fill}}r@{\extracolsep{\fill}}r
@{\extracolsep{\fill}}r@{\extracolsep{\fill}}r
@{\extracolsep{\fill}}r@{\extracolsep{\fill}}r
@{\extracolsep{\fill}}r@{\extracolsep{\fill}}r
@{\extracolsep{\fill}}r@{\extracolsep{\fill}}r@{\,}}
\toprule
& & 
& \multicolumn{2}{c}{$|\hat{F}_\alpha(\polygons)|$}
& \multicolumn{2}{c}{$|{F}_\alpha(\polygons)|$}
& \multicolumn{2}{c}{Merges}
& \multicolumn{2}{c}{Cells}
& \multicolumn{2}{c}{Regions} \\
\cmidrule{4-5} \cmidrule{6-7} \cmidrule{8-9} \cmidrule{10-11} \cmidrule{12-13}

Dataset & $|V(\mathcal{B})|$ & $\alpha$ 
& no PP & PP 
& no PP & PP 
& Tested & Succ 
& PP & Final 
& PP & Sol \\
\midrule
Andernach & 21\,036 & 100 & 27\,220 & 27\,892 & 14\,185 & 14\,813 & 2\,591 & 49\,\% & 4.2 & 8\,616 & 2\,180 & 1\,879 \\
Andernach & 21\,036 & 260 & 468\,642 & 121\,414 & 82\,389 & 29\,845 & 15\,509 & 64\,\% & 5.0 & 5\,272 & 218 & 147 \\
Andernach & 21\,036 & 3\,000 & 4\;651\;862 & 217\,147 &  69\;163 & 12\,057 & 3\,778 & 90\,\% & 4.5 & 2705 & 50 & 11\\[5pt]

Euskirchen & 43\,146 & 100 & 68\,201 & 65\,885 & 28\,328 & 27\,700 & 3\,948 & 56\,\% & 4.3 & 12\,367 & 2\,651 & 1\,941 \\
Euskirchen & 43\,146 & 500 & 1\,025\,116 & 233\,526 & 64\,510 & 24\,773 & 5\,748 & 80\,\% & 4.7 & 47\,780 & 259 & 123 \\
Euskirchen & 43\,146 & 3\,000 & 8\,018\, 632 & 386\,706 & 92\,171 & 22\,794 & 5\,403 & 89\,\% & 4.5 & 368 & 66 & 20 \\
\bottomrule
\end{tabular*}
\end{table}

\subsection{Comparison Binary Bisection and Approximate Intersection Points}
\label{sec:appendix:chord_scheme_comparison}
In \Cref{fig:speedup_intersection_approx}, we plot the speedup of the dichotomic algorithm using approximate intersection points compared to the binary bisection–based approximation scheme, for all medium-sized instances and $\varepsilon \in \{10^{-2}, 10^{-4}, 10^{-6}\}$. Similarly, \Cref{fig:ratio_intersection_approx} shows the ratio of recursive calls.

Across all instances and independently of $\varepsilon$, the intersection-based scheme consistently achieves a speedup greater than one, with a maximum of $2.6$. The speedup decreases for smaller $\varepsilon$, as the $\alpha$-range must be explored more thoroughly. Consequently, the intervals of $\alpha$ for which $\optimizer$ is called in both approaches become smaller.
Thus, the “unnecessary” calls performed by the binary bisection scheme can be executed very quickly due to our engineering optimizations.

\begin{figure}[bt]
    \centering
    \begin{minipage}[t]{0.48\textwidth}
        \centering
        \includegraphics[width=\textwidth]{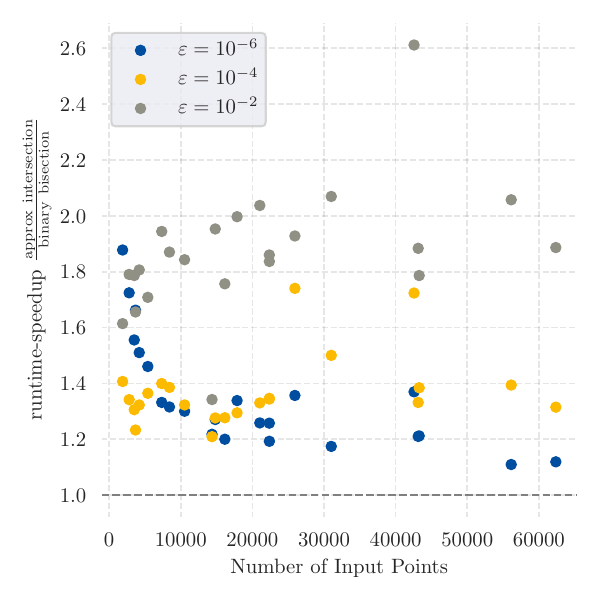}
        \caption{The runtime speedup when choosing $\alpha_N$ as the approximate intersection point instead of the binary bisection.}
        \label{fig:speedup_intersection_approx}
    \end{minipage}\hfill
    \begin{minipage}[t]{0.48\textwidth}
        \centering
        \includegraphics[width=\textwidth]{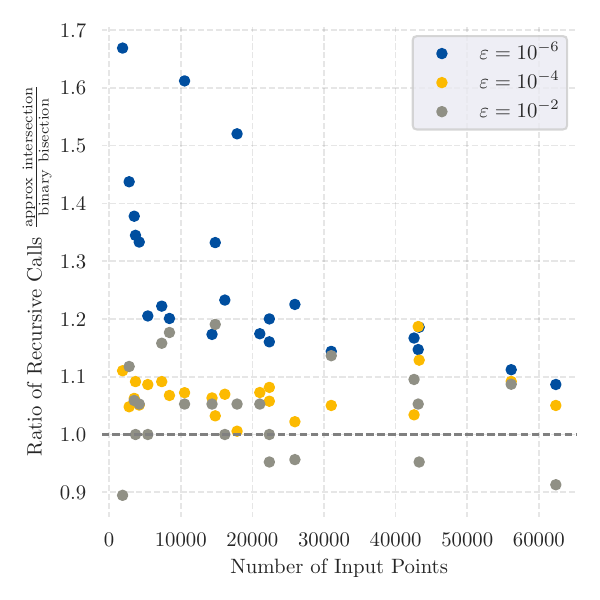}
        \captionsetup{textformat=simple} \caption{Ratio of recursive calls comparing the bisection approach using approximate intersection points to  binary bisection.}
        \label{fig:ratio_intersection_approx}
    \end{minipage}
\end{figure}%

\begin{table}[htp]
\caption{An overview of the different datasets used for our experiments. ``Diameter'' refers to the diameter of the vertex set. ``Density'' is the proportion of the convex hull of the vertex set that is covered by input polygons. Column $\Phi$ is the segment spread, i.e., the diameter divided by the minimum distance from any polygon edge to a non-adjacent vertex.}
\label{tab:dataset_overview}
\centering
\begin{tabular*}{\textwidth}{l@{\extracolsep{\fill}}r@{\extracolsep{\fill}}r@{\extracolsep{\fill}}r@{\extracolsep{\fill}}r@{\extracolsep{\fill}}r@{\extracolsep{\fill}}r}
\toprule
Instance & Vertices & Polygons & Diameter & Density & $\Phi$ & $\log_2(\Phi)$ \\
\midrule
bockelskamp & 1\,883 & 360 & 1\,324 & 10.7\,\% & 36\,819 & 15.17 \\
bokeloh & 2\,793 & 579 & 1\,612 & 11.7\,\% & 23\,900 & 14.54 \\
erlenbach & 3\,505 & 650 & 1\,789 & 12.0\,\% & 17\,892 & 14.13 \\
ahrem & 3\,683 & 334 & 1\,754 & 7.3\,\% & 60\,604 & 15.89 \\
goddula & 4\,202 & 731 & 2\,214 & 7.7\,\% & 10\,025 & 13.29 \\
friesheim & 5\,402 & 862 & 2\,090 & 9.6\,\% & 99\,379 & 16.60 \\
gerolstein & 7\,344 & 1\,312 & 3\,291 & 8.0\,\% & 33\,149 & 15.02 \\
belm & 8\,420 & 1\,722 & 3\,589 & 11.4\,\% & 35\,896 & 15.13 \\
edendorf & 10\,536 & 1\,464 & 2\,605 & 9.0\,\% & 26\,794 & 14.71 \\
gruppe8 & 14\,365 & 2\,437 & 10\,554 & 0.8\,\% & 373\,983 & 18.51 \\
gruppe2 & 14\,814 & 2\,192 & 4\,941 & 3.6\,\% & 134\,884 & 17.04 \\
forsbach & 16\,156 & 1\,801 & 2\,439 & 12.4\,\% & 42\,529 & 15.38 \\
gruppe3 & 17\,864 & 2\,866 & 3\,560 & 6.9\,\% & 668\,344 & 19.35 \\
andernach & 21\,036 & 3\,450 & 5\,476 & 11.5\,\% & 72\,523 & 16.15 \\
jena & 22\,365 & 3\,970 & 6\,643 & 9.9\,\% & 288\,002 & 18.14 \\
bad\texttt{\textunderscore}neuenahr & 22\,371 & 3\,989 & 6\,533 & 9.8\,\% & 1\,270\,417 & 20.28 \\
gruppe4 & 25\,933 & 4\,305 & 9\,140 & 1.9\,\% & 2\,618\,358 & 21.32 \\
gruppe6 & 31\,011 & 4\,730 & 7\,936 & 2.7\,\% & 2\,741\,171 & 21.39 \\
weimar & 42\,572 & 6\,426 & 5\,797 & 10.7\,\% & 638\,680 & 19.28 \\
euskirchen & 43\,146 & 4\,865 & 4\,248 & 14.8\,\% & 68\,756 & 16.07 \\
bergedorf & 43\,283 & 4\,865 & 6\,384 & 11.1\,\% & 96\,323 & 16.56 \\
gruppe10 & 56\,132 & 6\,752 & 9\,704 & 3.1\,\% & 1\,485\,154 & 20.50 \\
celle & 62\,364 & 8\,411 & 6\,155 & 10.0\,\% & 974\,941 & 19.89 \\
ludwigshafen & 81\,262 & 11\,341 & 11\,094 & 13.9\,\% & 1\,955\,340 & 20.90 \\
mainz & 116\,944 & 17\,551 & 9\,821 & 13.9\,\% & 16\,997\,667 & 24.02 \\
koblenz & 180\,073 & 20\,679 & 11\,404 & 9.6\,\% & 1\,483\,791 & 20.50 \\
aachen & 279\,046 & 28\,460 & 15\,785 & 9.2\,\% & 2\,766\,948 & 21.40 \\
bonn & 388\,237 & 38\,258 & 16\,814 & 8.0\,\% & 19\,453\,005 & 24.21 \\
\bottomrule
\end{tabular*}
\end{table}

\subsection{Comparison to the Subdivision-Restricted Variant}\label{appendix:subdivisioncomparison}

\begin{figure}[tbp]
    \centering 
    \includegraphics[width = 1\textwidth]{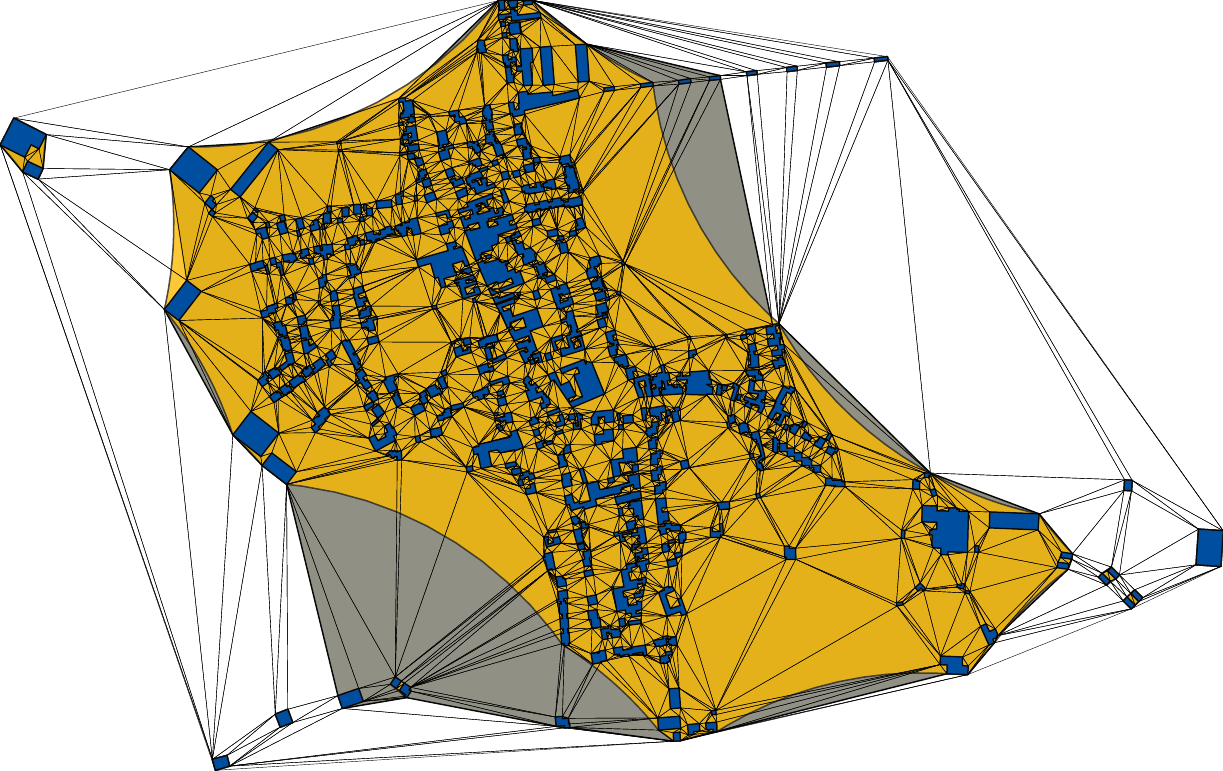}
    \caption{An example in which the triangulation-based aggregation approach by Rottmann et al.~\cite{rottmann2024bicritshapes} produces a significantly different solution than the optimal one. The dataset is Ahrem with~$\alpha = 5000$ for both solutions.}\label{fig:RottmannVsFree}
\end{figure}

Figure~\ref{fig:RottmannVsFree} displays a situation in which the unrestricted and subdivision-restricted approaches produce substantially different solutions for the same value of~$\alpha$.
The chosen subdivision is a constrained Delaunay triangulation (as suggested in \cite{rottmann2024bicritshapes}).
As can be seen, the outer limits of settlement areas tend to be sparsely developed, so a decision has to be made whether to include outlier buildings in the main component or not.
The subdivision has to incorporate these buildings somehow, which leads to long edges that span across large empty areas.
These edges prescribe certain directions in the solution boundary and preclude others.
In this case, the boundary arc in the optimal (unconstrained) solution is perpendicular to many edges of the Delaunay triangulation. Hence, the subdivision-based only has the choice between staying closer to the buildings near the boundary (which increases the perimeter) or including the outlier buildings (which vastly increases the area).
For aggregations that place a high emphasis on the perimeter, the better choice is to include the outlier buildings, which leads to a drastically different aggregation.
\end{document}